\documentclass{amsart}
\usepackage[foot]{amsaddr}
\usepackage{amssymb}
\usepackage[dvips]{epsfig}
\usepackage{booktabs}
\usepackage{multirow}
\usepackage{bbm}
\usepackage{dsfont}
\usepackage{mathrsfs}
\usepackage{xltabular}
\usepackage{tabularx}
\usepackage{makecell} 
\usepackage{ragged2e} 
\usepackage{graphicx}
\usepackage{color}
\usepackage{amsmath}
\allowdisplaybreaks
\usepackage{amssymb}
\usepackage{amsfonts}
\usepackage{geometry}
\usepackage{xstring}
\usepackage{amsthm}
\usepackage[normalem]{ulem}    
\usepackage{bm}
\usepackage{cite}
\usepackage{tikz}
\usetikzlibrary{arrows.meta}

\makeatletter
\@namedef{subjclassname@2020}{\textup{2020} Mathematics Subject 
    Classification}
\makeatother

\providecommand{\vol}{\operatorname{vol}}

\providecommand{\dist}{\operatorname{dist}}

\newcommand{\Z}{\mathbb{Z}}
\newcommand{\C}{\mathbb{C}}
\newcommand{\N}{\mathbb{N}}
\newcommand{\R}{\mathbb{R}}

\newcommand{\E}{\mathbb{E}}

\newcommand{\spec}{{\rm Spec}}

\renewcommand{\tt}{\tilde{t}}

\renewcommand{\P}{\mathbb{P}}

\newcommand{\mcA}{\mathcal{A}}

\newcommand{\mcD}{\mathcal{D}}

\newcommand{\mcE}{\mathcal{E}}
\newcommand{\mcF}{\mathcal{F}}

\newcommand{\mcK}{\mathcal{K}}

\newcommand{\mcL}{\mathcal{L}}
\newcommand{\mcM}{\mathcal{M}}

\theoremstyle{plain}
\newtheorem{theorem}{Theorem}[section]
\newtheorem{lemma}[theorem]{Lemma}
\newtheorem{proposition}[theorem]{Proposition}

\theoremstyle{definition}
\newtheorem{construction}[theorem]{Construction}
\newtheorem{definition}[theorem]{Definition}

\theoremstyle{remark}
\newtheorem{remark}[theorem]{Remark}

\theoremstyle{plain}

\numberwithin{equation}{section}

\usepackage[colorlinks=true,pdfstartview=FitV,linkcolor=magenta,citecolor=cyan]{hyperref}

\usepackage{caption}

\begin{document}
\title[Anderson-Bernoulli General Dimension]{Localization near the edge for the lattice Anderson-Bernoulli model in general dimension}

\author{Linjun Li}
\address{Department of Mathematics, University of Pennsylvania, PA, USA}
\email{linjun@sas.upenn.edu}

\author{Shihe Liu}
\address{School of Mathematical Sciences, Peking University, Beijing, China}
\email{2301110021@stu.pku.edu.cn}

\author{Lingfu Zhang}
\address{Division of Physics, Mathematics and Astronomy, California Institute of Technology, Pasadena, CA, USA}
\email{lingfuz@caltech.edu}

\dedicatory{Dedicated to the memory of Professor Jean Bourgain}

\maketitle

\begin{abstract}
The Anderson tight-binding model is a fundamental model of quantum transport and localization in disordered media. Solving a problem left open by Bourgain and Kenig, this paper proves Anderson localization near the bottom of the spectrum for the lattice Anderson model with Bernoulli potential, in any dimension $d\ge 2$. The proof uses the multiscale framework of Fröhlich–Spencer and Bourgain–Kenig, and the main new ingredient is a probabilistic discrete unique continuation principle (PDUC) for the discrete Schr\"odinger equation. This PDUC is established via a bootstrap argument and a key probabilistic lemma proved by adaptively revealing the
random potential.
\end{abstract}

\section{Introduction}

\subsection{Main result and background}
The \emph{Anderson tight-binding model} is a fundamental quantum mechanical framework used to describe wave propagation and electronic conduction in disordered media, such as a metal with impurities \cite{And58}.
We consider the Bernoulli potential setting, i.e., the \emph{Anderson-Bernoulli model (ABM)}, which is the random Schr\"odinger operator on $\ell^2(\Z^d)$,
\begin{equation}\label{eq:ABM-schrodinger-operator}
    H=-\Delta+\lambda V, \qquad \lambda>0,
\end{equation}
where $\Delta$ is the discrete Laplacian,
\[
(\Delta u)(x)
    =\sum_{y:\,|y-x|=1}\bigl(u(y)-u(x)\bigr),
    \qquad x\in\Z^d \footnote{Here and throughout this paper, $|\cdot|$ denotes the Euclidean distance.},
\]
and $V$ is the multiplication operator $(Vu)(x)=V_xu(x)$.
The random variables $\{V_x\}_{x\in\Z^d}$ are independent and
identically distributed, with
\[
\mathbb{P}[V_x=0]=\mathbb{P}[V_x=1]=\frac12,
    \qquad \forall\, x\in\Z^d.
\]
The parameter $\lambda$ is typically referred to as the \emph{disorder strength}.

A central phenomenon of the Anderson model is the so-called \emph{Anderson localization},
where waves are trapped by disorder.
To state it mathematically, recall that almost surely, the spectrum of $H$ is the deterministic set
\cite[Theorem~3.9]{Kir08}
\[
\spec(H)=[0,4d]\cup[\lambda,\lambda+4d].
\]
Anderson localization can be described through both spectral
properties and the long-time behavior of wave packets.
For an energy interval $I\subset \spec(H)$, \emph{spectral localization (SL)} means that $H$ has pure point spectrum in $I$, with exponentially decaying eigenfunctions.
\emph{Dynamical localization (DL)} states the absence of spatial
spreading, through uniform-in-time bounds on spatial moments of
wave packets restricted to $I$.
Such dynamical bounds imply pure point spectrum in $I$ by the RAGE theorem; see \cite[Section~5.4]{CS87} and \cite[Theorem~7.14]{Kir08}.

Our main result establishes SL and a form of \emph{strong dynamical localization (SDL)} near the bottom of the spectrum, with moment bounds in expectation as in~\cite{germinet2012comprehensive,RZ}.
To state it, let $P_H(I)=\mathds{1}[I](H)$ denote the spectral projection onto $I$, let $\delta_0$ be the standard basis vector at the origin, and let $\langle X\rangle^q$ denote multiplication by $(1+|x|^2)^{q/2}$ for any $q>0$.

\begin{theorem}\label{thm:localization-near-the-edge-band}
For every $d\geq4$, there exists $p_0=p_0(d)>0$ such that, for every $\lambda>0$, there exists $E_*=E_*(d,\lambda)\in(0,4d)$ such that the following holds.
Almost surely, the operator $H$ in \eqref{eq:ABM-schrodinger-operator} has pure point spectrum in $[0,E_*]$, and every eigenfunction with eigenvalue in $[0,E_*]$ decays exponentially in space.
Moreover, for every $b>0$ and $0<s<p_0/(b+1/2)$,
    \[
        \E\!\left[
            \sup_{t\in\mathbb{R}}
            \bigl\|\langle X\rangle^{bd}
            \exp(-itH)P_H([0,E_*])\delta_0
            \bigr\|_{\ell^2(\Z^d)}^s
        \right]<\infty.
    \]
\end{theorem}
Several remarks are in order.
First, prior to this work, Anderson localization for the ABM has been established for $d=1$ in \cite{CKM87}, $d=2$ in \cite{DS20}, and $d=3$ in \cite{LZ22}. 
Our proof applies in every dimension $d\ge2$.

Second, the parameter $p_0$ in the dynamical bound is governed by the probability exponent in the multiscale analysis. More precisely, according to \cite[Definition~1 and Theorem~1]{RZ}, $p_0$ here can be taken as $\kappa_*/d$ for $\kappa_*$ from Theorem~\ref{thm:msa-short-large} below.

Third, our result also extends to $V$ being i.i.d.~ $\mathrm{Bernoulli}(p)$ for every $p\in(0,1)$.
Together with the Bernoulli decomposition result in \cite{AGKW09,germinet2012comprehensive}, this yields localization near the bottom of the spectrum for i.i.d.~potentials with any nondegenerate, compactly supported single-site distribution. 

We next review some previous works on the Anderson model, in particular progress towards localization of the ABM.

\smallskip
\noindent\textbf{Phase diagram, early works, and regularity assumptions.}
We first review the Anderson model with a general i.i.d.~potential.
It is expected that localization holds throughout the whole spectrum in dimensions $d=1,2$.
For $d\ge3$, weak disorder is expected to produce delocalization in the spectral bulk and localization near the spectral edges, with transition energies called \emph{mobility edges}, whereas strong enough disorder (with regular enough potential laws) is expected to still localize the entire spectrum.
Related questions of localization, delocalization, and the phase transition have also motivated extensive work on related models, such as the Anderson model on other graphs (trees, sparse random graphs, and strips), random band matrices, and classical random trajectory models such as Manhattan pinball and Lorentz mirrors~\cite{Kle98,AL25,Bou13Strip,BGV15,HMY24,LL26,Sch09,CS22,Gol22,CPSS24,Dro25,shcherbina2017characteristic,Shc20,shcherbina2019universality,bourgade2020random,YY25,ER25,XYYY24,DKYYY25,DFYYY25,Li21,Li26Mirror,EGH25}.
For the Anderson model on the lattice, however, rigorous results have mainly established the localized part of this picture.

In 1D, the localization picture is essentially complete.
Following the early continuum and lattice results of \cite{goldsheid1977pure,kunz1980spectre}, Carmona, Klein, and Martinelli~\cite{CKM87} established spectral localization throughout the spectrum for a nondegenerate i.i.d.~potential $V$ with a finite positive moment, including Bernoulli laws.
A key feature of the 1D model exploited is its reduction to products of transfer matrices.
For some more recent alternative proofs, see~\cite{bucaj2019localization,JZ19}.
We remark that localization has also been shown for some related 1D models, in e.g., \cite{damanik2002localization} for a continuum space operator with a Bernoulli distribution, modeling alloys, and \cite{LSZ26b} for long-range models with rational hopping symbols.

Results are more limited in $d\ge 2$.
A central approach is \emph{multiscale analysis (MSA)}, introduced by Fr{\"o}hlich and Spencer \cite{FS83}.
Then with Martinelli and Scoppola, they established that localization holds at large disorder and near spectral edges for i.i.d.~random potentials with bounded probability densities~\cite{frohlich1985constructive}.
Carmona, Klein, and Martinelli~\cite{CKM87} extended these results to laws with H{\"o}lder continuous distribution functions.
Another approach is the fractional moment method of Aizenman and Molchanov~\cite{AM93}, which also requires regularity assumptions on the single-site distribution.

Removing these regularity assumptions is motivated both by physics and by the expected universality of Anderson localization.
In models of alloys, the potential naturally takes finitely many values corresponding to the constituent atomic species, and in particular a binary alloy leads to a Bernoulli potential.
More broadly, physicists regard Anderson localization as a universal phenomenon: the existence of a localized regime is expected to be independent of the precise form of the single-site distribution.
The ABM therefore provides a natural test of this universality.

\smallskip
\noindent\textbf{Framework for ABM in $d\ge 2$.}
Toward localization for the ABM, Bourgain~\cite{Bou04} developed a multiscale framework based on the MSA approach from the 1980s.
Its key innovations included retaining free sites throughout the iteration and using Sperner's theorem to control eigenvalue concentration.
In their breakthrough work, Bourgain and Kenig~\cite{BK05} developed this framework further and proved localization near the bottom of the spectrum for the continuum ABM on $\R^d$, $d\ge2$.
Their argument exploits the continuum setting through a quantitative unique continuation principle, which supplies lower bounds on eigenfunction mass and hence on eigenvalue variation at the free sites; see also~\cite{bourgain2005anderson} for further continuum results.
Germinet and Klein~\cite{germinet2012comprehensive} subsequently extended this framework to establish strong forms of localization near the bottom of the spectrum for continuum Anderson models with arbitrary nondegenerate, compactly supported single-site distributions.

\smallskip
\noindent\textbf{Discrete unique continuation.}
It was first revealed by \cite{BK05} that the localization problem has a deep connection with the unique continuation principle: the localization near the edge of the spectrum for the continuum ABM on $\R^d$ is established with the help of a quantitative unique continuation principle on $\R^d$ (\cite[Lemma 3.10]{BK05}).
To apply the Bourgain-Kenig framework to the standard lattice ABM, a key ingredient needed is a suitable \emph{discrete unique continuation principle (DUC)}.
However, the direct continuum analogue fails on the lattice, as illustrated by Jitomirskaya's counterexample \cite[Theorem~2]{Jit07}; see also \cite{Li26support} for constructions of solutions with lower-dimensional support.
Ding and Smart \cite{DS20} made a breakthrough by showing that a weaker estimate, controlling the number of sites where an eigenfunction is not too small, suffices when combined with a generalized Sperner theorem.
Inspired by the polynomial methods of Buhovsky, Logunov, Malinnikova, and Sodin \cite{BLMS} for discrete harmonic functions, Ding and Smart proved a probabilistic version of such a DUC on $\Z^2$, and used it in the MSA to get localization near the spectral edge for the $d=2$ ABM.
Li \cite{Li22} subsequently proved localization at large disorder for the symmetric Bernoulli model on $\Z^2$ outside small neighborhoods of finitely many exceptional energies.

However, extending these ideas to $d\ge 3$ remained difficult.
The polynomial structure used in \cite{BLMS} is specific to the planar lattice, while the random propagation argument of Ding and Smart exploits the planar geometry of $\Z^2$.
The later arguments of Bou-Rabee, Cooperman, and Ganguly \cite{BAW25}, which extend the harmonic unique continuation result to periodic planar graphs, likewise use planarity and do not directly provide a higher-dimensional substitute.
The next progress was due to Li and Zhang \cite{LZ22}, who proved spectral-edge localization on $\Z^3$ using a deterministic DUC for arbitrary bounded potentials.
Their proof uses geometric objects such as tetrahedrons in $\Z^3$ to reduce several estimates to 2D triangular lattices, where polynomial structure becomes available again.
Hurtado \cite{Hur26a,Hur26b} further extended localization near the bottom of the spectrum in 2D and 3D to independent, but not necessarily identically distributed, uniformly bounded potentials with uniformly positive variances. Apart from the works above, we mention the work of Imbrie \cite{Imb21}, which proved the localization for Anderson models on $\Z^d$ with a discrete distribution taking $N$ values (for $N$ sufficiently large) at
large disorder.

Recent work has made further progress on DUC in higher dimensions.
It is natural to start with a \textit{non-quantitative DUC}, i.e., to lower bound the dimension of the support of a nontrivial solution to the discrete Schr\"odinger equation.
Krymskii \cite{Kry24} obtained such a non-quantitative DUC, growing logarithmically with the dimension $d$.
Li \cite{Li26support} proved a dimension-reduction principle, yielding a nearly sharp non-quantitative DUC in 4D.
Li \cite{Li26} then worked on the simplex (in general dimension) and established a sharp non-quantitative DUC for functions satisfying oriented-simplex relations, using a Pascal uncertainty principle.
Extending Li's approach quantitatively, Liu, Shi, and Zhang \cite{LSZ26quc} obtained DUC estimates on both simplices in general dimension and $\Z^d$.
Another result \cite{LSZ26} of Liu, Shi, and Zhang pursued a probabilistic approach, proving lower bounds of full dimensional order on the number of sites where a solution is not too small, but with prescribed boundary data.

\subsection{Proof strategy}\label{ssec:iop}
Our proof follows the multiscale framework developed in
\cite{FS83,frohlich1985constructive,Bou04,BK05,DS20,LZ22}.
As already alluded to, the main new input needed is a version of DUC, as follows.
Denote $Q_R=[-R,R]^d\cap\Z^d$ for any $R>0$. The necessary DUC coming from the generalized Sperner argument of \cite{DS20} is an estimate of the form
\begin{equation}\label{eq:DUCneed}
 \left|\left\{x\in Q_R:
 |u(x)|\ge \exp(-CR^\alpha)|u(0)|\right\}\right|
 \ge cR^\beta,
\end{equation}
for any $R>0$, $u:\Z^d\to \C$ satisfying $Hu=Eu$, where $E\in \R$ and $H=-\Delta+\lambda V$, with arbitrary $\lambda >0$ and a bounded and deterministic potential $V$.
Here the exponents $\alpha$ and $\beta$ need to satisfy
\[
 \alpha(\alpha-1)<\frac\beta d-\frac12,
 \quad \alpha\ge1.
\]
See also \cite{LSZ26quc}.
Especially, when $\beta=d$, one would need $\alpha<\frac{1+\sqrt 3}{2}$, which is exactly the threshold pointed out by the remark in \cite[Section 5]{BK05}.

The sparse solutions of \cite{Jit07,Li26support} impose $\beta\le\lceil d/2\rceil$; thus \eqref{eq:DUCneed} with the above constraint on $\alpha$ and $\beta$ cannot hold in even dimensions. In odd dimensions it leaves only a narrow range of $\alpha$ near $1$.
Indeed, this has only been proved in 3D with $\alpha=1$ and $\beta>3/2$ \cite[Theorem~1.3]{LZ22}, leading to the 3D ABM localization.
The other deterministic DUC results discussed above do not yet provide the required constraint on $\alpha$ and $\beta$:
for example, in the above mentioned work \cite{LSZ26quc}, \eqref{eq:DUCneed} is proved for  $\alpha=2$ and $\beta=\lceil d/2\rceil$ (up to a logarithmic factor).

In this paper, we instead prove a \emph{probabilistic discrete unique continuation principle (PDUC)} with $\alpha=1$ and $\beta=d-o(1)$.
More precisely, take $H$ from \eqref{eq:ABM-schrodinger-operator} with $V$ being i.i.d.~Bernoulli, and fix an energy $E_0$ in a compact interval $I$ and a growth constant $K\ge0$, and let $\gamma=(d-1)/d$.
Modulo some free sites arguments, Theorem~\ref{thm:duc-linear} below essentially says that there exists an event of probability at least
\[
 1-\exp(-c( R/\log R)^\gamma)
\]
on which every nonzero local solution $u:Q_R\to\C$ of $Hu=Eu$ on $Q_{R-1}$ satisfying
\[
 E\in I,\qquad |E-E_0|\le \exp(-CR),\qquad
 \|u\|_{\ell^\infty(Q_R)}\le \exp(KR)|u(0)|
\]
has
\[
 \left|\left\{x\in Q_R:
 |u(x)|\ge \exp(-CR)|u(0)|\right\}\right|
 \ge \frac{cR^d}{(\log R)^d}=R^{d-o(1)}.
\]
The constants $C,c>0$ depend on the fixed parameters.
This provides every fixed $\beta<d$ in \eqref{eq:DUCneed} for large $R$.
The essential point is that the estimate holds for all these solutions simultaneously; for comparison, see \cite{LSZ26} for PDUC with a certain fixed boundary condition.
We also remark that, with this PDUC, localization for the hierarchical ABM established in \cite[Theorem 1.3]{LSZ26} can also be strengthened, by removing the condition of $\alpha>N$ therein.    

We next explain the ideas in the proof of this PDUC.

\noindent\textbf{Bootstrap argument.}
The main strategy here is a bootstrap on the number of not too small sites.
The cone property (see e.g., \cite[Section 2]{LZ22} or Proposition \ref{prop:cone-property} below) supplies an initial set of size $R^{1-o(1)}$.
Normalize $u(0)=1$ and suppose that, in an inner cube $Q_r$, we have a set
\[
 S=\{x\in Q_r:|u(x)|\ge \exp(-C_1R)\},
\]
for some $C_1>0$.
Take $r$ to be a small fixed multiple of $R$, leaving room in $Q_R\setminus Q_r$ for the construction below.
The goal is to produce many sites where $|u|$ is above a smaller threshold, so that for some $C_2$ slightly larger than $C_1$,
\[
 S'=\{x\in Q_R:|u(x)|\ge \exp(-C_2R)\}
\]
has size of order $R|S|^\gamma$, for $\gamma=(d-1)/d$.
The two factors are due to the following: geometry supplies order $|S|^\gamma$ source lines, and random propagation supplies order $R$ target sites per successful line.

More precisely, for the first factor, Lemma~\ref{lem:many-lattice-lines}, based on the discrete Loomis-Whitney inequality, gives a direction $\nu$ in
\[
 \mathcal V=\{e_1+e_2,e_1-e_2,e_1+e_3,\ldots,e_1+e_d\}
\]
such that at least $(|S|/2)^\gamma$ parallel lines in this direction meet $S$.
By symmetry, assume $\nu=e_1+e_2$, and write these lines as
\[
 L(\kappa,\zeta)=\{(\kappa+a,a,\zeta):a\in\Z\},
 \qquad \kappa\in\Z,\;\zeta\in\Z^{d-2}.
\]
Enumerate all such lines meeting $Q_r$, first by increasing $\kappa$ and then lexicographically by $\zeta$.
Their intersections with $Q_r$ are the \textit{source blocks} $\mathcal J_f$, for $1\le f\le n_{\mathrm{blk}}$, and such a block is \textit{large} when it meets $S$.
For the block on $L(\kappa_f,\zeta_f)$, put its \textit{target set} $\mathcal K_f$ on the shifted line $L(\kappa_f+1,\zeta_f)$, beyond the forward exit from $Q_r$.
We can choose order $R$ many targets per block, with all targets distinct and contained in $Q_R$.
This is the source--target geometry of Construction~\ref{construction:fc-system} below.

To view $u$ at the source blocks and target sets as functions of common boundary data, use a bowl boundary 
\[
\mathfrak B=T_s=\bigl(\{-s-1,-s\}\times\Omega_\Sigma\bigr)\cup\bigl(\{-s,\ldots,s-1\}\times\partial^+\Omega_\Sigma\bigr),
\]
where $\Omega_\Sigma=[-s,s]^{d-1}\cap\mathbb Z^{d-1}$ 
and 
\[\partial^+\Omega_\Sigma=\left\{ z\in \Z^{d-1}\setminus \Omega_\Sigma: \exists \ z'\in \Omega_\Sigma \ {\rm such\ that } \  |z-z'|_1=1  \right\}. \]
Here $s$ is taken to be a suitable fixed multiple of $R$, with $r<s<R$.
With given $u|_{\mathfrak B}$, the Schr\"odinger equation inside $Q_s$ can be solved by propagating the data in the $e_1$ direction.
Thus, when all source blocks and target sets are inside the cube $Q_s$, once the potential $V$ is given, all source and target values are linear functions of the boundary data $u|_{\mathfrak B}$, with coefficients bounded exponentially in $R$.

For the moment fix the boundary data in advance, and condition on the potential $V$ outside $Q_r$.
For the block with index $f$,
let $\{S_{f,j}\}_{j=1}^{r_f}$ and $\{Y_{f,k}\}_{k=1}^{b_f}$ denote the values of $u$ in the source block $\mathcal J_f$ and target set $\mathcal K_f$.
The ordering of the blocks ensures that source values in block $f$ depend only on $V$ in earlier blocks, while its target values depend only on $V$ in earlier blocks and block $f$ itself.
Changing $V$ at one source in $\mathcal J_f$ from $0$ to $1$ gives the exact identity
\[
 Y_{f,k}^{(1)}-Y_{f,k}^{(0)}=\pm\lambda S_{f,j}
\]
for every target of that block.
Consequently, if $|S_{f,j}|\ge\delta$, one choice of that bit makes at least half the target values have modulus at least $\lambda\delta/2$.
After exposing $V$ in the rest of the block, this choice still has probability at least $1/2$.
Processing the blocks in order therefore gives a conditional success estimate, and an exponential supermartingale controls the number of unsuccessful large blocks.
For fixed boundary data, this gives the amplification $|S|\mapsto R|S|^\gamma$.

\smallskip
\noindent\textbf{Uniform over boundary data: choice of bowl and propagation lemma.}
So far we have seen that largeness can be propagated from the source block $\mathcal J_f$ to the target set $\mathcal K_f$ with high probability, given the boundary data $u|_{\mathfrak B}$.
Such an estimate for each choice of boundary data $u|_{\mathfrak B}$ fixed before sampling $V$, as in \cite{LSZ26}, cannot in general be applied to an eigenfunction $u$ chosen after the potential $V$ is revealed.
Therefore, the main difficulty is to make this probabilistic estimate hold simultaneously over the boundary data.
We overcome this through two reductions.

First, we can pick the bowl size $s$ that minimizes $|\mathfrak B \cap S'|$. More precisely, 
taking $R-r$ to be of order $R$, we have order $R$ disjoint bowl boundaries in $Q_R\setminus Q_r$.
Then one of these bowls $\mathfrak B$ satisfies\footnote{For nonnegative $A,B$, we write $A\lesssim B$ or $B\gtrsim A$ if $A\le CB$ for a universal constant $C>0$. When the constant is allowed to depend on some parameters $\mathcal P$, we write $A\lesssim_{\mathcal P}B$ or $B\gtrsim_{\mathcal P}A$.}
\[
 p:=|\mathfrak B \cap S'|\lesssim \frac{|S'|}{R}.
\]
If such $p$ is of order at least $|S|^\gamma$, the desired lower bound on $|S'|$ would already follow.
Otherwise, retain only the data $\vartheta=u|_{\mathfrak B \cap S'}\in\C^p$ and set the other boundary values to zero.
With suitable separation of the thresholds, the reconstruction error is smaller than the target values produced by propagation.
The relevant boundary data dimension is therefore reduced to $p$.

Second, we make propagation uniform over these $p$ boundary data coordinates using a key propagation lemma, i.e., Lemma~\ref{lem:propagation} below.
To obtain a single event valid for all boundary data $\vartheta$, we basically need a union bound. A standard sufficiently fine net argument would be too expensive here.
Instead, we classify boundary data by \emph{records} produced by a cut procedure.
More precisely, fix a realization of the potential $V$ and boundary data $\vartheta$, and a set $\mathscr{K}$ of possible boundary data.
We then scan the source blocks and target sets in the given order.
Whenever a source or target value is small at $\vartheta$ but has a large supremum on the current set $\mathscr{K}$, cut the set $\mathscr{K}$ by the strip where that value is small.
We call the ordered list of sources and targets at which these cuts are made the \emph{record} of $V$ and $\vartheta$.
Each cut shrinks the volume of $\mathscr K$.
Volume comparison bounds the number of cuts by a budget $m$, which is of order $p$.
With $N$ being the size of the union of all source blocks and target sets, there are at most $(N+1)^m$ possible records.

Now for each given record, i.e., a list of sources and targets of cuts, a conditional success argument gives an exponential tail $\exp(-t)$ for its \textit{propagation loss} with tail parameter $t\ge0$, uniform over the boundary data $\vartheta$ producing the record.
Then a union bound therefore gives
\[
 \P[\text{at least one record violates its propagation loss bound by error }t]
 \le (N+1)^m \exp(-t).
\]
Every possible realization of $V$ and boundary data $\vartheta$ is covered by a record, so this estimate implies that, with probability $>1-(N+1)^m \exp(-t)$ for the potential $V$, every choice of boundary data $\vartheta$ has propagation loss bounded as desired, up to an error of $t$.
Since $N$ is polynomial in $R$, the upper bound $(N+1)^m\exp(-t)$ on the failure probability is made small whenever $t$ is a sufficiently large multiple of $p\log R$.
By taking into account this and all the other error terms, we get the desired one-step gain
\[
 |S'|\ge R^{1-o(1)}|S|^\gamma.
\]

Finally, we repeat this amplification.
If $|S|=R^\beta$, one step raises the exponent to $1+\gamma\beta-o(1)$.
Ignoring the $o(1)$ losses, the exponents follow
\[
 1\longmapsto 1+\gamma\longmapsto 1+\gamma+\gamma^2
 \longmapsto\cdots\longrightarrow\frac1{1-\gamma}=d.
\]
Repeating the argument until the remaining exponent deficit tends to zero yields a lower bound of $R^{d-o(1)}$.

Finally, the resulting PDUC supplies the eigenfunction bounds for the rank-one perturbation and generalized Sperner argument in Section~\ref{sec:WegnerMSA}.
This gives the Wegner estimate, closes the multiscale induction near the spectral bottom, and yields localization.

\subsection{Discussion on the use of AI and intermediate results}   \label{ssec:ai}
The proof was developed through sustained, highly interactive exchanges
between the authors and Large Language Models (LLMs).
Experience with the ABM and related problems led the authors to expect that combinatorial structures on the lattice, such as the pyramid in~\cite{LZ22}, would play a central role in establishing the desired DUC or PDUC.
This motivated their use of LLMs to help identify and explore such structures.
The authors guided the direction of
the investigation, while LLMs contributed to the exploration of
possible approaches and the discovery of proofs of intermediate
results. Throughout this iterative process, the authors examined and
refined the proposed arguments, developing their understanding of the
proof as the investigation progressed. The proofs are digested and written by the authors, but LLMs are used for improving exposition, reviewing the literature, and generating figures. The authors take full responsibility for the
mathematical content and correctness.

The broad outline of our research is as follows. With the assistance of LLMs, we sought to establish a suitable DUC of the form \eqref{eq:DUCneed}.
\begin{itemize}
    \item Initially, we aimed to prove a deterministic DUC, by understanding higher dimensional simplex structures, following \cite{BLMS,DS20,LZ22}. With the help of LLMs, this led to a simplex DUC in \cite{Li26}, which further gave a DUC on $\Z^d$ in \cite{LSZ26quc}.
    \item As discussed following \eqref{eq:DUCneed}, for even $d$, such a deterministic DUC cannot be sufficient for localization; for odd $d$, the DUC in \cite{LSZ26quc} needs to be improved to an exponent $\alpha$ sufficiently close to $1$, and we find this challenging with the techniques in \cite{Li26,LSZ26quc}.
    Thus we shifted our target towards a PDUC. With LLMs, we explored integrating the simplex structure with probabilistic arguments, with inputs including various geometric constructions in dimensions 3 to 5, developed by us over the years. The strongest PDUC we thus obtained is \(\alpha=\frac{1+\sqrt{3}}{2}, \beta=\frac{3}{4}d\). Here we intentionally pursued the constants $\frac{1+\sqrt{3}}{2}$ and $\frac{3}{4}$, inspired by the continuous space DUC used in \cite{BK05}. However, this PDUC is still insufficient for localization.
    \item The definitive progress came when we introduced an extra maximum condition 
    \[|u(0)|\geq \exp(-K R^2)\| u\|_{\ell^{\infty}(Q_R)},\] 
    and tried to use LLMs to additionally exploit the martingale structures from \cite{LSZ26},  instead of using the traditional simplex-based approach, towards proving a sufficient PDUC with this extra condition. The reason we artificially introduce this maximum condition is that PDUC is often applied near the resonant cubes, where this maximum condition is generally satisfied. Along this path, we used LLMs to obtain such a PDUC with $\alpha=2, \beta=d-$. The proof uses the bootstrap strategy, but starts from the insufficient PDUC obtained in the previous step. (See Appendix \ref{appendix:uniform-duc} for arguments along this line.)
    \item Finally, after synthesizing the relevant facts and developments, we realized that (1) the bootstrap can start from a much weaker bound supplied by the cone property, and (2) the final PDUC can also be improved to \(\alpha=1, \beta=d-\). These yield the PDUC and its proof presented in this paper.
\end{itemize}

\smallskip
\noindent\textbf{Organization of the remaining text.}
In Section \ref{sec:linear-martingale-system}, we will introduce a probabilistic system, which we call the linear-martingale system. This system provides an abstract description of how solutions of the Schr\"odinger equation propagate.
Then we prove the propagation lemma in the language of the linear martingale system. 
In Section \ref{sec:DUC}, we establish the key PDUC through the outlined bootstrap strategy. 
In Section \ref{sec:WegnerMSA}, we feed the PDUC into the existing framework to establish the Wegner estimate, and then the MSA, which implies localization.
In the end, we provide Appendix \ref{appendix:uniform-duc}, strengthening the PDUC for $d\ge 3$ to be uniform in the disorder strength $\lambda$ and energy $E$, using the deterministic DUC from \cite{LZ22,LSZ26quc} as input for the bootstrap arguments. We record this stronger PDUC as a result of independent interest, and it is not used in the proof of Theorem \ref{thm:localization-near-the-edge-band}.

\smallskip

\noindent\textbf{Acknowledgements.}
LZ is supported by NSF grant DMS-2505625 and a Sloan Fellowship.
SL is currently a PhD student at Peking University, where he is working on this problem under the supervision of Zhifei Zhang, with support from the grant NSFC (12288101).
LL began working on this problem during his PhD studies at the University of Pennsylvania. His University of Pennsylvania affiliation reflects the institution at which this work was initiated and a substantial part of it was carried out.
The authors are grateful to Yunfeng Shi for his sustained interest in and support of this project, as well as for many helpful suggestions.
LL and LZ are grateful to Jian Ding for introducing them to this topic and for his continued support, and to Charles Smart and Aleksandr Logunov
 for helpful discussions.
SL thanks Yunfeng Shi and Zhifei Zhang for their guidance and unwavering encouragement throughout this work.
LZ thanks Ahmed Bou-Rabee, Carlos Kenig, Eugenia Malinnikova, Hong Wang, and Ruixiang Zhang for many valuable discussions about this problem over the years.

\section{The linear-martingale system}\label{sec:linear-martingale-system}

In this section, we introduce an abstract linear structure to characterize the level-by-level propagation of solutions to Schr\"odinger equations on $\Z^d$. The readers are encouraged to refer to Construction \ref{construction:fc-system} below for the motivation of this structure. We observe that, in a certain sense, this structure bears some resemblance to the \textit{site-mixed martingale} in \cite[Appendix A]{LSZ26}.

\begin{definition}
\label{def:linear-martingale-system}
Fix an integer $n_{\mathrm{blk}}\ge1$.
For each $1\leq f\leq n_{\mathrm{blk}}$, take integers $r_f,b_f\ge1$, and 
\[
 \mathcal J_f=\{(f,j):1\le j\le r_f\},
\]
which we call the $f$-th \textit{source block}. 
Each element $(f,j)$ in $\mathcal J_f$ is called a \textit{source}.
We also denote $\mathcal J=\bigcup_{f=1}^{n_{\mathrm{blk}}}\mathcal J_f$.
We equip each source block $\mathcal J_f$ with a \textit{target set}
\[ \mathcal{K}_f=\{(f,k)_{\rm tag}:1\leq k\leq b_f\},\]
and each element $(f,k)_{\rm tag}$ in $\mathcal{K}_f$ is called a \textit{target}. We package these objects together and denote 
\[\mathcal{M}=\{(\mathcal J_f,\mathcal K_f):1\leq f\leq n_{\mathrm{blk}}\}.\]

Let $p\geq 1$ be an integer, and  $c_0>0,C_0>0$ be fixed constants. Assume $\mathscr{K}\subset \C^p$ and $(\Omega,\P)$ is a probability space. 
We say $\mathcal{M}$ is a \textit{$(c_0,C_0)$-linear-martingale system} in $(\Omega,\P)$ with \textit{input body}
$\mathscr{K}$, if the following hold:

\begin{enumerate}
\item[(1)] \textbf{\emph{Bits and words.}}
For each $(f,j)\in \mathcal J$, it is equipped with a random variable $\omega_{f,j}$ in $(\Omega,\P)$. These random variables are i.i.d. Their common distribution is 
\[
 \P[\omega_{f,j}=0]=\P[\omega_{f,j}=1]=1/2, \quad \forall  \ (f,j)\in \mathcal{J}.
\]
Equivalently, we can let $\Omega=\{0,1\}^{\mathcal J}$, and take $\mathbb P$ to be the uniform distribution on the Boolean cube $\Omega$.

The random variables $\omega_{f,j}$ are called \textit{bits}, and
$\omega_f=(\omega_{f,1},\ldots,\omega_{f,r_f})$ is the \textit{bit sequence}
of source block $\mathcal J_f$. A complete configuration $\omega\in\Omega$ is called
a \textit{word}. Write 
\[\omega_{<f}=(\omega_1,\omega_2,\cdots,\omega_{f-1}),\] 
\[\omega_{\leq f}=(\omega_1,\omega_2,\cdots,\omega_{f}),\] 
for the restrictions of the word $\omega$ to the corresponding indices, and define the filtration of $\sigma$-algebras by
\[
 \mathscr F_0=\{\varnothing,\Omega\},
 \qquad
 \mathscr F_f
 =\sigma(\omega_{\leq f}),\quad 1\leq f\leq n_{\mathrm{blk}}.
\]

\item[(2)] \textbf{\emph{Inputs, source values, target values, and fans.}}
An \textit{input} is a vector $\vartheta\in \mathscr{K}$. To each source $(f,j)$ (resp. each target $(f,k)_{\rm tag}$), we associate a complex-linear functional $S_{f,j}$ (resp. $Y_{f,k}$) of the input $\vartheta$, called its \textit{source value} (resp. \textit{target value}). 
They can be represented by
\begin{align}\label{eq:source}
 S_{f,j}(\omega_{<f};\vartheta)
 &=\sum_{i=1}^p s_{f,j,i}(\omega_{<f})\vartheta_i, \qquad 1\le j\le r_f,
\end{align}
\begin{align}\label{eq:sink}
Y_{f,k}(\omega_{\le f};\vartheta)
 &=\sum_{i=1}^p y_{f,k,i}(\omega_{\le f})\vartheta_i,\qquad 1\le k\le b_f.
\end{align}
The above two formulae show that the source value $S_{f,j}$ and its coefficients $s_{f,j,i}$ depend only on $\omega_{<f}$ and are $\mathscr F_{f-1}$-measurable. In particular, the source values $S_{1,j}$ are independent of $\omega$ and hence deterministic. The target value $Y_{f,k}$ and its coefficients $y_{f,k,i}$ depend only on $\omega_{\leq f}$ and are 
 $\mathscr F_f$-measurable. 
The ordered sequence of the target values 
\[
 \bigl(Y_{f,1}(\omega_{\le f};\cdot),\ldots,
         Y_{f,b_f}(\omega_{\le f};\cdot)\bigr)
\]
is called the \textit{fan} of source block $f$.  See Figure \ref{fig:fan-input-body} for an illustration.

\begin{figure}[htbp]
\centering
\begin{tikzpicture}[
    x=1cm, y=1cm,
    line cap=round,
    line join=round,
    target red/.style={
        draw=red!65!black,
        line width=0.9pt
    },
    target green/.style={
        draw=green!45!black,
        line width=0.9pt
    },
    target blue/.style={
        draw=blue!65!black,
        line width=0.9pt
    }
]
    \path[rotate=28, fill=black!3]
        (0,0) ellipse [x radius=2.9, y radius=1.05];

    \draw[target red]
        (-3.7,1.65) -- (3.7,-1.65)
        node[right=3pt, text=red!65!black]
        {$Y_{f,k_1}$};

    \draw[target green]
        (4.1,1.5) -- (-4.1,-1.5)
        node[left=3pt, text=green!45!black]
        {$Y_{f,k_2}$};

    \draw[target blue]
        (0,-2.8) -- (0,2.8)
        node[above=3pt, text=blue!65!black]
        {$Y_{f,k_3}$};

    \draw[rotate=28, black, line width=0.95pt]
        (0,0) ellipse [x radius=2.9, y radius=1.05];

    \fill[black] (1.4,1.02) circle[radius=1.6pt]
        node[above right=2pt] {$\vartheta$};

    \node at (2.95,0.50) {$\mathscr{K}$};
\end{tikzpicture}
\caption{An illustration of a real two-dimensional input body
$\mathscr{K}$, with an input $\vartheta$ and three target values from the same fan.
The colored lines illustrate the kernels of nonzero real-linear target values.}
\label{fig:fan-input-body}
\end{figure}

\item[(3)] \textbf{\emph{Propagation condition with $c_0$.}}
Fix a source $(f,j)\in \mathcal{J}$ and a configuration 
\[
 \widehat\omega_{f,j}
 =\bigl(\omega_{<f},(\omega_{f,h})_{h\ne j}\bigr).
\]
Since $Y_{f,k}$ depends only on $\omega_{\leq f}=(\widehat\omega_{f,j},\omega_{f,j})$, for $e\in\{0,1\}$, we can let
$Y_{f,k}^{(e)}(\widehat\omega_{f,j};\cdot)$ denote the target value
obtained by setting $\omega_{f,j}=e$ and keeping all the other bits in $\omega_{\leq f}$ unchanged. For every such configuration $\widehat\omega_{f,j}$ and every target $(f,k)_{\rm tag}$, there
exists a scalar $c_{f,j,k}(\widehat\omega_{f,j})$, independent of $\vartheta$,
such that
\begin{equation}
\label{eq:propagation-condition}
 Y_{f,k}^{(1)}(\widehat\omega_{f,j};\vartheta)
 -Y_{f,k}^{(0)}(\widehat\omega_{f,j};\vartheta)
 =c_{f,j,k}(\widehat\omega_{f,j})S_{f,j}(\omega_{<f};\vartheta), \qquad |c_{f,j,k}(\widehat\omega_{f,j})|\ge c_0.
\end{equation}
The first equality is an identity of complex-linear functionals:
it holds for every $\vartheta\in \mathscr{K}$.

\item[(4)] \textbf{\emph{Uniform coefficient bounds with $C_0$.}}
Writing $\mathbf s_{f,j}$ and $\mathbf y_{f,k}$ (these are indeed vectors in $\C^p$) for the respective
coefficient rows in \eqref{eq:source} and \eqref{eq:sink}, we require that their $\ell^2$-norms are deterministically bounded by $C_0$, i.e.,
\begin{equation}\label{eq:coefficient-bound-condition}
   \sup_{\omega\in\Omega}\max_{(f,j)}
 \|\mathbf s_{f,j}(\omega_{<f})\|_2\le C_0,
 \qquad
 \sup_{\omega\in\Omega}\max_{(f,k)_{\rm tag}}
 \|\mathbf y_{f,k}(\omega_{\le f})\|_2\le C_0.
\end{equation}

\end{enumerate}

\end{definition}

  We remark that in the definition above, all source blocks, their orders, the target sets, and the
coefficient maps are fixed before the word (or the randomness) is revealed.

We next give a volume estimate that will be used to bound the number of cuts, in the cut procedure in the propagation lemma.

\subsection{Volume estimate for a convex set cut by a strip}

A set $\mathscr{K}\subset\C^p$ is called \textit{circled} if $\exp(i\theta)\mathscr{K}=\mathscr{K}$ for every
$\theta\in\R$.  We also denote Lebesgue volume on $\C^p\simeq\R^{2p}$ by
$\vol_{2p}$.  The closed Euclidean ball in $\C^p$ is denoted by 
\[B_\rho=\{\vartheta\in\C^p:\|\vartheta\|_2\le\rho\}.\] 
If $\mathscr{K}$ is compact and $L:\C^p\to\C$ is complex-linear,
we define the \textit{width} of $\mathscr{K}$ under $L$ as
\[
w_{\mathscr{K}}(L)=\max_{\vartheta\in \mathscr{K}}|L(\vartheta)|.
\]

We now state a lemma to estimate the change of volume when $\mathscr{K}$ is cut by a strip. Such a cut will happen when $\mathscr{K}$ has a large width under a source value, but that value is small at the actual input, so that we can cut a fairly large portion of $\mathscr{K}$ while still retaining the actual input in the remaining part. This lemma is to estimate the loss of volume in one cut. See Figure \ref{fig:complex-central-strip} for an illustration. 

\begin{figure}[htbp]
\centering
\begin{tikzpicture}[
  scale=0.95,
  line cap=round,
  line join=round,
  >=Latex,
  every node/.style={font=\small},
  strip boundary/.style={
    red!65!black,
    line width=0.85pt
  },
  support/.style={
    blue!55!black,
    densely dashed,
    line width=0.65pt
  },
  measure/.style={
    {Latex[length=1.3mm,width=1mm]}-{Latex[length=1.3mm,width=1mm]},
    line width=0.75pt
  }
]
  \def\stripmajor{3.2}
  \def\stripminor{1.05}
  \def\striptilt{20}
  \def\striptau{0.40}

  \pgfmathsetmacro{\stripwidth}{
    sqrt((\stripmajor*sin(\striptilt))^2
        +(\stripminor*cos(\striptilt))^2)
  }

  \begin{scope}[rotate=10]

    \draw[support]
      (-4.25,\stripwidth) -- (4.25,\stripwidth);
    \draw[support]
      (-4.25,-\stripwidth) -- (4.25,-\stripwidth);

    \path[fill=black!3,rotate=\striptilt]
      (0,0) ellipse[
        x radius=\stripmajor,
        y radius=\stripminor
      ];

    \begin{scope}

      \clip[rotate=\striptilt]
        (0,0) ellipse[
          x radius=\stripmajor,
          y radius=\stripminor
        ];

      \clip
        (-4.3,-\striptau) rectangle (4.3,\striptau);

      \fill[red!9]
        (-4.3,-\striptau) rectangle (4.3,\striptau);

      \foreach \hatchindex in {-22,...,22} {
        \draw[red!45,line width=0.35pt]
          ({0.24*\hatchindex},\striptau)
          --
          ({0.24*\hatchindex+2*\striptau},-\striptau);
      }

    \end{scope}

    \draw[line width=0.8pt]
      (-4.35,0) -- (4.45,0)
      node[right] {$\ker L$};

    \draw[strip boundary]
      (-4.25,\striptau) -- (4.25,\striptau);
    \draw[strip boundary]
      (-4.25,-\striptau) -- (4.25,-\striptau);

    \draw[line width=1pt,rotate=\striptilt]
      (0,0) ellipse[
        x radius=\stripmajor,
        y radius=\stripminor
      ];

    \draw[measure,blue!55!black]
      (-3.90,0) -- (-3.90,\stripwidth)
      node[
        midway,
        left=4pt,
        fill=white,
        inner sep=2pt
      ] {$w_{\mathscr{K}}(L)$};

    \draw[measure,red!65!black]
      (3.85,0) -- (3.85,\striptau)
      node[
        midway,
        right=3pt,
        fill=white,
        inner sep=1pt
      ] {$\tau$};

    \node[fill=black!3,inner sep=2pt]
      at (1.35,1.03) {$\mathscr{K}$};

    \node[
      red!65!black,
      anchor=north
    ] (striplabel) at (1.45,-2.00)
      {$\mathscr{K}\cap\{|L|\le\tau\}$};

    \draw[->,red!65!black,line width=0.75pt]
      (striplabel.north)
      to[out=110,in=-65] (0.75,-0.15);

  \end{scope}
\end{tikzpicture}

\caption{
An illustration of the estimate in Lemma \ref{lem:complexstrip}.
}
\label{fig:complex-central-strip}
\end{figure}

\begin{lemma}\label{lem:complexstrip}
Let $p\ge1$, let $\mathscr{K}\subset\C^p$ be a compact convex circled set with
nonempty interior, and let $L$ be a nonzero complex-linear functional.
For every $\tau>0$,
\begin{equation}\label{eq:complexstrip}
 \frac{\vol_{2p}\bigl(\mathscr{K}\cap\{|L|\le\tau\}\bigr)}{\vol_{2p}(\mathscr{K})}
 \le p(2p-1)\left(\frac{\tau}{w_{\mathscr{K}}(L)}\right)^2.
\end{equation}
\end{lemma}
\begin{proof}
We denote $r=w_{\mathscr{K}}(L)>0$.  When $p=1$, $\mathscr{K}$ is a disk centered at
zero and can be represented by
\[\mathscr{K}=\{|z|\leq r_{\mathscr{K}}\} \ {\rm with \ some}\ r_{\mathscr{K}}>0.\]
$L$ can be represented by 
\[L(\vartheta)= z_0 \vartheta, \quad \forall \ \vartheta\in \C  \ {\rm with\ some}\ z_0\in \C\setminus\{0\}.\]
Then $r=|z_0|r_{\mathscr{K}}$ and the left-hand side of \eqref{eq:complexstrip} is 
\[\frac{\min (r_{\mathscr{K}}^2, \tau^2/|z_0|^2)}{r_{\mathscr{K}}^2}=\min(1,\tau^2/r^2).\]
So \eqref{eq:complexstrip} holds for $p=1$. Now suppose $p\ge2$ and put
$m=2p-2$.  Choose complex-linear coordinates $(z,y)\in\C\times\C^{p-1}$
with $z=L(\vartheta)$.
Under this coordinate change, the constant Jacobian will cancel in the volume
ratio in the left-hand side of \eqref{eq:complexstrip}.  Convexity and circular symmetry of $\mathscr{K}$ remain after this coordinate change, and give $L(\mathscr{K})=\{|z|\le r\}$.
Define
\[
 \mathscr{K}_z=\{y:(z,y)\in \mathscr{K}\},\qquad f(z)=\vol_m(\mathscr{K}_z).
\]
The circular symmetry yields 
 \begin{equation}\label{eq:rotation-invariance-section-volume}
   f(\exp(i\theta)z)=\vol_m(\mathscr{K}_{\exp(i\theta)z})= \vol_m(\exp(i\theta)\mathscr{K}_{z}) =f(z),
 \end{equation}
 and convexity also yields
\begin{equation}\label{eq:BM-seed}
   \tfrac12(\mathscr{K}_z+\mathscr{K}_{-z})\subseteq \mathscr{K}_0.
\end{equation}
Applying the Brunn--Minkowski inequality in $\R^m$ to \eqref{eq:BM-seed}, we obtain that 
\[
f(0)= \vol_m(\mathscr{K}_0)\geq \vol_m(\tfrac12(\mathscr{K}_z+\mathscr{K}_{-z})) 
    \geq \left[  (\vol_m(\mathscr{K}_z)^{1/m}+\vol_m(\mathscr{K}_{-z})^{1/m})/2 \right]^m =f(z).
\]
Here for the last equality above we used \eqref{eq:rotation-invariance-section-volume} with $\theta=\pi$ to give $f(-z)=f(z)$.

For each $\theta\in [0, 2\pi)$, we take $(r\exp(i\theta),y_\theta)\in \mathscr{K}$.
Convexity of $\mathscr{K}$ and $(1-s/r)\cdot 0 +(s/r)\cdot r\exp(i\theta)=s\exp(i\theta)$ imply
\[
 (1-s/r)\mathscr{K}_0+(s/r)y_\theta\subseteq \mathscr{K}_{s\exp(i\theta)} \qquad {\rm for}\ 0\le s\le r.
 \]
Therefore, since shifting a set by $(s/r)y_\theta$ doesn't change its volume, we have 
\[f(s\exp(i\theta))\ge(1-s/r)^m f(0)\qquad {\rm for}\ 0\le s\le r.\]
Then we can integrate in $y\in \C^{p-1}$ first, then use polar coordinates to integrate in $z\in \C$, and obtain
\begin{align*}
 \vol_{2p}(\mathscr{K})
 \ge2\pi f(0)\int_0^r s(1-s/r)^m\,ds
   &=\frac{2\pi f(0)r^2}{(m+1)(m+2)},\\
 \vol_{2p}\bigl(\mathscr{K}\cap\{|L|\le\tau\}\bigr)
 &\le\pi\tau^2 f(0).
\end{align*}
The quotient of the above two inequalities is exactly \eqref{eq:complexstrip}, since
$(m+1)(m+2)/2=p(2p-1)$.
\end{proof}

\subsection{The propagation lemma}
We now state and prove the propagation lemma for a $(c_0,C_0)$-linear-martingale system. 

The propagation condition \eqref{eq:propagation-condition} implies that if a source value satisfies $|S_{f,j}(\omega_{<f};\vartheta)|>\delta>0$, then, with all other bits and the input $\vartheta$ fixed, there is a choice of $\omega_{f,j}$ for which at least half of the target values in block $f$ satisfy $|Y_{f,k}(\omega_{\leq f};\vartheta)|>c_0\delta/2$. 
Such a propagation mechanism also appears in \cite[Page 72]{LSZ26}. We will show that, with high probability, sufficiently many source blocks containing a source with a value of large modulus force many targets to have values of large modulus, simultaneously for all inputs $\vartheta$ (see \eqref{eq:propagation-bound} and \eqref{eq:propagation-bound-equal}).

For the setup, take a $(c_0,C_0)$-linear-martingale system $\mathcal{M}$. Set
\begin{equation}\label{eq:fan-NB}
 N=\sum_{f=1}^{n_{\mathrm{blk}}}(r_f+b_f),\qquad B=\max_{1\le f\le n_{\mathrm{blk}}}b_f
\end{equation}
to be the total size of the system and the maximum of the sizes of target sets. In the remainder of the paper, we will occasionally suppress, for notational convenience, the dependence of the source values $S_{f,j}(\omega_{<f};\vartheta)$ and target values $Y_{f,k}(\omega_{\leq f};\vartheta)$ on the random word $\omega$. Define
\begin{align}
 W_\delta(\vartheta)
 &=\sum_{f=1}^{n_{\mathrm{blk}}} b_f\,\mathds{1}\!\left[
       \max_{1\le j\le r_f}|S_{f,j}(\vartheta)|\ge\delta\right],
       \label{eq:fan-highmass}\\
 M_\tau(\vartheta)
 &=\sum_{f=1}^{n_{\mathrm{blk}}}\sum_{k=1}^{b_f}
       \mathds{1} \left[|Y_{f,k}(\vartheta)|>\tau\right].
       \label{eq:fan-outputcount}
\end{align}
Thus $W_{\delta}(\vartheta)$ counts source blocks containing at least one source whose value has modulus at least $\delta$, with each block weighted by its target count $b_f$, while $M_\tau(\vartheta)$ counts targets whose values have modulus greater than $\tau$.
Here and in the rest of this paper, $\mathds{1}[\mathcal E]$ equals one when the event $\mathcal E$ holds and zero otherwise. 

The following propagation lemma states that with large probability, \eqref{eq:fan-outputcount} can be lower bounded by \eqref{eq:fan-highmass}.

\begin{lemma}\label{lem:propagation}
Assume $\mathcal M$ is a $(c_0,C_0)$-linear-martingale system defined as in Definition \ref{def:linear-martingale-system} with input body
\[\mathscr{K}=B_{R_{\rm in}}=\{\vartheta\in\C^p:\|\vartheta\|_2\le R_{\rm in}\}\subset \C^p,\qquad R_{in}>0 .\]
Let $N,B$ be defined as in \eqref{eq:fan-NB} for $\mathcal{M}$. Let $\delta>0$, $\tau>0$, and set
\begin{equation}\label{eq:fan-budget}
 h=\delta\min(1,c_0/2),\qquad D_p=\sqrt{p(2p-1)},\qquad
 m=\min\left\{N,\left\lfloor
 \frac{p\log(C_0 R_{\rm in}/\tau)}{\log(h/(D_p\tau))}
 \right\rfloor\right\}.
\end{equation}
Suppose also that the parameters $c_0,C_0,\delta,\tau$ satisfy
\[\tau/C_0 \le R_{\rm in} \ {\rm and} \  h>D_p\tau.\]
Then for each $L>0$, outside an event of probability at most $\exp(-L)$,
simultaneously for all inputs $\vartheta\in \mathscr{K}$, we uniformly have
\begin{equation}\label{eq:propagation-bound}
 2M_\tau(\vartheta)\ge
 \frac{1-\exp(-1)}{2}W_\delta(\vartheta)
 -B\left[m\bigl(1+\log(N+1)\bigr)+L\right].
\end{equation}
In particular, if all target sets $\mcK_f$ have size $B$ and 
\begin{equation}\label{eq:def-P-delta}
P_\delta(\vartheta) = \sum_{f=1}^{n_{\mathrm{blk}}} \,\mathds{1}\!\left[
       \max_{1\le j\le r_f}|S_{f,j}(\vartheta)|\ge\delta\right],
\end{equation} 
is the number of source blocks containing a source whose value has modulus at least $\delta$, then
\begin{equation}\label{eq:propagation-bound-equal}
 2M_\tau(\vartheta)\ge
 \frac{1-\exp(-1)}{2} B P_\delta(\vartheta)
 -B\left[m\bigl(1+\log(N+1)\bigr)+L\right].
\end{equation}
\end{lemma}

\begin{remark}
  We emphasize once again the order in which these objects are fixed throughout.  All source blocks, targets, and the input body $\mathscr{K}$ are fixed before randomness is introduced. We first fix the word (i.e. the randomness) $\omega$ in the high-probability event on which our argument is carried out. This determines the coefficients of all source and target values. We then take an arbitrary input $\vartheta\in\mathscr{K}$. At this point, the evaluations of all source and target values at $\vartheta$ are determined, and consequently so are $W_{\delta}(\vartheta)$ and $M_{\tau}(\vartheta)$.
\end{remark}

We explain the ideas behind the proof.
Suppose instead that we fix an input $\vartheta\in \mathscr{K}$ before sampling the random word.
In a block containing a source whose value has modulus at least $\delta$, choose the first such source $(f,j)$; this choice depends only on the preceding blocks.
Reveal all other bits in the current block, leaving $\omega_{f,j}$ unrevealed.
By \eqref{eq:propagation-condition}, for each target $(f,k)_{\rm tag}$, at least one of the two possible target values $Y_{f,k}^{(0)}(\widehat\omega_{f,j};\vartheta)$ and $Y_{f,k}^{(1)}(\widehat\omega_{f,j};\vartheta)$ has modulus at least $c_0\delta/2>\tau$.
The pigeonhole principle therefore gives one common choice of the remaining bit for which at least half of the target values have modulus greater than $\tau$.
The conditional probability of this success is at least $1/2$.
Then using independence of the bits, we have that for every fixed $\vartheta$ and every $t\ge0$,
\[
\P\left[2M_\tau(\vartheta)<\frac{1-\exp(-1)}2W_\delta(\vartheta)-Bt\right]\le \exp(-t).
\]

The difficulty is to obtain a single event on which such an inequality holds for every $\vartheta\in \mathscr{K}$.
A direct union bound over the uncountable set $\mathscr{K}$ is not plausible.
An ordinary Euclidean net, with suitable slack in the thresholds, would need a mesh of order $\tau/C_0$ and incur an entropy cost of order $p\log(1+C_0R_{\rm in}/\tau)$, which may be too large to absorb into the fixed-input tail estimate.

Instead, we encode the relevant information about each pair $(\vartheta,\omega)$ in a \textit{record} of cuts of the input body $\mathscr{K}$.
Whenever a source value or target value is small at $\vartheta$ but has large width on the current body, we intersect that body with the strip on which the functional has modulus at most $\tau$ and record its label.
Each cut retains $\vartheta$ and the common inner ball $B_{\tau/C_0}$, while removing a definite fraction of the volume.
The volume estimate Lemma~\ref{lem:complexstrip} bounds the number of cuts by $m$.
The \emph{record} is the sequence of source and target labels at which cuts are made, listed in the given scanning order.
Since the labels come from a fixed set of size $N$ and are scanned in a fixed order, the deterministic family $\mathscr R_m$ of possible records has cardinality at most $(N+1)^m$.

We will then bound the propagation loss for each record, and take a union bound over at most $(N+1)^m$ records.

More precisely, for each given record $\mathcal I$, we replay its prescribed cuts without choosing an input.
This defines a random loss function $H_{\mathcal I}(\omega)$, measuring the imbalance between active source blocks and successful blocks.
Crucially, the replay runs under the original probability law: we do not condition on the event that this record arises from a particular input.
Define the loss in propagation by
\begin{equation}   \label{eq:defnD}
 D(\vartheta,\omega):=\frac{1-\exp(-1)}2W_\delta(\vartheta)-2M_\tau(\vartheta).
\end{equation}
A deterministic comparison with the record of a pair $(\vartheta,\omega)$ gives
\begin{equation}\label{eq:propagation-loss-domination}
 \sup_{\vartheta\in \mathscr{K}}D(\vartheta,\omega)
 \le Bm+\max_{\mathcal I\in\mathscr R_m}H_{\mathcal I}(\omega).
\end{equation}
The same conditional success mechanism as in the fixed-input argument gives, for every fixed record $\mathcal I$,
\[
 \P\left[H_{\mathcal I}>Bt\right]\le \exp(-t),
 \qquad t\ge0.
\]
Thus the key union bound is over the finite family of records:
\[
 \P\left[\sup_{\vartheta\in \mathscr{K}}D(\vartheta,\omega)>Bm+Bt\right]
 \le \sum_{\mathcal I\in\mathscr R_m}\P[H_{\mathcal I}>Bt]
 \le (N+1)^m \exp(-t).
\]
Taking $t=L+m\log(N+1)$ then gives Lemma~\ref{lem:propagation}.
We remark that the three error terms $Bm$, $Bm\log(N+1)$, and $BL$ account respectively for the cuts, the number of records, and the chosen failure probability.

We next provide the detailed proof, implementing this outline.
\begin{proof}[Proof of Lemma \ref{lem:propagation}]
\noindent\textbf{Step 1: define the records through cuts.}
Let
\[
 \mathscr L=\bigcup_{f=1}^{n_{\mathrm{blk}}}(\mathcal J_f\cup\mathcal K_f),
 \qquad |\mathscr L|=N.
\]
Fix a scanning order before revealing any randomness: visit the blocks in increasing order of $f$, and within each block scan the sources in increasing order of $j$, followed by the targets in increasing order of $k$.

For now fix a word $\omega$ and an input $\vartheta\in \mathscr{K}$. Starting from the current body $\mathscr{K}'=\mathscr{K}=B_{R_{\rm in}}$, perform the following greedy cuts in the prescribed order. All source and target values in this step are evaluated at the fixed word $\omega$.
\begin{itemize}
\item At a source label $(f,j)$, make a cut if
\[
|S_{f,j}(\vartheta)|\le\tau
 \quad\hbox{and}\quad
 w_{\mathscr{K}'}(S_{f,j})\ge\delta.
\]
Record $(f,j)$ and replace $\mathscr{K}'$ by $\mathscr{K}'\cap\{|S_{f,j}|\le\tau\}$.
\item At a target label $(f,k)_{\rm tag}$, make a cut if
\begin{equation}\label{eq:cut-in-targets}
 |Y_{f,k}(\vartheta)|\le\tau
 \quad\hbox{and}\quad
 w_{\mathscr{K}'}(Y_{f,k})\ge c_0\delta/2.
\end{equation}
Record $(f,k)_{\rm tag}$ and replace $\mathscr{K}'$ by $\mathscr{K}'\cap\{|Y_{f,k}|\le\tau\}$.
\end{itemize}
If the relevant condition fails, leave $\mathscr{K}'$ unchanged. Widths are always computed on the current body, after all preceding cuts. Let $\mathcal I(\vartheta,\omega)\subset\mathscr L$ be the set of recorded labels. The fixed scanning order determines the order of the cuts from this set.

Every cut keeps $\vartheta\in \mathscr{K}'$. It also preserves compactness, convexity, and the circled property of the current body. By \eqref{eq:coefficient-bound-condition}, every source value or target value $L$ satisfies $|L(\vartheta')|\le C_0\|\vartheta'\|_2$ for every $\vartheta'\in\C^p$, so every intermediate body contains the same ball:
\[
B_{\tau/C_0}\subset \mathscr{K}'\subset B_{R_{\rm in}}.
\]
At each recorded cut, the width of the corresponding functional is at least $h=\delta\min(1,c_0/2)$. Lemma~\ref{lem:complexstrip} therefore gives
\[
\frac{\vol_{2p}(\mathscr{K}'\cap\{|L|\le\tau\})}{\vol_{2p}(\mathscr{K}')}
 \le D_p^2(\tau/h)^2<1.
\]
If $q=|\mathcal I(\vartheta,\omega)|$ is the number of cuts, and $\mathscr{K}_{\rm end}$ is the final body, then
\begin{equation}\label{eq:fan-volume-ledger}
 \left(\frac{\tau}{C_0R_{\rm in}}\right)^{2p}
 \le\frac{\vol_{2p}(\mathscr{K}_{\rm end})}{\vol_{2p}(B_{R_{\rm in}})}
 \le\left[D_p^2(\tau/h)^2\right]^q.
\end{equation}
 See Figure \ref{fig:common-inner-ball} for an illustration.
Taking logarithms and using $q\le N$ yields the deterministic bound
\[
|\mathcal I(\vartheta,\omega)|\le
 \min\left\{N,\left\lfloor
 \frac{p\log(C_0R_{\rm in}/\tau)}{\log(h/(D_p\tau))}
 \right\rfloor\right\}=m.
\]

\begin{figure}[htbp]
\centering
\begin{tikzpicture}[
  scale=0.95,
  line cap=round,
  line join=round,
  >=Latex,
  every node/.style={font=\small}
]
  \def\outerradius{3.0}
  \def\innerradius{0.70}

  \fill[black!2]
    (0,0) circle (\outerradius);

  \foreach \bandangle/\bandcolor/\halfwidth in {
    8/red!65!black/0.98,
    60/green!45!black/0.86,
    120/blue!65!black/0.75
  } {

    \begin{scope}
      \clip (0,0) circle (\outerradius);

      \begin{scope}[rotate=\bandangle]
        \clip
          (-3.7,-\halfwidth)
          rectangle
          (3.7,\halfwidth);

        \foreach \hatchindex in {-16,...,16} {
          \draw[
            \bandcolor,
            opacity=0.45,
            line width=0.35pt
          ]
            ({0.34*\hatchindex},-\halfwidth)
            --
            ({0.34*\hatchindex+0.55},\halfwidth);
        }
      \end{scope}
    \end{scope}

    \begin{scope}[rotate=\bandangle]
      \draw[\bandcolor,line width=0.85pt]
        (-3.55,\halfwidth) -- (3.55,\halfwidth);

      \draw[\bandcolor,line width=0.85pt]
        (-3.55,-\halfwidth) -- (3.55,-\halfwidth);
    \end{scope}
  }

  \draw[line width=1pt]
    (0,0) circle (\outerradius);

  \filldraw[
    fill=white,
    draw=black,
    line width=1pt
  ]
    (0,0) circle (\innerradius);

  \node at (0,0) {$B_{\tau/C_0}$};

  \node[anchor=west] at (2.65,-2.25)
    {$B_{R_{\mathrm{in}}}$};

\end{tikzpicture}
\caption{
An illustration of the estimate \eqref{eq:fan-volume-ledger}.
}
\label{fig:common-inner-ball}
\end{figure}

We will take a union bound over the deterministic family
\[
 \mathscr R_m:=\{\mathcal I\subset\mathscr L:|\mathcal I|\le m\}.
\]
Its cardinality satisfies
\[
|\mathscr R_m|=\sum_{q=0}^m\binom Nq\le(N+1)^m.
\]

\smallskip
\noindent\textbf{Step 2: define $H_{\mathcal I}$ by replaying a fixed record.}
We now define the quantity $H_{\mathcal I}$ which will be used to bound the loss in propagation for a given record.

Fix $\mathcal I\in\mathscr R_m$, before sampling $\omega$. Starting from $\mathscr{K}'=B_{R_{\rm in}}$, replay the prescribed cuts as follows, using the same scanning order as in Step 1.
\begin{enumerate}
\item Scan the sources in block $f$. At each $(f,j)\in\mathcal I$, replace $\mathscr{K}'$ by $\mathscr{K}'\cap\{|S_{f,j}|\le\tau\}$; at every other source leave it unchanged. Let $\mathscr{K}_f$ be the body after these source cuts.
\item Define the activity indicator
\[
A_f
 :=\mathds{1}\left[\max_{1\le j\le r_f}
 w_{\mathscr{K}_f}(S_{f,j})\ge\delta\right].
\]
If $A_f=1$, choose the first index $j_f$ whose source value has width at least $\delta$.
\item In an active block (i.e., where $A_f=1$), reveal all other bits in block $f$ while leaving $\omega_{f,j_f}$ unrevealed. For $e\in\{0,1\}$, let $Y_{f,k}^{(e)}$ be the target value obtained by setting this remaining bit equal to $e$. Define
\[
G_f
 :=\left\{e\in\{0,1\}:
 \left|\left\{1\le k\le b_f:
 w_{\mathscr{K}_f}(Y_{f,k}^{(e)})\ge c_0\delta/2
 \right\}\right|\ge b_f/2\right\}.
\]
Then reveal $\omega_{f,j_f}$ and set
\[
I_f
 :=\mathds{1}[\omega_{f,j_f}\in G_f].
\]
In an inactive block, reveal all its bits and set $I_f=0$. Thus
\[
0\le I_f\le A_f.
\]
\item Scan the targets in block $f$. At each label $(f,k)_{\rm tag}\in\mathcal I$, replace $\mathscr{K}'$ by $\mathscr{K}'\cap\{|Y_{f,k}|\le\tau\}$; otherwise leave it unchanged.
\end{enumerate}

This procedure is defined for every record and every word, without an input $\vartheta$. Throughout this procedure, we always have $B_{\tau/C_0}\subset \mathscr{K}'$. Since source values in block $f$ depend only on $\omega_{<f}$, the body $\mathscr{K}_f$, the indicator $A_f$, and the selected index $j_f$ on active blocks are measurable with respect to $\mathscr F_{f-1}$. The indicator $I_f$ is $\mathscr F_f$-measurable.

Define
\[
H_{\mathcal I}(\omega)
 :=a_0\sum_{f=1}^{n_{\mathrm{blk}}} b_fA_f
   -\sum_{f=1}^{n_{\mathrm{blk}}} b_fI_f,
 \qquad a_0=\frac{1-\exp(-1)}2.
\]

\smallskip
\noindent\textbf{Step 3: bound the upper tail of $H_{\mathcal I}$.}
Keep $\mathcal I$ fixed as in Step 2. On an active block, condition on the history $\mathscr F_{f-1}$ and all current bits except the selected bit $\omega_{f,j_f}$. The selected bit is still $\mathrm{Bernoulli}(1/2)$. The propagation condition \eqref{eq:propagation-condition}, applied on the fixed body $\mathscr{K}_f$, gives for every target
\[
 w_{\mathscr{K}_f}(Y_{f,k}^{(0)})
 +w_{\mathscr{K}_f}(Y_{f,k}^{(1)})
 \ge |c_{f,j_f,k}|w_{\mathscr{K}_f}(S_{f,j_f})
 \ge c_0\delta.
\]
For each $k$, at least one of the two widths is at least $c_0\delta/2$. By the pigeonhole principle, one common choice of $e$ has this property for at least half the target values. Hence
\[
G_f\ne\varnothing
 \qquad\hbox{whenever }A_f=1.
\]
The conditional success probability is therefore at least $1/2$ on an active block. Averaging over the other current bits, and using $I_f=0$ on inactive blocks, yields
\begin{equation}\label{eq:fan-fresh-active}
 \P\left[I_f=1\mid\mathscr F_{f-1}\right]
 \ge \frac12 A_f.
\end{equation}
Let $x_f=b_f/B\in(0,1]$. Since $A_f$ is $\mathscr F_{f-1}$-measurable,
\[
 \E\left(\exp(-x_fI_f)\mid\mathscr F_{f-1}\right)
 =1-(1-\exp(-x_f))
       \P\left[I_f=1\mid\mathscr F_{f-1}\right]
 \le 1-\frac12 A_f(1-\exp(-x_f))
 \le \exp\left(-a_0x_fA_f\right),
\]
where we used $1-\exp(-x)\ge(1-\exp(-1))x$ for $0\le x\le1$, and $1-\exp(-x)\leq x$ for $x\in \R$. Thus
\[
\E\left[
 \exp\left(\frac{a_0b_fA_f-b_fI_f}B\right)
 \,\middle|\,\mathscr F_{f-1}\right]\le1.
\]
Consequently,
\[
 Z_n:=\exp\left(\frac1B
 \sum_{f=1}^n(a_0b_fA_f-b_fI_f)\right),
 \qquad 0\le n\le n_{\mathrm{blk}},
\]
is a nonnegative supermartingale with $Z_0=1$. In particular,
\[
\E\exp\left(H_{\mathcal I}/B\right)\le1.
\]
Markov's inequality now gives the required upper tail:
\begin{equation}\label{eq:record-upper-tail}
 \P\left[H_{\mathcal I}>Bt\right]\le \exp(-t),
 \qquad t\ge0.
\end{equation}

\smallskip
\noindent\textbf{Step 4: bound the loss in propagation and apply a union bound.}
Fix any word $\omega$ and any input $\vartheta\in \mathscr{K}$, and take the record $\mathcal I=\mathcal I(\vartheta,\omega)$ from Step 1. Replay this record with the same word $\omega$, as in Step 2. The procedures have the same $\mathscr{K}'$ at every point in the scan.

Use $\mathscr{K}_f,A_f,G_f,I_f$ for the replay quantities associated with this selected record. Since every greedy cut retains $\vartheta$, we have $\vartheta\in \mathscr{K}_f$. A source value with $|S_{f,j}(\vartheta)|\ge\delta$ therefore has width at least $\delta$ on this body. Hence
\begin{equation}\label{eq:fan-activity}
 W_\delta(\vartheta)\le\sum_{f=1}^{n_{\mathrm{blk}}} b_fA_f.
\end{equation}

Let
\[
 q_f=|\mathcal I\cap\mathcal K_f|,
 \qquad
 M_f(\vartheta)=\sum_{k=1}^{b_f}\mathds{1}[|Y_{f,k}(\vartheta)|>\tau],
 \qquad M_\tau(\vartheta)=\sum_{f=1}^{n_{\mathrm{blk}}} M_f(\vartheta).
\]
If $q_f=0$, the body remains equal to $\mathscr{K}_f$ throughout the target scan. If also $I_f=1$, at least $b_f/2$ target values have width at least $c_0\delta/2$ on this body. Each of these target values must have modulus greater than $\tau$ at $\vartheta$, since otherwise \eqref{eq:cut-in-targets} would force a target cut. Thus $b_fI_f\le2M_f(\vartheta)$ when $q_f=0$. If $q_f\ge1$, the bound $b_fI_f\le Bq_f$ suffices. Summing these two cases gives
\[
\sum_{f=1}^{n_{\mathrm{blk}}} b_fI_f
 \le2M_\tau(\vartheta)+B\sum_{f=1}^{n_{\mathrm{blk}}} q_f
 \le2M_\tau(\vartheta)+Bm.
\]
Together with \eqref{eq:fan-activity} and the definition of $H_{\mathcal I}$, this yields
\[
D(\vartheta,\omega)
 \le Bm+H_{\mathcal I(\vartheta,\omega)}(\omega)
\]
with $D(\vartheta,\omega)$ defined by \eqref{eq:defnD}.
Since $|\mathcal I(\vartheta,\omega)|\le m$ for every pair $(\vartheta,\omega)$, taking the supremum over $\vartheta$ gives \eqref{eq:propagation-loss-domination}.

We can now take a union bound over the fixed finite family $\mathscr R_m$. For every $t\ge0$, by \eqref{eq:propagation-loss-domination} and \eqref{eq:record-upper-tail},
\[
\P\left[\sup_{\vartheta\in \mathscr{K}}D(\vartheta,\omega)>Bm+Bt\right]
 \le \P\left[\max_{\mathcal I\in\mathscr R_m}H_{\mathcal I}>Bt\right]\le\sum_{\mathcal I\in\mathscr R_m}
       \P[H_{\mathcal I}>Bt]\le (N+1)^m \exp(-t).
\]
Taking $t=L+m\log(N+1)$, we obtain an event of probability at least $1-\exp(-L)$ on which, simultaneously for every $\vartheta\in \mathscr{K}$,
\[
 2M_\tau(\vartheta)\ge a_0W_\delta(\vartheta)
 -B\left[m\bigl(1+\log(N+1)\bigr)+L\right].
\]
This is \eqref{eq:propagation-bound}. When all target sets have size $B$, substituting $W_\delta(\vartheta)=B P_\delta(\vartheta)$ gives \eqref{eq:propagation-bound-equal}.
\end{proof}

At the end of this section, we comment on how Lemma \ref{lem:propagation} is related to the previous work \cite{LSZ26}. 
A key point is the conditional success estimate \eqref{eq:fan-fresh-active}, which leads to the exponential supermartingale in Step 3. Similar conditional estimates and exponential arguments also arise in \cite[Page 72]{LSZ26}.
Another important idea in the above proof is the discretization of the initial input body $\mathscr{K}$, through the cut procedure. This is precisely the goal pursued in \cite[Problem 2.7]{LSZ26}. In \cite[Theorem 2.6]{LSZ26}, a version of PDUC is proved; however, the corresponding good event depends on the input, although the probability lower bound is uniform in these data. Theorem \ref{thm:duc-linear} below essentially resolves \cite[Problem 2.7]{LSZ26} in the affirmative.

\section{Probabilistic discrete unique continuation principle}\label{sec:DUC}

In this section, we establish the main step of PDUC for the Schr\"odinger equation on $\Z^d$, through the bootstrap strategy outlined in Section \ref{ssec:iop}.

\subsection{Graded sets and the statement of PDUC}\label{sec:graded-set}

Before stating the main PDUC, we need to introduce the definition of the graded set. 
This object was introduced in \cite[Definition 3.3]{LZ22}, where it is instrumental in describing the scale-by-scale evolution of the probabilities arising in MSA. Indeed, the MSA framework requires the incorporation of the so-called free sites, which were first introduced by \cite{Bou04} to overcome the lack of eigenvalue variation in the localization problem for the ABM. The graded set is designed to capture the geometric sparsity of the frozen sites (i.e., the sites that are not free) throughout the iteration. This, in turn, ensures that sufficiently many free sites remain available at every step of the iteration.

Below, for any $r>0$ and $a\in \Z^d$, we write
\[
Q_r(a)=\{x\in\mathbb Z^d:\|x-a\|_\infty\le r\},\qquad Q_r=Q_r(0),
\]
for the cube ($\ell^{\infty}$-ball) in $\Z^d$.

\begin{definition}\label{def:fsa-graded}
Let $N\in\mathbb N$, $0<\varepsilon<1$, and let $\boldsymbol\rho=(\rho_1,\ldots,\rho_{n_{\mathrm{lev}}})$ be a finite list of real numbers greater than one, with $n_{\mathrm{lev}}\ge1$ and
\[
\rho_{i+1}\ge\rho_i^{1+2\varepsilon}.
\]
We say the union $\mathcal F=\mathcal F_0\cup\cdots\cup\mathcal F_{n_{\mathrm{lev}}}\subset\mathbb R^d$ is a \textit{$(N,\varepsilon,\boldsymbol\rho)$-graded set}, if it has the following properties:
\begin{itemize}
  \item[(1)] The set $\mathcal F_0$ can be represented by 
                  \[\mcF_0=\bigcup_{j\in \mathbb{N}} Z^{(j)},\]
  where each $Z^{(j)}$ is an open unit ball (under the $\ell^2$-norm) in $\R^d$ with a center in $\mathbb Z^d$ and 
  \[   \dist(Z^{(i)},Z^{(j)}) \geq  \varepsilon^{-1}, \ i\neq j.\]
  \item[(2)] For each $i\ge1$, the set $\mathcal F_i$ can be represented by 
   \[\mcF_i=\bigcup_{1\leq t\leq N} \mcF_i^{(t)}, \ \mcF_i^{(t)}= \bigcup_{j\in \mathbb{N}} Z^{(t,j)}.\]
   Here $Z^{(t,j)}$ are open balls of radius $\rho_i$ in $\R^d$; within each $\mcF_i^{(t)}$, we have 
     \[   \dist(Z^{(t,j)},Z^{(t,j')}) \geq  \rho_i^{1+\varepsilon}, \ j\neq j'.\]
\end{itemize}
We say $F$ is a \textit{discrete $(N,\varepsilon,\boldsymbol\rho)$-graded set}, if $F=\mathcal F\cap\mathbb Z^d$ for some $(N,\varepsilon,\boldsymbol\rho)$-graded set $\mcF$. 

For a $(N,\varepsilon,\boldsymbol\rho)$-graded set $\mcF$, or its discrete version $F=\mcF\cap \Z^d$, we say it is \textit{normal} in $Q_R(a)$ if for each $i\geq 1$, $Q_R(a)\cap \mcF_i\neq \emptyset$ implies
\[
\rho_i\le\operatorname{diam} (Q_R(a))^{1-\varepsilon/2}\le(2\sqrt d\,R)^{1-\varepsilon/2}.
\]
\end{definition}
See Figure \ref{fig:graded-set} for an illustration of a $(2,\varepsilon,\boldsymbol{\rho})$-graded set.

\begingroup
\definecolor{gfunit}{HTML}{6B7280}
\definecolor{gfblue}{HTML}{0072B2}
\definecolor{gfteal}{HTML}{008A67}
\definecolor{gforange}{HTML}{D55E00}
\definecolor{gfpurple}{HTML}{9B5198}

\begin{figure}

\begin{tikzpicture}[
  x=1cm,y=1cm,font=\small,
  gfball/.style={line width=0.85pt,fill opacity=0.14},
  gffamilytwo/.style={dash pattern=on 3pt off 1.8pt},
  gfmeasure/.style={<->,>=Latex,line width=0.65pt},
  gflegend/.style={anchor=west,align=left,inner sep=0pt}
]
  \draw[black!45,line width=0.7pt] (0,0) rectangle (10,10);
  \node[anchor=south west,inner sep=0pt] at (0,10.15)
    {$Q_R(a)$};

  \foreach \x/\y in {2/7.5,6.5/7.5,2.2/2.3}{
    \draw[gfball,draw=gforange,fill=gforange]
      (\x,\y) circle[radius=0.8];
  }
  \foreach \x/\y in {2.6/6.5,7/6.4,7/2.2}{
    \draw[gfball,gffamilytwo,draw=gfpurple,fill=gfpurple]
      (\x,\y) circle[radius=0.8];
  }

  \foreach \x/\y in {1/1,4.3/1.2,9/4.2,4.3/5,1/4.8,4.5/9,8.5/8.8}{
    \draw[gfball,draw=gfblue,fill=gfblue]
      (\x,\y) circle[radius=0.4];
  }
  \foreach \x/\y in {1.4/1.2,3.1/2.8,4.6/5.5,8.5/4.5,8.7/1.2,1/8.8,7/9}{
    \draw[gfball,gffamilytwo,draw=gfteal,fill=gfteal]
      (\x,\y) circle[radius=0.4];
  }

  \foreach \x/\y in {
    0.7/2.7,1.1/3.2,1.7/4,3/0.8,3/4.5,3.6/3.8,
    3.8/6.5,4.6/7,5.1/1.9,5.5/3.7,5.6/6,5.5/8.6,
    6.2/4.8,7.3/4.4,8.3/5.8,9.1/6.9,9.4/2.6,9.3/9.1,
    0.8/6.3,4/8,6.8/1,8.8/0.6,2/7.5}{
    \draw[draw=gfunit,fill=gfunit!18,line width=0.5pt]
      (\x,\y) circle[radius=0.1];
    \fill[gfunit] (\x,\y) circle[radius=0.018];
  }

  \draw[gfmeasure,draw=gforange] (2.8,7.5) -- (5.7,7.5)
    node[midway,above=5pt,text=black,fill=white,inner sep=2pt]
    {$\ge \rho_2^{\,1+\varepsilon}$};

  \fill[gforange] (2.2,2.3) circle[radius=0.025];
  \draw[-{Latex[length=1.5mm]},gforange,line width=0.65pt]
    (2.2,2.3) -- (3,2.3)
    node[midway,below=3pt,text=black] {$\rho_2$};
  \fill[gfblue] (4.3,1.2) circle[radius=0.022];
  \draw[-{Latex[length=1.3mm]},gfblue,line width=0.6pt]
    (4.3,1.2) -- (4.7,1.2);
  \node[inner sep=0pt] at (4.3,0.55) {$\rho_1$};

  \node[gflegend,font=\normalsize] at (10.55,9.55)
    {$\mathcal F=\mathcal F_0\cup\mathcal F_1\cup\mathcal F_2$};
  \node[gflegend] at (10.55,8.95)
    {$\mathcal F_i=\mathcal F_i^{(1)}\cup\mathcal F_i^{(2)}$,
     $i=1,2$};

  \draw[draw=gfunit,fill=gfunit!18,line width=0.6pt]
    (10.75,8.05) circle[radius=0.1];
  \node[gflegend] at (11.2,8.05)
    {$\mathcal F_0$\quad radius $1$};

  \node[gflegend,font=\bfseries] at (10.55,7.15)
    {Scale $\rho_1$};
  \draw[gfball,draw=gfblue,fill=gfblue]
    (10.75,6.55) circle[radius=0.18];
  \node[gflegend] at (11.2,6.55)
    {$\mathcal F_1^{(1)}$\quad family $1$};
  \draw[gfball,gffamilytwo,draw=gfteal,fill=gfteal]
    (10.75,5.95) circle[radius=0.18];
  \node[gflegend] at (11.2,5.95)
    {$\mathcal F_1^{(2)}$\quad family $2$};

  \node[gflegend,font=\bfseries] at (10.55,4.95)
    {Scale $\rho_2$};
  \draw[gfball,draw=gforange,fill=gforange]
    (10.75,4.35) circle[radius=0.25];
  \node[gflegend] at (11.2,4.35)
    {$\mathcal F_2^{(1)}$\quad family $1$};
  \draw[gfball,gffamilytwo,draw=gfpurple,fill=gfpurple]
    (10.75,3.6) circle[radius=0.25];
  \node[gflegend] at (11.2,3.6)
    {$\mathcal F_2^{(2)}$\quad family $2$};

  \node[gflegend] at (10.55,2.55)
    {$1<\rho_1<\rho_2$,\quad $N=2$};
  \node[gflegend] at (10.55,1.9)
    {$\rho_2\ge\rho_1^{\,1+2\varepsilon}$};
  \node[gflegend] at (10.55,0.8)
    {$F=\mathcal F\cap\mathbb Z^2$};

\end{tikzpicture}
\caption{An illustration of a $(2,\varepsilon,\boldsymbol\rho)$-graded set.}
\label{fig:graded-set}
\end{figure}
\endgroup

\smallskip

We next state the main PDUC, which is compatible with free sites.

Let $e_1,e_2,\cdots,e_d$ be the standard basis of $\Z^d$. 
Recall the operator $H$ given by \eqref{eq:ABM-schrodinger-operator}. The equation $Hu=Eu$ can be written as
\begin{equation}\label{eq:pointwise-equation}
\sum_{i=1}^d\bigl(u(x+e_i)+u(x-e_i)\bigr)=(2d-E+\lambda V_x)u(x).
\end{equation}

\begin{theorem}\label{thm:duc-linear}
Fix $d\ge2$, $\lambda>0$, a compact interval $I\subset\mathbb R$, $K\ge0$,
and $N\in\mathbb{N}$. There is $\varepsilon_0=\varepsilon_0(d)>0$ such that, for all $0<\varepsilon<\varepsilon_0$,
there are a scale $ L_*=L_*(N,\varepsilon)$ and constants
$C,c,R_0>0$ depending on
$d,\lambda,I,K,N,\varepsilon$, such that the following holds.

Take any $\boldsymbol\rho=(\rho_1,\ldots,\rho_{n_{\mathrm{lev}}})$ as in Definition~\ref{def:fsa-graded} with $\rho_1\ge L_*$, and $R\ge R_0$. Let $F$ be a fixed discrete $(N,\varepsilon,\boldsymbol\rho)$-graded set that is normal in $Q_R$.
For arbitrary prescribed fixed values $v_F\in[0,1]^{F\cap Q_{R-1}}$, set
$V_x=v_F(x)$ on $F\cap Q_{R-1}$ (i.e., the frozen sites) and let
\[\P[V_x=0]=\P[V_x=1]=\frac{1}{2},\quad \forall \ x\in Q_{R-1}\setminus F,\, \text{independently} .\]
Then for every fixed $E_0\in I$, there is
an event $\mathcal E_R(F,v_F,E_0) \subset \{0,1\}^{Q_{R-1}\setminus F}$, such that
\begin{equation}\label{eq:fsa-probability}
\mathbb P\bigl[\mathcal E_R(F,v_F,E_0)\bigr]
\ge
1-\exp\big(
-c(R/\log R)^{(d-1)/d}
\big),
\end{equation}
and on this event, for any $E\in I$ and non-zero
$u:Q_R\to\mathbb C$ satisfying \eqref{eq:pointwise-equation} for $x\in Q_{R-1}$, and
\[
|E-E_0|\le \exp(-CR),\qquad
\|u\|_{\ell^\infty(Q_R)}\le \exp(KR)|u(0)|,
\]
we have
\begin{equation}\label{eq:fsa-quadratic-count}
\left|\left\{
x\in Q_R\setminus F:
|u(x)|\ge \exp(-CR)|u(0)|
\right\}\right|
\ge \frac{cR^d}{(\log R)^d}.
\end{equation}
\end{theorem}

We remark that Theorem \ref{thm:duc-linear} only concerns a fixed disorder strength $\lambda$ and energies in a small interval near $E_0$. This is sufficient for the proof of localization. On the other hand, for $d\geq 3$, the PDUC without removing a graded set can be strengthened to hold uniformly for $\lambda, E$. We will discuss this in Appendix \ref{appendix:uniform-duc}.   

The rest of this section is devoted to the proof of Theorem \ref{thm:duc-linear}. We start with some preparations.

\subsection{Sparsity of graded sets}
In this subsection, we will discuss some facts concerning the sparsity of a graded set. 
First, we recall the so-called \textit{cone property}. 

\begin{definition}\label{def:cone-chain}
  For any $x \in \mathbb{Z}^d$, $1 \leq k \leq d$, and $\sigma \in \{1,-1\}$, define the \textit{cone} with root $x$ and direction $\sigma e_k$  to be 
  \[\mathfrak{C}(x,\sigma e_k)=\left\{y\in \Z^d: |y-(x+\sigma e_k)|\leq 1,\ y\neq x\right\}.\]
  A \textit{cone-type chain} is a path $P=(x_0,x_1,\cdots,x_m)$ in $\Z^d$ satisfying
  \[x_{j}\in \mathfrak{C}(x_{j-1},\sigma e_k) \qquad \forall \ 1\leq j\leq m.\]
  Moreover, we call $m$ the length of $P$, $x_0$ the root of $P$, and $\sigma e_k$ the direction of $P$.
\end{definition}

Based on the above definition, we have the following fact (see also \cite[Section 2.1]{LSZ26}):
\begin{proposition}[\protect{\cite[Section 2]{LZ22}}]\label{prop:cone-property}
Fix $d\ge2$, $\lambda>0$, and a compact interval $I\subset\R$, and take $0\le V_x\le1$ for each $x\in\Z^d$.
There is $C_{\rm cone}=C_{\rm cone}(d,\lambda,I)>0$ such that, for every $E\in I$, the following holds.
If \eqref{eq:pointwise-equation} holds at $x+\sigma e_k$, then
\[
\max_{y\in\mathfrak C(x,\sigma e_k)}|u(y)|
 \ge \exp(-C_{\rm cone})|u(x)|.
\]

Suppose that $u:Q_r(a)\to\C$ satisfies \eqref{eq:pointwise-equation} on $Q_{r-1}(a)$.
For every $x_0\in Q_r(a)$, every direction $\sigma e_k$, and every integer $m\ge0$ satisfying
\[
 2m\le r-\|x_0-a\|_\infty,
\]
there is a cone-type chain $P=(x_0,\ldots,x_m)\subset Q_r(a)$ in that direction such that
\[
 |u(x_j)|\ge \exp(-C_{\rm cone}j)|u(x_0)|,
 \qquad 0\le j\le m.
\]
In particular, if $r\ge4$ and $x_0\in Q_{r/2}(a)$, one may take $m=\lfloor r/4\rfloor$.

If $u$ instead satisfies \eqref{eq:pointwise-equation} on $Q_r(a)$, with the extension $u(x)=0$ for any $x\not\in Q_r(a)$ and $u(x_0)\ne0$, the same chain conclusion then holds whenever
\[
 2m\le r-\sigma(x_0-a)_k.
\]
In particular, choosing $\sigma$ so that $\sigma(x_0-a)_k\le0$, one may take $m=\lfloor r/2\rfloor$ even when $x_0$ lies on the boundary.
\end{proposition}

The next lemma reveals that a cone-type chain can intersect a graded set along only a small portion of the chain, due to sparsity of the graded set. 

\begin{lemma}\label{lem:fsa-packing}
Take $d\ge2$, $N\in\mathbb N$, $0<\varepsilon<1$, and let $\boldsymbol\rho=(\rho_1,\ldots,\rho_{n_{\mathrm{lev}}})$ be a scale sequence as in Definition~\ref{def:fsa-graded}. Let
$F$ be a discrete $(N,\varepsilon,\boldsymbol\rho)$-graded set.
Then every cone-type chain $P=(x_0,\ldots,x_m)$ satisfies
\begin{equation}\label{eq:fsa-packing}
|P\cap F|\lesssim \varepsilon m+1+N\sum_{\substack{1\le i\le n_{\mathrm{lev}}\\\mathcal F_i\cap P\ne\varnothing}}(m\rho_i^{-\varepsilon}+\rho_i).
\end{equation}
Moreover, there is an $\varepsilon_0=\varepsilon_0(d)\in(0,1)$ with the following property.
For every $N\in\mathbb N$, $0<\varepsilon<\varepsilon_0$, and $b>0$, one can choose $L_*=L_*(N,\varepsilon)>1$ and $R_{\rm spa}=R_{\rm spa}(b,d,N,\varepsilon)>e$, such that for every $R\ge R_{\rm spa}$, if $\rho_1\ge L_*$ and $F$ is normal in $Q_R$, then
\begin{equation}\label{eq:ample-free-site-in-long-chain}
  |P\setminus F|\ge 3m/4
\end{equation}
for every cone-type chain $P\subset Q_R$ of length $m\ge bR/\log R$.
\end{lemma}

\begin{proof}
Up to a coordinate exchange and a reflection, we can assume that $P$ has direction $e_1$ without loss of generality. Represent $F$ by 
\[F=\mcF\cap \Z^d,\ \mathcal F=\mathcal F_0\cup\cdots\cup\mathcal F_{n_{\mathrm{lev}}} \]
as in Definition \ref{def:fsa-graded}. 
Since each $x_j\in \mathfrak{C}(x_{j-1},e_1)$, we must have
\begin{equation}\label{eq:cone-chain-coordinates-changes}
  1\leq (x_j-x_{j-1})\cdot e_1\leq 2, \quad |x_j-x_{j-1}|\leq 2, \quad 1\leq j\leq m.
\end{equation}
In each $\mcF_{i}^{(t)}=\bigcup_{j\in \N}Z^{(t,j)}$, assume that $P$ intersects $g$ many balls
\[Z^{(t,j_1)},Z^{(t,j_2)},\cdots, Z^{(t,j_g)},\]
successively. To go from one ball to the other, the path $P$ must take at least 
$\rho_i^{1+\varepsilon}/2$
steps. This implies that 
\[g\leq 2m/\rho_{i}^{1+\varepsilon}+1.\]
By \eqref{eq:cone-chain-coordinates-changes}, at each step of $P$ the first coordinate must increase by at least one.
Thus a single ball in $\mcF_{i}^{(t)}$ contains at most $2\rho_i+2$ vertices of $P$. Therefore,
\begin{equation}\label{eq:chain-in-i-t-grade}
  |\mcF_i^{(t)}\cap P|\leq g (2\rho_i+2) \lesssim m\rho_i^{-\varepsilon}+\rho_i
\end{equation}
holds for each $i\geq 1$ and $1\leq t\leq N$. 

Moreover, for $\mcF_0$, the same argument used to derive \eqref{eq:chain-in-i-t-grade} also gives 
\begin{equation}\label{eq:chain-in-0-grade}
  |\mcF_0 \cap P| \lesssim \varepsilon m+1.
\end{equation}
Combining \eqref{eq:chain-in-i-t-grade} and \eqref{eq:chain-in-0-grade} yields \eqref{eq:fsa-packing}.

We next prove \eqref{eq:ample-free-site-in-long-chain}.
When $\varepsilon<1$, for $\rho_1$ sufficiently large depending on $\varepsilon$, we have
\[
\sum_{i=1}^{n_{\mathrm{lev}}}\rho_i^{-\varepsilon}
\le \sum_{i=0}^\infty \rho_1^{-\varepsilon(1+2\varepsilon)^i}
\le2\rho_1^{-\varepsilon},\qquad
\sum_{i=1}^j\rho_i
\le \rho_j\sum_{k=0}^{j-1}\rho_j^{(1+2\varepsilon)^{-k}-1}\le 2\rho_j,
\]
for any $1\le j\le n_{\mathrm{lev}}$.
Indeed, iterating the grading condition backwards gives $\rho_{j-k}\le\rho_j^{(1+2\varepsilon)^{-k}}$ for $0\le k<j$.
The successive terms in this power sum decrease by a factor of at least $\rho_1^{2\varepsilon}$, so their sum after the leading term is at most $2\rho_j^{1/(1+2\varepsilon)}$ when $\rho_1$ is sufficiently large.
Therefore, by choosing $\varepsilon$ small enough, then $\rho_1$ large enough depending on $N, \varepsilon$, we can make
\[ \varepsilon+N \sum_{i:\mathcal F_i\cap P\ne\varnothing}\rho_i^{-\varepsilon}\]
smaller than any prescribed positive constant.
Moreover, since $F$ is normal in $Q_R$, if $m\geq bR/\log R$ we can estimate
\[
N m^{-1} \sum_{i:\mathcal F_i\cap P\ne\varnothing} \rho_i  \lesssim \frac{N\log R}{b R} \operatorname{diam} (Q_R)^{1-\varepsilon/2} \lesssim_{b,N,d}R^{-\varepsilon/2} \log R.
\]
Thus, by taking $R$ large enough depending on $b,d,N,\varepsilon$, we can also make \[N m^{-1} \sum_{i:\mathcal F_i\cap P\ne\varnothing} \rho_i\]
smaller than any prescribed positive constant.
From these and \eqref{eq:fsa-packing}, we can get \eqref{eq:ample-free-site-in-long-chain} when $m\geq bR/\log(R)$, by taking $\varepsilon$ small enough, then $\rho_1$ large enough depending on $N, \varepsilon$, and finally $R$ large enough depending on $b,d,N,\varepsilon$.
\end{proof}

\subsection{Minimal adjusted bowl-shaped boundary}
For $u$ satisfying \eqref{eq:pointwise-equation}, we next construct certain sets, such that $u$ in the interior can be determined by its values on the boundary.

Write a lattice point as $(t,z)\in\mathbb Z\times\mathbb Z^{d-1}$. As in Section~\ref{ssec:iop}, for an integer $s\ge1$, let
$\Omega_\Sigma=[-s,s]^{d-1}\cap\mathbb Z^{d-1}$ 
and 
\[\partial^+\Omega_\Sigma=\left\{ z\in \Z^{d-1}\setminus \Omega_\Sigma: \exists \ z'\in \Omega_\Sigma \ {\rm such\ that } \  |z-z'|_1=1  \right\}, \]
and we take the bowl-shaped boundary
\[
T_s=\bigl(\{-s-1,-s\}\times\Omega_\Sigma\bigr)\cup\bigl(\{-s,\ldots,s-1\}\times\partial^+\Omega_\Sigma\bigr).
\]
It's easy to verify that if we fix $u$ on the boundary $T_s$, the values of $u$ on $Q_{s}$ can be recursively determined by
\begin{equation}\label{eq:fsa-recurrence}
u(t+1,z)=(2d-E+\lambda V_{t,z})u(t,z)-u(t-1,z)-\sum_{i=2}^d \bigl( u(t,z+e_i) + u(t,z-e_i)\bigr),
\end{equation}
for $t=-s, \ldots, s-1$.
Such bowl-shaped boundaries have been used by Bourgain-Klein \cite{BK13}, and also in the proof of the weak unique continuation estimate in \cite[Lemma 5.3]{DS20}.

However, due to the existence of frozen sites, we can not directly use the bowl-shaped boundary $T_s$.
We make adjustments as follows. Take any $F\subset \Z^d$. Define the \textit{height function} $h$ on $T_s\cup Q_s$ inductively, by setting
\[h(x)=0, \ {\rm for\ all\ } x\in T_s,\]
and letting
\begin{equation}\label{eq:fsa-depth}
h(x)=\mathds{1}[x\notin F]+\min_{y\in \mathfrak{C}(x,-e_1)}h(y),\qquad x\in Q_s\setminus T_s.
\end{equation}
Equivalently, $h(x)$ is the minimum number of free sites in a cone-type chain that originates from $T_s$ and ends at $x$. That is to say,
\begin{equation}\label{eq:height-interpretion}
  h(x)=\min_{\substack{P:T_s\rightarrow x \\ P {\rm \ is\ a \ cone-type\ chain \ in\ direction\ } e_1.}} |P\setminus (F\cup T_s)|.
\end{equation}
Indeed, the right-hand side of \eqref{eq:height-interpretion} equals zero on $T_s$, and also satisfies \eqref{eq:fsa-depth} since $x\in \mathfrak{C}(y,e_1)$ is equivalent to $y\in \mathfrak{C}(x,-e_1)$.

We define the \textit{adjusted bowl-shaped boundaries} to be the level sets of $h$:
\begin{equation}\label{eq:adjust-boundary}
  \mathfrak{B}_k:=\{ x\in Q_s\setminus F: h(x)=k \}, \quad k\geq 1.
\end{equation}
We also denote
\begin{equation}\label{eq:adjust-region}
  \mathfrak{D}_k:=\{ x\in Q_s: h(x)\geq k \}, \quad k\geq 1.
\end{equation}
See Figure \ref{fig:adjusted-boundaries-two-dimensions} for an illustration in 2D.

\begingroup
\definecolor{bkdepth1}{HTML}{CBDCF3}
\definecolor{bkdepth2}{HTML}{FBE7AA}
\definecolor{bkdepth3}{HTML}{F3CADA}
\definecolor{bkdepth4}{HTML}{CAE5D4}
\definecolor{bkdepth5}{HTML}{DDD1EE}
\definecolor{bkdepth6}{HTML}{F5D0AD}
\definecolor{bkdepth7}{HTML}{C6E4EB}
\definecolor{bkdepth8}{HTML}{E0E7B8}
\definecolor{bkdepth9}{HTML}{D9DCE3}
\definecolor{bkfrozen}{HTML}{008A91}

\begin{figure}[htbp]
\centering
\resizebox{0.98\linewidth}{!}{%
\begin{tikzpicture}[
  x=5mm,y=5mm,font=\small,
  bkpoint/.style={circle,minimum size=3.85mm,inner sep=0pt,
    font=\scriptsize,text=black},
  bktrace/.style={draw=black!65,fill=white,line width=0.55pt},
  bklegend/.style={anchor=west,align=left,inner sep=0pt},
  bkcross/.style={draw=bkfrozen,line width=1.25pt,line cap=round}
]
  \draw[step=1,black!7,line width=0.25pt]
    (-9,-10) grid (9,9);

  \foreach \t in {-10,-9}{
    \foreach \z in {-9,...,9}{
      \expandafter\xdef\csname bkfigdepth@\t @\z\endcsname{0}
      \draw[bktrace]
        (\z-0.17,\t-0.17) rectangle (\z+0.17,\t+0.17);
    }
  }
  \foreach \t in {-9,...,8}{
    \foreach \z in {-10,10}{
      \expandafter\xdef\csname bkfigdepth@\t @\z\endcsname{0}
      \draw[bktrace]
        (\z-0.17,\t-0.17) rectangle (\z+0.17,\t+0.17);
    }
  }

  \foreach \t in {-8,...,9}{
    \pgfmathtruncatemacro{\bkprev}{\t-1}
    \pgfmathtruncatemacro{\bktwoprev}{\t-2}
    \foreach \z in {-9,...,9}{
      \pgfmathtruncatemacro{\bkleft}{\z-1}
      \pgfmathtruncatemacro{\bkright}{\z+1}
      \edef\bkqa{\csname bkfigdepth@\bkprev @\z\endcsname}
      \edef\bkqb{\csname bkfigdepth@\bktwoprev @\z\endcsname}
      \edef\bkqc{\csname bkfigdepth@\bkprev @\bkleft\endcsname}
      \edef\bkqd{\csname bkfigdepth@\bkprev @\bkright\endcsname}

      \def\bkisfrozen{0}
      \ifnum\t=-5\relax
        \ifnum\z=-3\relax\def\bkisfrozen{1}\fi
      \fi
      \ifnum\t=0\relax
        \ifnum\z=8\relax\def\bkisfrozen{1}\fi
      \fi
      \pgfmathtruncatemacro{\bkq}{
        min(\bkqa,\bkqb,\bkqc,\bkqd)+1-\bkisfrozen}
      \expandafter\xdef\csname bkfigdepth@\t @\z\endcsname{\bkq}

      \pgfmathtruncatemacro{\bkqzero}{min(floor((\t+10)/2),10-abs(\z))}
      \ifnum\bkisfrozen=1\relax
        \draw[bkcross] (\z-0.23,\t-0.23) -- (\z+0.23,\t+0.23);
        \draw[bkcross] (\z-0.23,\t+0.23) -- (\z+0.23,\t-0.23);
      \else
        \ifnum\bkq<\bkqzero\relax
          \node[bkpoint,fill=bkdepth\bkq,draw=black!65,line width=0.5pt]
            at (\z,\t) {$\bkq$};
        \else
          \node[bkpoint,fill=bkdepth\bkq]
            at (\z,\t) {$\bkq$};
        \fi
      \fi
      \ifdefined\BkVerify
        \typeout{BKFIG t=\t,z=\z,q=\bkq,fr=\bkisfrozen,qzero=\bkqzero}
      \fi
    }
  }

  \draw[-{Latex[length=2mm]},black!75]
    (-11.3,-10.6) -- (-11.3,10.1) node[above] {$t$};
  \draw[-{Latex[length=2mm]},black!75]
    (-10.6,-11.6) -- (10.7,-11.6) node[right] {$z$};
  \foreach \t in {-10,-8,...,8}{
    \draw[black!65] (-11.42,\t) -- (-11.18,\t);
    \node[anchor=east,font=\scriptsize,inner sep=2pt]
      at (-11.45,\t) {$\t$};
  }
  \foreach \z in {-10,-8,...,10}{
    \draw[black!65] (\z,-11.72) -- (\z,-11.48);
    \node[anchor=north,font=\scriptsize,inner sep=3pt]
      at (\z,-11.8) {$\z$};
  }

  \begin{scope}[yshift=-23mm]
    \node[bklegend,font=\bfseries] at (12.7,8.8)
      {Adjusted boundary};
    \foreach \k/\x/\y in {
      1/13/7.6,2/16.6/7.6,3/20.2/7.6,
      4/13/6.2,5/16.6/6.2,6/20.2/6.2,
      7/13/4.8,8/16.6/4.8,9/20.2/4.8}{
      \node[bkpoint,fill=bkdepth\k] at (\x,\y) {$\k$};
      \node[bklegend] at (\x+0.7,\y) {$\mathfrak{B}_{\k}$};
    }
    \draw[bktrace] (12.83,3.03) rectangle (13.17,3.37);
    \node[bklegend] at (14,3.2) {$T_9$};
    \draw[bkcross] (12.77,1.07) -- (13.23,1.53);
    \draw[bkcross] (12.77,1.53) -- (13.23,1.07);
    \node[bklegend] at (14,1.3) {Frozen site};
    \node[bkpoint,fill=bkdepth3,draw=black!65,line width=0.5pt]
      at (13,-0.6) {$3$};
    \node[bklegend] at (14,-0.6)
      {Locations where $ h<\bar h$.\\$\bar h$: height with no frozen site.};
  \end{scope}

\end{tikzpicture}%
}
\caption{An illustration of the height function and the adjusted boundaries in 2D.}
\label{fig:adjusted-boundaries-two-dimensions}
\end{figure}
\endgroup

\smallskip

Next we prove that $u$ on the boundary $\mathfrak{B}_k$ can be used to reconstruct $u$ on $\mathfrak{D}_k$.
\begin{theorem}\label{thm:reconstruction}
  Take any integer $s\ge 1$ and set $F \subset \Z^d$, and define the sets $\mathfrak{B}_k$ and $\mathfrak{D}_k$ as above.
  Assume \eqref{eq:pointwise-equation} holds on $Q_s$, and assume that all the values of $V_x,\ x\in Q_s$ have been revealed. Take any $k\geq 1$.
Then for each $x\in \mathfrak{D}_k$, $u(x)$ is a linear functional of the initial data $u|_{\mathfrak{B}_k}$, i.e.,
  \begin{equation}\label{eq:reconstruct-formula}
    u(x)=\sum_{y\in \mathfrak{B}_k} \mathscr{W}(x,y;d,\lambda,E, V|_{Q_s}) u(y),
  \end{equation}
  where the coefficient kernel $\mathscr W$ is real-valued and satisfies, uniformly for $E\in I$ and integers $k\ge1$,
  \begin{equation}\label{eq:bound-reconstruction-coefficients}
    \| \mathscr{W} \|_{\ell^{\infty}(\mathfrak{D}_k\times \mathfrak{B}_k)} \leq \exp(C_{\rm recon} s),\quad  \| \partial_E \mathscr{W} \|_{\ell^{\infty}(\mathfrak{D}_k\times \mathfrak{B}_k)} \leq \exp(C_{\rm recon} s).
  \end{equation}
  Moreover, we have 
  \begin{equation}\label{eq:reconstruct-ell-infty-control}
      \| u \|_{\ell^{\infty}(\mathfrak{D}_k)} \leq \exp(C_{\rm recon} s) \|u \|_{\ell^\infty(\mathfrak{B}_k)}.
  \end{equation}
  Here $C_{\rm recon}=C_{\rm recon}(d,I,\lambda)$ is a large constant depending on $d,I$ and $\lambda$. 
\end{theorem}

\begin{proof}
Fix $E\in I$ and $k\ge1$. If $\mathfrak D_k=\varnothing$, there is
nothing to prove. For $x\in\mathbb Z^d$, it is easy to see that
\[
y\in\mathfrak C(x,-e_1)
\quad\Longrightarrow\quad
y\cdot e_1\in\{x\cdot e_1-1,x\cdot e_1-2\}.
\]
For $x\in\mathfrak D_k\setminus\mathfrak B_k$,
the height recurrence gives
\[
\min_{y\in\mathfrak C(x,-e_1)}h(y)
=
h(x)-\mathds{1}[x\notin F].
\]
Moreover,
\[
\begin{cases}
h(x)\ge k,
&x\in\mathfrak D_k\cap F,\\
h(x)\ge k+1,
&x\in\mathfrak D_k\setminus(F\cup\mathfrak B_k).
\end{cases}
\]
Therefore
\[
\min_{y\in\mathfrak C(x,-e_1)}h(y)\ge k,
\qquad
\forall \ x\in\mathfrak D_k\setminus\mathfrak B_k,
\]
which indicates that
\begin{equation}\label{eq:reconstruction-domain-closure}
\mathfrak C(x,-e_1)\subset\mathfrak D_k,
\qquad
\forall \ x\in\mathfrak D_k\setminus\mathfrak B_k.
\end{equation}

The relation \eqref{eq:reconstruction-domain-closure} enables us to execute the following recursive expansion procedure: for each $x\in \mathfrak{D}_k\setminus \mathfrak{B}_k$, using \eqref{eq:fsa-recurrence}, we can write
\[
u(x)=\sum_{y\in \mathfrak{C}(x,-e_1)} {\rm weight}(x,y) \cdot  u(y),
\]
where
\begin{equation}\label{eq:one-step-coefficient}
  {\rm weight}(x,y):=\begin{cases}
        2d-E+\lambda V_{y},& {\rm \ if } \ y=x-e_1;\\
        -1,& {\rm \ if} \ y\in \mathfrak{C}(x,-e_1)\setminus \{x-e_1\}.
  \end{cases}
\end{equation}
In particular, such ${\rm weight}(x,y)$ is independent of $u$ and deterministic, as $V_x,x\in Q_s$ have been revealed. Next, for each $y\in \mathfrak{C}(x,-e_1)\setminus \mathfrak{B}_k$, we keep expanding $u(y)$ through
\[
  u(y)=\sum_{z\in \mathfrak{C}(y,-e_1)} {\rm weight}(y,z) \cdot  u(z).
\]
By repeating this until arriving at $\mathfrak{B}_k$, we can get
\[u(x)=\sum_{y\in \mathfrak{B}_k} \mathscr{W}(x,y) \cdot  u(y),\]
where
\begin{equation}\label{eq:exact-formula-reconstruct-coefficients}
  \mathscr{W}(x,y)=\sum_{P \in  \mathscr P_k(x,y)} {\rm weight}(P),
\end{equation}
\begin{equation}\label{eq:weight-P}
  {\rm weight}(P)=\prod_{1\leq i\leq m} {\rm weight}(x_{i-1},x_i) \quad {\rm for} \ P=(x_0,x_1,\cdots,x_m), 
\end{equation}
and $\mathscr P_k(x,y)$ is the set of cone-type chains in direction $-e_1$ from $x$ to $y$:
\[
\mathscr P_k(x,y)=\bigl\{P=(x_0,\ldots,x_m):\;x_0=x,\; x_m=y,\;
x_{j+1}\in\mathfrak C(x_j,-e_1),\; x_j\in\mathfrak D_k\setminus\mathfrak B_k,\; \forall \, 0\le j<m\bigr\},
\]
For $P$ with length zero, we assign ${\rm weight}(P)=1$. Thus for any $x,y\in\mathfrak B_k$, we have $\mathscr W(x,y)=\mathds{1}[x=y]$.

Since we take $0\leq V_x\leq 1$, the weights in \eqref{eq:one-step-coefficient} are bounded in absolute value by $2d+M_I+\lambda$, where $M_I=\sup_{E\in I}|E|$. Also, since any 
$P$ from $\mathfrak{D}_k$ to $\mathfrak{B}_k$ is contained in $Q_s$ and has length at most $2s$, we have
\[ |{\rm weight}(P)|\leq(2d+M_I +\lambda)^{2s},\quad |\partial_E \ {\rm weight}(P)|\leq 2s(2d+M_I +\lambda)^{2s-1}. \]
Moreover, since each cone has $2d$ many points, we have
\[|\mathscr P_k(x,y)| \leq (2d)^{2s},\]
and \eqref{eq:exact-formula-reconstruct-coefficients} yields 
\begin{align*}
    \| \mathscr{W} \|_{\ell^{\infty}(\mathfrak{D}_k\times \mathfrak{B}_k)} & \leq (2d)^{2s} (2d+M_I+\lambda)^{2s},\\
      \| \partial_E \mathscr{W} \|_{\ell^{\infty}(\mathfrak{D}_k\times \mathfrak{B}_k)} & \leq (2d)^{2s}\cdot 2s(2d+M_I+\lambda)^{2s-1}.
\end{align*}
This gives \eqref{eq:bound-reconstruction-coefficients} by taking $C_{\rm recon}(d,I,\lambda)$ large. Finally, we bound the $\ell^{\infty}$ norm of $u$ by 
  \begin{align*}
  \| u\|_{\ell^{\infty} (\mathfrak{D}_k)} & \leq |\mathfrak{B}_k|\cdot  \| \mathscr{W} \|_{\ell^{\infty}(\mathfrak{D}_k\times \mathfrak{B}_k)} \cdot \| u\|_{\ell^{\infty} (\mathfrak{B}_k)} \\
                         &\leq (2s+1)^d (2d)^{2s} (2d+M_I+\lambda)^{2s}\cdot \| u\|_{\ell^{\infty} (\mathfrak{B}_k)} \leq \exp(C_{\rm recon} s)\| u\|_{\ell^{\infty} (\mathfrak{B}_k)}
  \end{align*}
as $C_{\rm recon}(d,I,\lambda)$ is taken to be large.
\end{proof}

\subsection{Construction of the linear-martingale system}   \label{ssec:clms}

In this subsection, we construct a linear-martingale system from the equation \eqref{eq:pointwise-equation}, with the values of $u$ on an adjusted bowl-shaped boundary as the input.
 
We use the following setup in this subsection. Take $d\ge2$, $N\in\mathbb N$, and $b>0$.
Let $\varepsilon_0$, $L_*(N,\varepsilon)$, and $R_{\rm spa}(b,d,N,\varepsilon)$ be the thresholds in Lemma~\ref{lem:fsa-packing}.
Choose $0<\varepsilon<\varepsilon_0$, then a scale sequence $\boldsymbol\rho=(\rho_1,\ldots,\rho_{n_{\mathrm{lev}}})$ as in Definition~\ref{def:fsa-graded}, with $\rho_1\ge L_*(N,\varepsilon)$. Take $R$ such that
$R\ge R_{\rm spa}(b,d,N,\varepsilon)$ and $bR/\log R\ge40$.
Take a discrete $(N,\varepsilon,\boldsymbol\rho)$-graded set $F$ that is normal in $Q_R$.
Choose integers $r\ge0$ and $\ell\ge1$ satisfying
$\ell\ge bR/\log R$ and $r+4\ell<R/2$,
and set $s=r+3\ell$ and $L=\lfloor\ell/8\rfloor\ge\ell/10$.
Starting from $T_s$, define the height function $h$ and the associated sets $\mathfrak B_k,\mathfrak D_k$ by \eqref{eq:fsa-depth}, \eqref{eq:adjust-boundary}, and \eqref{eq:adjust-region}.

Before the construction, we take the adjusted bowl-shaped boundaries, and select, among those contained in the same annulus, one with few sites where $|u|$ is at least a prescribed threshold.

\begin{lemma}\label{lem:min-bowl}
Under the setup above, the adjusted bowl-shaped boundaries $\mathfrak B_1,\ldots,\mathfrak B_L$ are pairwise disjoint, and
\begin{equation}   \label{eq:BQDQ}
\mathfrak B_1,\ldots,\mathfrak B_L\subset Q_s\setminus(F\cup Q_{r+\ell+1}),\qquad
Q_{r+\ell+1}\subset\mathfrak D_L\subset\cdots\subset\mathfrak D_1\subset Q_s.
\end{equation}
And for every function $u:Q_{r+4\ell}\to\mathbb C$ and every $\tau>0$, there is $1\le k\le L$ satisfying
\begin{equation}\label{eq:min-bowl}
|\{x\in \mathfrak B_k:|u(x)|\ge\tau\}|\le\frac{10}{\ell}|\{x\in Q_{r+4\ell}\setminus F:|u(x)|\ge\tau\}|.
\end{equation}
\end{lemma}

\begin{proof}
Take any $x\in Q_{r+\ell+1}$, and choose a cone-type chain $P=(x_0,\ldots,x_m)$ in direction $e_1$ attaining the minimum in \eqref{eq:height-interpretion}.
By truncating $P$ at its last visit to $T_s$, we may assume that $P\cap T_s=\{x_0\}$.
For every $y\in Q_s\setminus T_s$, we have $\mathfrak C(y,-e_1)\subset Q_s\cup T_s$. Tracing $P$ backwards from $x$ therefore shows that all its vertices other than its root lie in $Q_s\setminus T_s$.
Consequently, $P\subset Q_{s+1}\subset Q_R$. Since its root lies in $T_s$ and each step changes the first coordinate by at most $2$ and every other coordinate by at most $1$, its length satisfies $m\ge\lceil(2\ell-1)/2\rceil=\ell\ge bR/\log R$.
Lemma~\ref{lem:fsa-packing} now gives $|P\setminus F|\ge3m/4$, and hence
\[
  h(x)=|P\setminus(F\cup T_s)|\ge 3m/4-1\ge 3\ell/4-1\ge \ell/2.
\]
This with our choice of $L$ implies \eqref{eq:BQDQ}.

Finally, $\mathfrak{B}_k$ for $1\leq k\leq L$ are pairwise disjoint, by their construction. Thus
\[\sum_{1\leq k\leq L} |\{x\in \mathfrak B_k:|u(x)|\ge\tau\}|\le |\{x\in Q_{r+4\ell}\setminus F:|u(x)|\ge\tau\}|.\]
Applying the pigeonhole principle yields \eqref{eq:min-bowl}.
\end{proof}

We next give the construction.
Take $\lambda>0$ and a compact interval $I\subset\mathbb R$.
Take any $1\le k\le L$, and write $\mathfrak B=\mathfrak B_k$ and $\mathfrak D=\mathfrak D_k$. 
By Lemma \ref{lem:min-bowl},
\[
\mathfrak B\subset Q_s\setminus(F\cup Q_{r+\ell+1}),\qquad
Q_{r+\ell+1}\subset\mathfrak D\setminus\mathfrak B\subset Q_s.
\]
Note that here we do not require any condition like \eqref{eq:min-bowl} for the choice of $k$.

Define the \textit{feasible direction set}
\begin{equation}\label{eq:feasible-direction}
\mathcal V=\{e_1-e_2\}\cup\{e_1+e_j:2\le j\le d\}.
\end{equation}
Take any $\nu\in\mathcal V$.
Take an integer $p\ge 1$, and an enumerated nonempty set $\mathcal A=\{\xi_1,\ldots,\xi_p\}\subset\mathfrak B$.
Take $E\in I$.

Set $\mathcal S=Q_r\setminus F$ and assume $\mathcal S\ne\varnothing$.
Prescribe $v_F\in[0,1]^{F\cap Q_s}$ and set $V_x=v_F(x)$ on $F\cap Q_s$.
On $Q_s\setminus F$, take $V_x$ to be independent Bernoulli$(1/2)$, and condition on their values in $Q_s\setminus(Q_r\cup F)$.
Thus $V$ is fixed on $Q_s\setminus\mathcal S$, and only the potentials in $\mathcal S$ remain random.

\begin{construction}\label{construction:fc-system}
Under the setup and conditioning above, 
the following objects give a linear-martingale system as in Definition~\ref{def:linear-martingale-system}, denoted by $\mathcal M^E_{r,\ell}(\nu,\mathfrak B,\mathcal A)$.
First assume $\nu=e_1+e_2$, without loss of generality. Write
\[
y=(t,x,\zeta)\in\mathbb Z\times\mathbb Z\times\mathbb Z^{d-2},
\] and take lines parallel to $\nu$:
\[
L(\kappa,\zeta')=\{(t,x,\zeta)\in\mathbb Z^d:t-x=\kappa,\ \zeta=\zeta'\}=\{(\kappa+a,a,\zeta')\in\mathbb Z^d:a\in\mathbb Z\}.
\]
We enumerate all pairs $(\kappa,\zeta)$ with $L(\kappa,\zeta)\cap\mathcal S\ne\varnothing$ as $(\kappa_f,\zeta_f)$, $1\le f\le n_{\mathrm{blk}}$, in increasing lexicographic order ($\kappa$ first and then $\zeta$).
Take the source blocks as
\[
r_f=|L(\kappa_f,\zeta_f)\cap\mathcal S|,\qquad
\mathcal J_f=\{(f,j):1\le j\le r_f\},\qquad
\mathcal J=\bigcup_{f=1}^{n_{\mathrm{blk}}}\mathcal J_f.
\]
List the integers $a$ for which $(\kappa_f+a,a,\zeta_f)\in\mathcal S$ as $a_{f,1}<\cdots<a_{f,r_f}$.
Write \[
X_{f,j}:=(\kappa_f+a_{f,j},a_{f,j},\zeta_f),\qquad 1\le j \le r_f.
\]

For each block $f$, take the target set as
\[
b_f=\ell_{\rm tag}:=\lfloor\ell/4\rfloor\ge\ell/8, \qquad \mathcal K_f=\{(f,k)_{\rm tag}:1\le k\le b_f\}.
\]
Set
\begin{equation}\label{eq:fc-original-targets}
a_+(\kappa)=\min(r,r-\kappa),\qquad a_f(m)=a_+(\kappa_f)+m,\quad m\in\Z_{\ge 0}.
\end{equation}
Let $m_{f,1}<\cdots<m_{f,b_f}$ be the first $b_f$ positive integers $m$ such that $(\kappa_f+a_f(m)+1,a_f(m),\zeta_f)\notin F\cap Q_s$, and write
\[
X_{f,k,\rm tag}:=(\kappa_f+a_f(m_{f,k})+1,a_f(m_{f,k}),\zeta_f),\qquad 1\le k\le b_f.
\]
Now we let
\[
\mathcal M^E_{r,\ell}(\nu,\mathfrak B,\mathcal A)
=\{(\mathcal J_f,\mathcal K_f):1\le f\le n_{\mathrm{blk}}\}.
\]
Fix an input body $\mathscr K\subset\mathbb C^p$ and set
\begin{equation}\label{eq:source-block-constants}
c_0=\lambda,\qquad C_0=\exp\bigl(C_{\rm recon}(d,I,\lambda)s\bigr).
\end{equation}
The four parts of Definition~\ref{def:linear-martingale-system} are specified below; their verification follows the construction.
\begin{enumerate}
\item[(1)] \textbf{\emph{Bits and words.}}
We let
\[
\omega_{f,j}=V_{X_{f,j}},\qquad 1\le f\le n_{\mathrm{blk}},\quad 1\le j\le r_f.
\]

\item[(2)] \textbf{\emph{Inputs, source values, target values, and fans.}}
For each input $\vartheta\in\mathscr K$ and word $\omega$, let
\[
u(\xi_i)=\vartheta_i,\quad 1\le i\le p;\qquad
u(x)=0,\quad x\in\mathfrak B\setminus\mathcal A.
\]
Define $u|_{\mathfrak D}$ by the reconstruction map in Theorem~\ref{thm:reconstruction}:
\[
u(x)=\sum_{i=1}^p\mathscr W(x,\xi_i)\vartheta_i,\qquad x\in\mathfrak D,
\]
with given energy $E$ and the potential determined by $\omega$ and the conditioned values.
The source and target values are
\[
S_{f,j}(\omega_{<f};\vartheta):=u(X_{f,j}),\qquad
Y_{f,k}(\omega_{\le f};\vartheta):=u(X_{f,k,\rm tag}).
\]

\item[(3)] \textbf{\emph{Propagation condition with $c_0$.}}
The identity \eqref{eq:propagation-condition} holds with the deterministic coefficients
\[
c_{f,j,k}=\lambda(-1)^{a_f(m_{f,k})-a_{f,j}},
\]
so $|c_{f,j,k}|=\lambda=c_0$.

\item[(4)] \textbf{\emph{Uniform coefficient bounds with $C_0$.}}
The coefficient rows of $S_{f,j}$ and $Y_{f,k}$ satisfy \eqref{eq:coefficient-bound-condition} with the constant $C_0$ in \eqref{eq:source-block-constants}.
\end{enumerate}
For any other feasible direction $\nu$, the objects are the same, up to a permutation or reflection of the coordinates $2,\ldots,d$ in all the data.
\end{construction}

See Figure~\ref{fig:fc-system} for illustrations in dimensions two and three.

\begingroup
\definecolor{fmblue}{HTML}{1766A0}
\definecolor{fmteal}{HTML}{008B82}
\definecolor{fmpurple}{HTML}{8656A3}
\definecolor{fmred}{HTML}{C84427}
\definecolor{fmamber}{HTML}{B77B17}
\definecolor{fmpink}{HTML}{B84B83}
\definecolor{fmolive}{HTML}{71862A}
\definecolor{fmslate}{HTML}{576E86}
\definecolor{fmcyan}{HTML}{2898AE}

\begin{figure}[htbp]
\centering
\begin{minipage}[t]{0.485\linewidth}
\vspace{0pt}
\centering
\resizebox{\linewidth}{!}{%
\begin{tikzpicture}[
  x=3.5mm,y=3.5mm,
  font=\small,
  line cap=round,line join=round,
  fmopen/.style={circle,draw,line width=0.8pt,
    fill=white,inner sep=0pt,minimum size=3.7pt},
  fmsolid/.style={circle,fill,draw=none,
    inner sep=0pt,minimum size=4.2pt},
  fmtext/.style={fill=white,inner sep=1.4pt}
]
\path[use as bounding box] (-11,-10.5) rectangle (11,11.5);

\fill[black!2] (-8.35,-7.35) rectangle (8.35,8.35);
\draw[black!13,line width=0.4pt]
  (-8.35,-7.35) rectangle (8.35,8.35);

\fill[fmblue!3] (-5,-5) rectangle (5,5);
\draw[step=1,black!8,line width=0.25pt] (-5,-5) grid (5,5);
\draw[black!50,dashed,line width=0.7pt] (-5,-5) rectangle (5,5);
\fill[white] (-2,-2) rectangle (2,2);
\draw[step=1,black!13,line width=0.3pt] (-2,-2) grid (2,2);
\draw[black!75,line width=0.9pt] (-2,-2) rectangle (2,2);

\foreach \tt in {-7,-6}{
  \foreach \xx in {-8,...,8}{
    \draw[black!58,fill=white,line width=0.45pt]
      (\xx-0.13,\tt-0.13) rectangle (\xx+0.13,\tt+0.13);
  }
}
\foreach \tt in {-5,...,8}{
  \foreach \xx in {-8,8}{
    \draw[black!58,fill=white,line width=0.45pt]
      (\xx-0.13,\tt-0.13) rectangle (\xx+0.13,\tt+0.13);
  }
}

\foreach \kap/\col in {
 -4/fmslate,-3/fmcyan,-2/fmpurple,-1/fmteal,
 0/fmblue,1/fmamber,2/fmpink,3/fmolive,4/fmred}{
  \pgfmathtruncatemacro{\amin}{max(-2,-2-\kap)}
  \pgfmathtruncatemacro{\amax}{min(2,2-\kap)}
  \draw[\col,line width=0.9pt]
    (\amin,{\kap+\amin}) -- (\amax,{\kap+\amax});
  \foreach \aa in {\amin,...,\amax}{
    \node[fmopen,draw=\col] at (\aa,{\kap+\aa}) {};
  }
  \draw[\col,line width=0.65pt]
    ({\amax+1},{\kap+\amax+2}) --
    ({\amax+2},{\kap+\amax+3});
  \foreach \mm in {1,2}{
    \node[fmsolid,fill=\col]
      at ({\amax+\mm},{\kap+\amax+\mm+1}) {};
  }
}

\draw[fmblue,densely dashed,line width=0.9pt,
      -{Latex[length=1.5mm,width=1mm]}]
  (2,2) -- (2,3) -- (3,4);

\node[fmtext] at (0,6.35) {$Q_{r+\ell+1}$};
\node[fmtext] at (-5.9,7.55) {$\mathfrak D$};
\node[fmtext] at (-3.45,-0.3) {$Q_r$};
\draw[black!55,line width=0.4pt] (-2.85,-0.3) -- (-2.1,-0.3);

\node[anchor=west,fmtext,text=fmblue] at (5.65,3.1) {targets};
\draw[fmblue!65,line width=0.45pt] (5.55,3.45) -- (4.25,4.75);
\node[anchor=east,fmtext,text=fmblue] at (-3.05,-3.15) {source block};
\draw[fmblue!65,line width=0.55pt] (-2.95,-3.05) -- (-1.55,-1.55);

\node[fmtext] at (0,-8.25) {$\mathfrak B$};
\draw[-{Latex[length=1.6mm]},black!75]
  (-9.1,-7.8) -- (-9.1,8.5) node[above] {$t$};
\node[anchor=east] at (-9.4,5.8) {$e_1$};
\draw[-{Latex[length=1.6mm]},black!75]
  (-8.4,-9.15) -- (8.8,-9.15) node[right] {$x$};
\end{tikzpicture}%
}
\par\smallskip
{\small\raggedright
\textbf{(a)}
An illustration of the construction of the linear-martingale system in 2D.
\par}
\end{minipage}
\hfill
\begin{minipage}[t]{0.485\linewidth}
\vspace{0pt}
\centering
\resizebox{\linewidth}{!}{%
\begin{tikzpicture}[
  x=1cm,y=1cm,font=\small,
  line cap=round,line join=round,
  fmopen/.style={circle,draw,line width=0.8pt,fill=white,
    inner sep=0pt,minimum size=3.8pt},
  fmsmall/.style={circle,draw=black!65,fill=white,
    line width=0.6pt,inner sep=0pt,minimum size=3.2pt},
  fmtext/.style={fill=white,inner sep=1.4pt}
]
\path[use as bounding box] (-3.85,-3.675) rectangle (3.85,4.025);

\begin{scope}[
  x={(0cm,0.65cm)},
  y={(0.73cm,0.18cm)},
  z={(0.48cm,-0.25cm)}
]
  \fill[fmblue!5]
    (-2,-2,-2) -- (2,2,-2) -- (2,2,2) -- (-2,-2,2) -- cycle;
  \foreach \aa in {-2,2}{
    \foreach \bb in {-2,2}{
      \draw[black!25,line width=0.45pt] (-2,\aa,\bb) -- (2,\aa,\bb);
      \draw[black!25,line width=0.45pt] (\aa,-2,\bb) -- (\aa,2,\bb);
      \draw[black!25,line width=0.45pt] (\aa,\bb,-2) -- (\aa,\bb,2);
    }
  }

  \foreach \zz/\col in {-1/fmteal,0/fmblue,1/fmamber}{
    \draw[\col,line width=0.9pt] (-2,-2,\zz) -- (2,2,\zz);
    \foreach \aa in {-2,...,2}{
      \node[fmopen,draw=\col] at (\aa,\aa,\zz) {};
    }
  }

  \draw[fmpurple,dashed,line width=0.8pt] (-2,-1,1) -- (1,2,1);
  \foreach \aa in {-1,...,2}{
    \node[fmopen,draw=fmpurple] at ({\aa-1},\aa,1) {};
  }

  \coordinate (fmZ) at (4,3,0);
  \coordinate (fmFzero) at (2,2,0);

  \foreach \tt/\xx/\zz in {
    1/0/0,-1/0/0,0/1/0,0/-1/0,0/0/1,0/0/-1}{
    \draw[black!65,densely dashed,line width=0.7pt]
      (0,0,0) -- (\tt,\xx,\zz);
    \node[fmsmall] at (\tt,\xx,\zz) {};
  }

  \draw[fmred,densely dashed,line width=1.2pt,
        -{Latex[length=1.8mm,width=1.2mm]}]
    (0,0,0) -- (1,0,0);
  \foreach \ii in {0,1,2}{
    \draw[fmred,line width=1.05pt,
          -{Latex[length=1.8mm,width=1.2mm]}]
      ({1+\ii},\ii,0) -- ({2+\ii},{1+\ii},0);
  }
  \foreach \ii in {0,1,2,3}{
    \fill[fmred] ({1+\ii},\ii,0) circle[radius=1.7pt];
  }
  \fill[fmred] (0,0,0) circle[radius=2.7pt];
\end{scope}

\node[anchor=east,fmtext,text=fmred] at (1.72,3.34) {target};
\draw[fmred!75,line width=0.45pt] (1.78,3.30) -- (fmZ);
\node[anchor=west,fmtext,text=fmblue] at (2.02,1.81) {source block};
\draw[fmblue!70,line width=0.45pt] (1.97,1.78) -- (fmFzero);
\end{tikzpicture}%
}
\par\smallskip
{\small\raggedright
\textbf{(b)}
An illustration of source blocks and their propagation in 3D. 
\par}
\end{minipage}
\caption{Illustrations of linear-martingale systems from Construction \ref{construction:fc-system}.}
\label{fig:fc-system}
\end{figure}
\endgroup

To verify the conditions in Definition~\ref{def:linear-martingale-system}, it remains to check:
\begin{enumerate}
\item the values $u(X_{f,j})$ and $u(X_{f,k,\rm tag})$ depend only on $\omega_{<f}$ and $\omega_{\le f}$, respectively;
\item the propagation identity \eqref{eq:propagation-condition} holds with $c_0=\lambda$;
\item the uniform coefficient bounds \eqref{eq:coefficient-bound-condition} hold with $C_0$ as in \eqref{eq:source-block-constants}.
\end{enumerate}

\begin{proof}[Verification of Construction \ref{construction:fc-system}]

\noindent\textbf{(1) dependence of $u(X_{f,j})$ and $u(X_{f,k,\rm tag})$.}
For each block $f$, we first verify that each $X_{f,k,\rm tag}\in Q_{r+\ell+1}\setminus(Q_r\cup F)$.

Indeed, the parameter $a_+(\kappa_f)$ in \eqref{eq:fc-original-targets} is the largest $a$ for which $(\kappa_f+a,a,\zeta_f)\in Q_r$.
Consequently,
\[
(\kappa_f+a_f(m)+1,a_f(m),\zeta_f)\in Q_{r+\ell+1}\setminus Q_r,\qquad 1\le m\le\ell.
\]
The sequence $(\kappa_f+a_f(m)+1,a_f(m),\zeta_f)$, $0\le m\le\ell$, is a cone-type chain of length $\ell\ge bR/\log R$ contained in $Q_R$.
Lemma~\ref{lem:fsa-packing} gives
\[
|\{1\le m\le\ell:(\kappa_f+a_f(m)+1,a_f(m),\zeta_f)\notin F\}|
\ge\tfrac34\ell-1\ge\lfloor\ell/4\rfloor=b_f.
\]
Hence $X_{f,k,\rm tag}\in Q_{r+\ell+1}\setminus(Q_r\cup F)$.

By \eqref{eq:one-step-coefficient}, \eqref{eq:exact-formula-reconstruct-coefficients}, and \eqref{eq:weight-P}, a reconstruction coefficient at a site $y$ can depend on $V_{y'}$ only if a cone-type chain in direction $-e_1$ from $y$ uses the edge from $y'+e_1$ to $y'$.
This is equivalent to
\[
y\in {\rm Cone}(y')=\{(t,x,\zeta):t-(t'+1)\ge |x-x'|+\|\zeta-\zeta'\|_1\},
\]
where $y'=(t',x',\zeta')$.
This restriction holds for every coefficient, independently of the input.

Suppose that $u(X_{f,j})$ depends on the bit $\omega_{f',j'}$.
Then $X_{f,j}\in{\rm Cone}(X_{f',j'})$, which gives
\[
\kappa_f-\kappa_{f'}
\ge1+(a_{f',j'}-a_{f,j})+|a_{f',j'}-a_{f,j}|+\|\zeta_f-\zeta_{f'}\|_1
\ge1.
\]
Thus $\kappa_{f'}<\kappa_f$, and $f'<f$ by the block order.

Similarly, if $u(X_{f,k,\rm tag})$ depends on the bit $\omega_{f',j'}$, then
\[
\kappa_f-\kappa_{f'}
\ge(a_{f',j'}-a_+(\kappa_f)-m_{f,k})+|a_{f',j'}-a_+(\kappa_f)-m_{f,k}|+\|\zeta_f-\zeta_{f'}\|_1
\ge0.
\]
Under this dependence hypothesis, either $\kappa_{f'}<\kappa_f$, or $\kappa_{f'}=\kappa_f$ and $\zeta_{f'}=\zeta_f$ (hence $f'=f$).

\smallskip
\noindent\textbf{(2) the propagation condition.}
Fix a source $(f,j)$, a target $(f,k)_{\rm tag}$ in the same block.
For $e\in\{0,1\}$, let $U^{(e)}$ be the function $u$ defined in (2) of Construction~\ref{construction:fc-system}, with $\omega_{f,j}=e$.
Let $D=U^{(1)}-U^{(0)}$.
Write
\[
X=X_{f,j},\quad X_n=X+n\nu,\quad \forall\, n\in \Z,\quad a=a_{f,j},\quad A:=a_f(m_{f,k}),\quad Z=X_{f,k,\rm tag}.
\]
Then $A\ge a+1$ and $Z=X_{A-a}+e_1$.
Arguing as in (1), we have
\begin{equation}\label{eq:zero-diff-low-line}
D(y)=0,\qquad \forall\, y\in L(\kappa,\zeta)\cap\mathfrak D,\quad\kappa\le\kappa_f.
\end{equation}
Moreover, $D|_{\mathfrak B}=0$. Whenever $X_n+e_1\in\mathfrak D\setminus\mathfrak B$, we have
\begin{multline}\label{eq:fc-difference}
 D(X_n+ e_1)+D(X_{n-1}+e_1) \\ = (2d-E+\lambda V_{X_n})D(X_n) -D(X_n-e_1)-D(X_n+e_2)
                        -\sum_{\sigma=\pm 1 \atop 3\leq i\leq d}D(X_n+\sigma e_i)+\lambda\mathds{1}[n=0]S_{f,j}(\omega_{<f};\vartheta).
\end{multline}
Using \eqref{eq:zero-diff-low-line}, one can remove all the $D$ terms in the right-hand side. Therefore, \eqref{eq:fc-difference} becomes 
\begin{equation}\label{eq:diff-iteration-along-nu}
  D(X_n+ e_1)+D(X_{n-1}+e_1)=\lambda\mathds{1}[n=0]S_{f,j}(\omega_{<f};\vartheta).
\end{equation}
For $0\le n\le A-a$, the points $X_n+e_1$ lie in $Q_{r+\ell+1}\subset\mathfrak D\setminus\mathfrak B$.
Take $n_0$ to be the largest negative integer such that $X_{n_0}+e_1\in \mathfrak B$. Thus $D(X_{n_0}+e_1)=0$.
By applying the recurrence \eqref{eq:diff-iteration-along-nu} to every $n_0<n\le A-a$, we get
\begingroup
\definecolor{btgreen}{HTML}{087F6B}
\definecolor{btred}{HTML}{B6474A}

\begin{figure}[htbp]
\centering
\resizebox{0.55\linewidth}{!}{%
\begin{tikzpicture}[
  x=6.5mm,y=6.5mm,
  font=\small,
  line cap=round,line join=round,
  btlabel/.style={fill=white,inner sep=1.3pt}
]
\path[use as bounding box] (-6.3,-6.65) rectangle (6.9,6.3);

\draw[black!85,line width=0.9pt] (-6,-6) rectangle (6,6);
\draw[black!85,line width=0.85pt] (-3,-3) rectangle (3,3);
\node[anchor=north west] at (6.05,-5.95) {$Q_s$};
\node[anchor=north west,btlabel] at (3.03,-3.02) {$Q_r$};

\begin{scope}
  \clip (0,-3) -- (3,-3) -- (3,0) -- cycle;
  \fill[black!3] (0,-3) -- (3,-3) -- (3,0) -- cycle;
  \foreach \cc in {-3,-2.5,...,3}{
    \draw[black!33,line width=0.4pt]
      (-0.5,{\cc+0.5}) -- (3.5,{\cc-3.5});
  }
\end{scope}

\fill[btgreen!6] (-1,-3) -- (3,1) -- (3,0) -- (0,-3) -- cycle;

\draw[btred,line width=1.05pt]
  (-4.9,5.1) --
  (-4.9,4.55) -- (-4.55,4.55) -- (-4.55,3.7) --
  (-4.85,3.7) -- (-4.85,3.2) -- (-4.4,3.2) --
  (-4.4,2.4) -- (-4.7,2.4) -- (-4.7,1.75) --
  (-4.35,1.75) -- (-4.35,0.9) -- (-4.7,0.9) --
  (-4.7,0.25) -- (-4.45,0.25) -- (-4.45,-0.6) --
  (-4.8,-0.6) -- (-4.8,-1.3) -- (-4.55,-1.3) --
  (-4.55,-2.1) -- (-4.85,-2.1) -- (-4.85,-2.8) --
  (-4.5,-2.8) -- (-4.5,-4.6) --
  (-3.9,-4.6) -- (-3.9,-4.35) -- (-3.15,-4.35) --
  (-3.15,-4.65) -- (-2.65,-4.65) -- (-2.65,-4.25) --
  (-2,-4.25) -- (-2,-4) --
  (-1.35,-4) -- (-1.35,-4.45) -- (-0.55,-4.45) --
  (-0.55,-4.2) -- (0.3,-4.2) -- (0.3,-4.6) --
  (1.05,-4.6) -- (1.05,-4.15) -- (1.8,-4.15) --
  (1.8,-4.45) -- (2.65,-4.45) -- (2.65,-4.2) --
  (3.4,-4.2) -- (3.4,-4.5) -- (4.7,-4.5) --
  (4.7,-3.3) -- (5,-3.3) -- (5,-2.55) --
  (4.65,-2.55) -- (4.65,-1.8) -- (4.35,-1.8) --
  (4.35,-1.1) -- (4.75,-1.1) -- (4.75,-0.2) --
  (4.5,-0.2) -- (4.5,0.55) -- (4.95,0.55) --
  (4.95,1.35) -- (4.55,1.35) -- (4.55,2) --
  (4.85,2) -- (4.85,2.8) -- (4.6,2.8) --
  (4.6,3.5) -- (5,3.5) -- (5,4.3) --
  (4.6,4.3) -- (4.6,5.1);
\node[anchor=west,text=btred] at (-4.5,5.25) {$\mathfrak B$};

\draw[black,line width=1.2pt] (0,-3) -- (3,0);
\node[anchor=west,btlabel] at (3.12,0.05) {$\mathcal J_f$};

\foreach \xx in {0,1,2,3}{
  \draw[btgreen!75,line width=0.7pt]
    ({\xx-1},{\xx-3}) -- (\xx,{\xx-3}) -- (\xx,{\xx-2});
}

\draw[btgreen,line width=1.15pt] (-2,-4) -- (4,2);
\draw[btgreen,line width=1.15pt,
      -{Latex[length=1.8mm,width=1.2mm]}]
  (-0.9,-2.9) -- (-1.55,-3.55);

\foreach \xx in {0,1,2,3}{
  \fill[black] (\xx,{\xx-3}) circle[radius=1.9pt];
  \fill[btgreen] (\xx,{\xx-2}) circle[radius=1.8pt];
}
\fill[btgreen] (-2,-4) circle[radius=2.5pt];
\fill[btgreen] (4,2) circle[radius=2pt];

\node[anchor=north west,btlabel] at (2.13,-1.12) {$X_0$};
\node[anchor=south east,text=btgreen,btlabel]
  at (1.86,0.18) {$X_0+e_1$};
\node[anchor=south east,text=btgreen,btlabel]
  at (0.87,-0.79) {$X_{-1}+e_1$};
\node[anchor=south east,text=btgreen,btlabel] at (3.87,2.15) {$Z$};
\node[anchor=north,text=btgreen,btlabel]
  at (-2,-4.93) {$X_{n_0}+e_1$};
\node[btlabel] at (2.13,-2.55) {$D=0$};

\node[text=btgreen,btlabel] at (-0.65,1.8)
  {$L(\kappa_f+1,\zeta_f)$};
\draw[btgreen!65,line width=0.45pt] (0.75,1.48) -- (2.9,0.9);

\end{tikzpicture}%
}
\caption{An illustration of the arguments leading to \eqref{eq:formula-DXn}.}
\label{fig:backward-shifted-source-block}
\end{figure}
\endgroup
\begin{equation}\label{eq:formula-DXn}
  D(X_n+e_1)=\begin{cases}
     0,& n_0\le n\le -1;\\
     (-1)^n \lambda S_{f,j}(\omega_{<f};\vartheta),& 0\le n\le A-a.
  \end{cases}
\end{equation}
See Figure \ref{fig:backward-shifted-source-block} for an illustration.
Therefore, taking $n=A-a$ above, we exactly get
\[
Y^{(1)}_{f,k}(\widehat\omega_{f,j};\vartheta)
-Y^{(0)}_{f,k}(\widehat\omega_{f,j};\vartheta)
=\lambda(-1)^{a_f(m_{f,k})-a_{f,j}}
 S_{f,j}(\omega_{<f};\vartheta).
\]

\smallskip
\noindent\textbf{(3) coefficient bounds.}
The coefficient rows are given by
\[
s_{f,j,i}(\omega_{<f})=\mathscr W(X_{f,j},\xi_i),\qquad
 y_{f,k,i}(\omega_{\le f})=\mathscr W(X_{f,k,\rm tag},\xi_i),\qquad 1\le i\le p.
\]
We use the reconstruction bound \eqref{eq:reconstruct-ell-infty-control}: for any $z\in\mathfrak D$, 
take $u(\xi_i)\in \{1, -1\}$ to be the sign of $\mathscr W(z,\xi_i)$ (with $u(\xi_i)=1$ if $\mathscr W(z,\xi_i)=0$).
Then \eqref{eq:reconstruct-ell-infty-control} implies that
\[
\sum_{i=1}^p|\mathscr W(z,\xi_i)|\le \exp(C_{\rm recon}s).
\]
It follows that
\[
\|\mathbf s_{f,j}(\omega_{<f})\|_2,\|\mathbf y_{f,k}(\omega_{\le f})\|_2
\le\sup_{z\in\mathfrak D}\sum_{i=1}^p|\mathscr W(z,\xi_i)|
\le \exp(C_{\rm recon}s)=C_0.
\]
This proves part~(4) and completes the verification of the conditions in Definition~\ref{def:linear-martingale-system}.
\end{proof}

\subsection{Multiscale bootstrap and the proof of PDUC}\label{sec:bootstrap-annular}
Take $d\ge 2$ and set $\gamma=(d-1)/d$ in this subsection.
Recall the feasible direction set $\mathcal V$ from \eqref{eq:feasible-direction}.
We next make the final combinatorial preparation, which is essentially the discrete Loomis-Whitney inequality.

\begin{lemma}
\label{lem:many-lattice-lines}
For every finite $S\subset\mathbb Z^d$, there exists a $\nu\in \mathcal V$ such that
\begin{equation}\label{eq:many-lattice-lines}
  \left|\left\{(x+\mathbb Z \nu ): x\in \Z^d, (x+\mathbb Z \nu ) \cap S \neq \emptyset \right\}\right| \geq (|S|/2)^\gamma.
\end{equation}
\end{lemma}

\begin{proof}
The assertion is trivial when $S=\varnothing$, so below we assume $S\ne\varnothing$.
Set
\[
v_1:=e_1+e_2,\quad
v_2:=e_1-e_2,\quad
v_i:=e_1+e_i, \quad 3\le i\le d,
\]
and let $A$ be the matrix with $v_1, \cdots, v_d$ being the columns.
Let $\mathcal L:=A\mathbb Z^d$.
It is straightforward to check that
\[
\mathcal L
=
\bigg\{
x\in\mathbb Z^d:
\sum_{j=1}^d x_j\in 2\mathbb Z
\bigg\}.
\]
Thus $\mathbb Z^d$ is the disjoint union of $\mathcal L$ and $e_1+\mathcal L$, so at least one of $\mcL$ and $\mcL+e_1$ contains at least half of $S$.
Without loss of generality we assume $|\mcL\cap S|\ge |S|/2$.

For each $1\le i \le d$, let $\pi_i : \Z^d\to \Z^{d-1}$ be the projection:
\[
\pi_i(y_1,\ldots,y_d)
:=
(y_1,\ldots,y_{i-1},y_{i+1},\ldots,y_d).
\]
Then the discrete Loomis-Whitney inequality gives
\[
|S\cap \mcL|^{d-1}
\le
\prod_{i=1}^d |\pi_i\circ A^{-1}(S\cap \mcL )|.
\]
Therefore, for some $i\in\{1,\ldots,d\}$,
\[
|\pi_i\circ A^{-1}(S\cap\mcL)|
\ge |S\cap\mcL |^{\gamma}
\ge (|S|/2)^\gamma.
\]
This proves \eqref{eq:many-lattice-lines}.
\end{proof}

Now we prove Theorem \ref{thm:duc-linear}, using the following strategy.
In a large cube, we choose a family of scales, which divide the cube into annuli. Then using Construction \ref{construction:fc-system}, we construct a family of linear-martingale systems at each scale. Applying the propagation lemma (Lemma~\ref{lem:propagation}) to each constructed 
linear-martingale system gives, with overwhelming probability, a lower bound on the number of sites where $|u|$ reaches the threshold at the next scale. Iterating this estimate increases the number of such sites through successive annuli.

We remark that the idea of bootstrapping along annuli has appeared in the study of localization for the ABM. In \cite[Proposition 3.2]{Imb21}, a related idea was used to prove, with high probability, the strict monotonicity of the number of resonant eigenvalues. The annular structure was made more explicit in \cite[Section 4]{LSZ26}, where the authors employed an idea similar to that of \cite[Proposition 3.2]{Imb21} to establish a Wegner estimate for the corresponding model when $d \geq 4$.

\begin{proof}[Proof of Theorem \ref{thm:duc-linear}]
We take $\varepsilon_0$ and $L_*(N,\varepsilon)$ as in Lemma \ref{lem:fsa-packing}. 
Define
\[J =\min\{j\ge1:\gamma^j\log R\le1\};
 \quad w_j =\gamma^{J-1-j},\quad 0\le j<J.\]
 Then $J$ is of order $\log\log R$ as $R\to\infty$, with $d$ fixed.
We consider the following scales:
\[
\begin{aligned}
r_0=\left\lfloor\frac{Rw_0}{64d}\right\rfloor,
\qquad
\ell_j=\left\lfloor\frac{Rw_j}{64d}\right\rfloor,
\qquad
r_{j+1}=r_j+4\ell_j.
 \end{aligned}
\]
For $R$ large enough (depending only on $d$), we can ensure that
\begin{equation}\label{eq:fsa-scale-bounds}
 \frac1{\log R}<w_0\le\frac1{\gamma \log R},\qquad
 \frac{Rw_j}{128d}\le\ell_j\le\frac{Rw_j}{64d},\qquad
 r_J<R/4.
\end{equation}
For each $r_i,\ell_i$, take $s_i=r_i+3\ell_i$ and $b=\frac{1}{4096 d}$.
Such $r_i, \ell_i, s_i$ then satisfy the scale condition at the beginning of Section~\ref{ssec:clms}.

Recall that $C_{\rm recon}=C_{\rm recon}(d,I,\lambda)$ is the constant in Theorem~\ref{thm:reconstruction}.
Set $C_{\rm dec}=C_{\rm recon}/8+8$.
For $0\leq j< J$, define the value thresholds by
\begin{equation}\label{eq:fsa-heights}
H_0=\exp(-R),\quad
H_{j+1}=H_j\exp(-C_{\rm dec}Rw_j),\quad
\delta_j=H_j/2,\quad
\tau_j=H_j\exp(-C_{\rm dec}Rw_j/2).
\end{equation}
These choices and \eqref{eq:fsa-scale-bounds} ensure that 
\begin{equation}\label{eq:low-bound-Hj}
  H_j\geq \exp \bigg( -C_{\rm dec} R\sum_{j'=0}^{J-1} w_{j'} \bigg)H_0 \geq \exp\bigl(-(dC_{\rm dec}+1)R\bigr),
\end{equation}
for each $0\leq j\leq J$. Likewise, $\tau_j\ge \exp\bigl(-(dC_{\rm dec}+1)R\bigr)$ for each $0\le j<J$.

\smallskip
\noindent\textbf{{Step 1: probabilistic estimates.}}
Fix $E_0\in I$. We first construct the event $\mcE_R(F,v_F,E_0)$, on which the propagation estimates hold uniformly over all inputs.

For each $0\le j <J$, Lemma~\ref{lem:min-bowl} gives $L_j=\lfloor\ell_j/8\rfloor\ge\ell_j/10$ adjusted bowl-shaped boundaries
\[
\mathfrak B^{(j)}_1,\ldots,\mathfrak B^{(j)}_{L_j}
\subset Q_{s_j}\setminus(F\cup Q_{r_j+\ell_j+1}),
\]
and we write $\mathfrak D_1^{(j)},\ldots, \mathfrak D_{L_j}^{(j)}$ for the corresponding sets defined by \eqref{eq:adjust-region} for $Q_{s_j}$.

We note that the set $Q_{r_j}\setminus F$ is nonempty. Indeed, since $r_0\ge bR/\log R$, Lemma~\ref{lem:fsa-packing} applied to the straight cone-type chain $(0,e_1,\ldots,r_0e_1)\subset Q_{r_0}$ gives $Q_{r_0}\setminus F\ne\varnothing$; and we also have $Q_{r_0}\subseteq Q_{r_j}$.
Then for every $1\le k\le L_j$, every $\nu\in\mathcal V$, and every nonempty $\mcA\subseteq\mathfrak B_k^{(j)}$, Construction~\ref{construction:fc-system} gives a system $\mcM^{E_0}_{r_j,\ell_j}(\nu,\mathfrak B_k^{(j)},\mcA)$ with
\[
c_0=\lambda,\qquad C_{0,j}=\exp(C_{\rm recon}Rw_j/16)\ge \exp(C_{\rm recon}s_j).
\]
Here the inequality is due to
\begin{equation}   \label{eq:srbd}
s_j=r_j+3\ell_j
\le\frac{Rw_j}{64d}\bigl(1+4(d-1)+3\bigr)
=\frac{Rw_j}{16},
\end{equation}
since $\sum_{i=0}^{j-1}w_i\le(d-1)w_j$.

The index set of these systems is
\begin{equation}\label{eq:index-set-of-all-system}
{\rm SysIndex}=\left\{\mathfrak{i}=(j,k,\nu,\mcA):
\begin{array}{l}
j\in\{0,\ldots,J-1\},\quad k\in\{1,\ldots,L_j\},\\
\nu\in\mathcal V,\quad\varnothing\ne\mcA\subseteq\mathfrak B_k^{(j)}
\end{array}\right\}.
\end{equation}
Write $\mcM^{E_0}(\mathfrak{i})$ for the corresponding system and
$
N_{\mathfrak{i}}=\sum_{f=1}^{n_{\mathrm{blk}}}(r_f+b_f)
$
for its total number of source and target labels in the system.
The sources enumerate $Q_{r_j}\setminus F$, and the target sites are distinct and lie in $Q_{r_j+\ell_j+1}\setminus(Q_{r_j}\cup F)$, so
\begin{equation}\label{eq:system-size-bound}
N_{\mathfrak{i}}\le |Q_{r_j+\ell_j+1}|\le |Q_R|\le (2R+1)^d.
\end{equation}

Take another parameter
$\eta=\frac1{80}\left(\frac3{8192d}\right)^\gamma$.
For each $\mathfrak{i}=(j,k,\nu,\mcA)$, we will apply Lemma~\ref{lem:propagation} to $\mcM^{E_0}(\mathfrak{i})$ with
\[
\mathscr K=B_{R_{\rm in}}\subset\C^{|\mcA|},\qquad R_{\rm in}=\exp\bigl((K+1)R\bigr),\qquad
\delta=\delta_j,\quad\tau=\tau_j,
\]
and probability parameter $L=L(\mathfrak{i})$, where
\begin{equation}\label{eq:Tp}
L(\mathfrak{i})=(d+1)|\mcA|\log R+(d+3)\log R
+\eta\left(\frac{R}{\log R}\right)^\gamma.
\end{equation}
Here the target-count parameter $B$ in Lemma~\ref{lem:propagation} is $\lfloor\ell_j/4\rfloor$, the size of each target set in the system.

We first check the parameter conditions for $c_0, C_{0,j}, \delta_j, \tau_j$, required by Lemma~\ref{lem:propagation}.
The condition $\tau_j/C_{0,j}\le R_{\rm in}$ follows from $\tau_j\le H_j\le1$, $C_{0,j}\ge1$, and $R_{\rm in}\ge1$.
For the other condition $\delta_j\min(1,c_0/2)>\sqrt{|\mcA|(2|\mcA|-1)}\tau_j$, 
we have
\begin{equation}\label{eq:check-ke-lemma-parameter}
\log\left(\frac{\delta_j\min(1,c_0/2)}{\sqrt{|\mcA|(2|\mcA|-1)}\tau_j}\right)
=\frac{C_{\rm dec}}2Rw_j
-\log\left(\frac{2\sqrt{|\mcA|(2|\mcA|-1)}}{\min(1,\lambda/2)}\right)
>\frac{C_{\rm dec}}3Rw_j>0,
\end{equation}
where the first inequality follows from $|\mcA|\le |Q_R|\le (2R+1)^d$ and $Rw_j\ge R/\log R$, with $R$ taken sufficiently large.

Lemma~\ref{lem:propagation} now gives an event $\mathcal E(\mathfrak{i})$ satisfying
\begin{equation}\label{eq:prob-of-Ui}
\P\bigl[\mathcal E(\mathfrak{i})^c\bigr]
\le \exp\bigl(-L(\mathfrak{i})\bigr),
\end{equation}
on which, simultaneously for all $\vartheta\in\mathscr K$,
\begin{equation}\label{eq:event-discribe-out-of-Ui}
2M_{\tau_j}(\mathfrak{i},\vartheta)
\ge\frac{1-\exp(-1)}2\lfloor\ell_j/4\rfloor P_{\delta_j}(\mathfrak{i},\vartheta)-\lfloor\ell_j/4\rfloor\left(m(\mathfrak{i})\bigl(1+\log(N_{\mathfrak{i}}+1)\bigr)+L(\mathfrak{i})\right).
\end{equation}
Here $M_{\tau_j}(\mathfrak{i},\cdot)$ and $P_{\delta_j}(\mathfrak{i},\cdot)$ are the counts in \eqref{eq:fan-outputcount} and \eqref{eq:def-P-delta}, and $m(\mathfrak{i})$ is defined by \eqref{eq:fan-budget}, all for the system $\mcM^{E_0}(\mathfrak{i})$.
We can further estimate $m(\mathfrak{i})$: by $w_j\le1$ and the lower bound for $\tau_j$ following \eqref{eq:low-bound-Hj}, we have
\[
\log(C_{0,j}R_{\rm in}/\tau_j)
=C_{\rm recon}Rw_j/16+(K+1)R+\log(1/\tau_j)
\le(C_{\rm recon}/16+K+2+dC_{\rm dec})R.
\]
Combining this numerator bound with \eqref{eq:check-ke-lemma-parameter}, we get
\begin{equation}\label{eq:fsa-record-bound}
m(\mathfrak{i})
\le\frac{3(C_{\rm recon}/16+K+2+dC_{\rm dec})}{C_{\rm dec}}\cdot\frac{|\mcA|}{w_j}.
\end{equation}

We now take a union bound over all $\mathfrak i\in \rm SysIndex$. For each fixed $j,k$, the binomial identity gives
\[
\sum_{\varnothing\ne\mcA\subseteq\mathfrak B_k^{(j)}}R^{-(d+1)|\mcA|}
=\bigl(1+R^{-(d+1)}\bigr)^{|\mathfrak B_k^{(j)}|}-1
\le\exp\left(\frac{|Q_R|}{R^{d+1}}\right)-1
\le1
\]
for sufficiently large $R$, since $|Q_R|\le (2R+1)^d$.
As $|\mathcal V|=d$ and $L_j\le\ell_j/8$, it follows that
\begin{equation}\label{eq:number-sys-p}
\sum_{\mathfrak{i}=(j,k,\nu,\mcA)\in{\rm SysIndex}}R^{-(d+1)|\mcA|}
\le d\sum_{j=0}^{J-1}L_j\le R^2.
\end{equation}
Set
\[
\mcE_R(F,v_F,E_0)=\bigcap_{\mathfrak{i}\in{\rm SysIndex}}\mathcal E(\mathfrak{i}).
\]
By  \eqref{eq:Tp}, \eqref{eq:prob-of-Ui}, and \eqref{eq:number-sys-p},
\begin{equation}\label{eq:prob-absorb-Tp}
\begin{aligned}
\P[\mcE_R(F,v_F,E_0)^c]
&\le\sum_{\mathfrak{i}\in{\rm SysIndex}}\exp\bigl(-L(\mathfrak{i})\bigr)\\
&=R^{-(d+3)}\exp\left(-\eta\left(\frac{R}{\log R}\right)^\gamma\right)
\sum_{\mathfrak{i}=(j,k,\nu,\mcA)\in{\rm SysIndex}}R^{-(d+1)|\mcA|}\\
&\le R^{-(d+1)}\exp\left(-\eta\left(\frac{R}{\log R}\right)^\gamma\right).
\end{aligned}
\end{equation}
On $\mcE_R(F,v_F,E_0)$, combining \eqref{eq:event-discribe-out-of-Ui} with \eqref{eq:Tp}, \eqref{eq:fsa-record-bound}, the bound \eqref{eq:system-size-bound}, and $\ell_j/5\le\lfloor\ell_j/4\rfloor\le\ell_j$ gives, for each $\mathfrak{i}$ and $\vartheta\in\mathscr K$,
\begin{equation}\label{eq:propagation-inequality-for-bootstrap}
2M_{\tau_j}(\mathfrak{i},\vartheta)\ge\frac{\ell_j}{20}P_{\delta_j}(\mathfrak{i},\vartheta)-C_*\ell_j\frac{|\mcA|}{w_j}\log R-\eta\ell_j\left(\frac{R}{\log R}\right)^\gamma,
\end{equation}
for some $C_*=C_*(d,\lambda,I,K)>0$.

\smallskip
\noindent\textbf{{Step 2: bootstrap along annuli.}}
We now take $V|_{Q_{R-1}\setminus F}$ that is in $\mcE_R(F,v_F,E_0)$, and we next prove \eqref{eq:fsa-quadratic-count} for any $E\in I$ and $u:Q_R\to\C$ satisfying \eqref{eq:pointwise-equation} in $Q_{R-1}$, with
\[
u(0)=1,\qquad
|E-E_0|\le \exp(-CR),\qquad
\|u\|_{\ell^\infty(Q_R)}\le \exp(KR).
\]
For such $u$, define its large-value sets and their cardinalities by
\[
 S_j=\{x\in Q_{r_j}\setminus F:|u(x)|\ge H_j\},\qquad m_j=|S_j|,\quad 0\leq j\leq J.
\]
Since those $r_j$ increase in $j$ and the thresholds $H_j$ decrease in $j$, we have $S_0\subset\cdots\subset S_J$. In particular, $m_j\ge m_0$ for $0\le j\le J$.

We first establish the initial bound for $m_0$.
Proposition \ref{prop:cone-property} supplies a cone-type chain
of length
$
\lfloor r_0/4\rfloor
$
from the origin inside $Q_{r_0}$, along which
\[
|u(x)|\ge \exp\bigl(-C_{\mathrm{cone}}\lfloor r_0/4\rfloor\bigr).
\]
For sufficiently large $R$, since
\[
\lfloor r_0/4\rfloor
\ge \frac{R}{1024d\log R}
\ge \frac{bR}{\log R},
\]
Lemma~\ref{lem:fsa-packing} shows that at least
$3\lfloor r_0/4\rfloor/4\geq 3R/(4096d\log R)$ vertices of this cone-type chain are outside $F$.
Moreover, when $R$ is large enough (depending on $d, I, \lambda$), we have $C_{\mathrm{cone}}\lfloor r_0/4\rfloor \le R$,
so $|u(x)|\ge H_0=\exp(-R)$ along the chain.
Consequently,
\begin{equation}\label{eq:fsa-seed}
m_0\ge\frac{3R}{4096d \log R}.
\end{equation}

We then propagate this initial bound through the scales, as illustrated in Figure~\ref{fig:free-site-annular-bootstrap}.
For each $0\leq j<J$, by Lemma \ref{lem:min-bowl}, there is an index $1\leq \iota_j \leq L_j$ such that 
\[
|\{x\in \mathfrak{B}^{(j)}_{\iota_j}:|u(x)|\ge H_{j+1}\}|\le\frac{10}{\ell_j}|\{x\in Q_{r_{j+1}}\setminus F:|u(x)|\ge H_{j+1}\}|=10m_{j+1}/\ell_j.
\]
Let 
\[\mcA= \{x\in \mathfrak{B}^{(j)}_{\iota_j}:|u(x)|\ge H_{j+1}\},\quad p=|\mathcal A|\le 10m_{j+1}/\ell_j.\]
We have $p\ge1$, since otherwise, 
\[
1=|u(0)|\le \exp(C_{\rm recon}s_j)\|u\|_{\ell^\infty(\mathfrak B_{\iota_j}^{(j)})}
<\exp(C_{\rm recon}Rw_j/16)H_{j+1}=H_j\exp\bigl(-(C_{\rm dec}-C_{\rm recon}/16)Rw_j\bigr)<1,
\]
leading to a contradiction. 
Here the first inequality is by \eqref{eq:reconstruct-ell-infty-control}, the second inequality uses \eqref{eq:srbd}, and the last inequality is due to $H_j\le1$ and $C_{\rm dec}=C_{\rm recon}/8+8>C_{\rm recon}/16$.

Write $\mcA=\{\xi_1,\ldots,\xi_p\}$, and take $\vartheta_i=u(\xi_i)$ for $1\le i\le p$. Such $\vartheta$ lies in the input body $\mathscr K=B_{R_{\rm in}}\subset \C^p$, since 
\begin{equation}\label{eq:fsa-mask}
\|\vartheta\|_2
\le \sqrt p\,\|u\|_{\ell^\infty(Q_R)}
\le \sqrt p\,\exp(KR)
\le R_{\mathrm{in}}=\exp\bigl((K+1)R\bigr).\end{equation}
Moreover, applying Lemma \ref{lem:many-lattice-lines} to $S_j$ gives a $\nu\in \mathcal{V}$ such that 
\begin{equation}\label{eq:many-lattice-lines-apply}
    \left|\left\{x+\mathbb Z \nu :x\in \Z^d,\  (x+\mathbb Z \nu ) \cap S_j \neq \emptyset \right\}\right| \geq (m_j/2)^\gamma.
\end{equation}
Take the index $\hat{\mathfrak{i}}=(j,\iota_j,\nu,\mcA)$. We next express \eqref{eq:propagation-inequality-for-bootstrap} for this index and input $\vartheta$ in terms of the function reconstructed at energy $E_0$.

\begingroup
\definecolor{bootinner}{HTML}{356EAC}
\definecolor{bootmiddle}{HTML}{C98432}
\definecolor{bootouter}{HTML}{318778}

\begin{figure}[htbp]
\centering
\resizebox{0.55\linewidth}{!}{%
\begin{tikzpicture}[
  x=9mm,y=9mm,font=\small,
  line cap=round,line join=round,
  bootarrow/.style={line width=1.3pt,
    -{Latex[length=3.1mm,width=2.2mm]}},
  bootlabel/.style={inner sep=2pt},
  boottrace/.style={line width=0.75pt,dash pattern=on 3pt off 2.5pt}
]

\fill[bootouter!10] (-3.4,-3.4) rectangle (3.4,3.4);
\fill[bootmiddle!13] (-2.25,-2.25) rectangle (2.25,2.25);
\fill[bootinner!10] (-1.15,-1.15) rectangle (1.15,1.15);

\draw[bootouter,line width=0.95pt]
  (-3.4,-3.4) rectangle (3.4,3.4);
\draw[bootmiddle,line width=0.95pt]
  (-2.25,-2.25) rectangle (2.25,2.25);
\draw[bootinner,line width=0.95pt]
  (-1.15,-1.15) rectangle (1.15,1.15);

\draw[boottrace,bootmiddle!85!black]
  (-1.77,1.91) -- (-1.77,1.02) -- (-1.64,1.02) --
  (-1.64,0.35) -- (-1.80,0.35) -- (-1.80,-0.60) --
  (-1.67,-0.60) -- (-1.67,-1.72) --
  (-0.65,-1.72) -- (-0.65,-1.85) -- (0.27,-1.85) --
  (0.27,-1.69) -- (1.76,-1.69) --
  (1.76,-0.52) -- (1.62,-0.52) -- (1.62,0.36) --
  (1.80,0.36) -- (1.80,1.20) -- (1.67,1.20) -- (1.67,1.91);
\node[text=bootmiddle!85!black,fill=bootmiddle!13,
      font=\scriptsize,inner sep=1.2pt]
  at (-1.77,1.40)
  {$\mathfrak B^{(j)}_{\iota_j}$};

\draw[boottrace,bootouter!85!black]
  (-2.90,3.05) -- (-2.90,1.94) -- (-2.74,1.94) --
  (-2.74,0.75) -- (-2.92,0.75) -- (-2.92,-0.48) --
  (-2.78,-0.48) -- (-2.78,-1.62) -- (-2.94,-1.62) --
  (-2.94,-2.86) -- (-1.56,-2.86) -- (-1.56,-2.72) --
  (-0.32,-2.72) -- (-0.32,-2.94) -- (1.03,-2.94) --
  (1.03,-2.80) -- (2.87,-2.80) --
  (2.87,-1.33) -- (2.72,-1.33) -- (2.72,-0.14) --
  (2.93,-0.14) -- (2.93,1.18) -- (2.77,1.18) --
  (2.77,2.15) -- (2.92,2.15) -- (2.92,3.05);
\node[text=bootouter!85!black,fill=bootouter!10,
      font=\scriptsize,inner sep=1.2pt]
  at (-2.83,1.40)
  {$\mathfrak B^{(j+1)}_{\iota_{j+1}}$};

\draw[bootarrow,bootmiddle]
  (-0.55,-0.48)
  .. controls (-0.78,-0.95) and (-1.10,-1.35) .. (-1.47,-1.50);
\draw[bootarrow,bootouter]
  (-1.92,-1.98)
  .. controls (-2.10,-2.25) and (-2.40,-2.50) .. (-2.62,-2.57);

\node[bootlabel,text=bootinner] at (0,0) {$Q_{r_j}$};
\node[bootlabel,text=bootmiddle,fill=white]
  at (1.30,-2.25) {$Q_{r_{j+1}}$};
\node[bootlabel,text=bootouter,fill=white]
  at (2.40,-3.40) {$Q_{r_{j+2}}$};


\end{tikzpicture}%
}
\caption{An illustration of the annular bootstrap across multiple scales. }
\label{fig:free-site-annular-bootstrap}
\end{figure}
\endgroup

Define $\tilde u_1:\mathfrak D_{\iota_j}^{(j)}\to\C$ by the reconstruction map \eqref{eq:reconstruct-formula} at energy $E$, with boundary values
\begin{equation}\label{eq:u_1-boundary-data}
\tilde u_1(x)=0,\quad x\in\mathfrak B_{\iota_j}^{(j)}\setminus\mcA;
\qquad \tilde u_1(\xi_i)=u(\xi_i)=\vartheta_i,\quad 1\le i\le p.
\end{equation}
Define $\tilde u_2:\mathfrak D_{\iota_j}^{(j)}\to\C$ by the same reconstruction map, with the same boundary data and potential, but with $E$ replaced by $E_0$.
The source and target values of $\mcM^{E_0}(\hat{\mathfrak{i}})$ at input $\vartheta$ are the values of $\tilde u_2$ at the corresponding sites.
Consequently, \eqref{eq:propagation-inequality-for-bootstrap} gives
\begin{multline}\label{eq:propagation-reconstructed-values}
2\left|\left\{x\in Q_{r_j+\ell_j+1}\setminus F:|\tilde u_2(x)|>\tau_j\right\}\right|
\ge\frac{\ell_j}{20}\left|\left\{x+\mathbb Z\nu:x\in Q_{r_j}\setminus F,\ |\tilde u_2(x)|\ge\delta_j\right\}\right|\\
-C_*\ell_j\frac{p}{w_j}\log R
-\eta\ell_j\left(\frac{R}{\log R}\right)^\gamma.
\end{multline}

We next compare $\tilde u_2$ with $u$.
For the truncation to $\tilde u_1$, we have
\begin{multline}\label{eq:diff-solution-1}
\|u-\tilde u_1\|_{\ell^\infty(Q_{r_j+\ell_j+1})}
\le \exp(C_{\rm recon}s_j)\|u-\tilde u_1\|_{\ell^\infty(\mathfrak B_{\iota_j}^{(j)})}\\
\le \exp(C_{\rm recon}s_j)H_{j+1}
=\exp(C_{\rm recon}s_j-C_{\rm dec}Rw_j/2)\tau_j
\le \exp(-4Rw_j)\tau_j\le\tau_j/8.
\end{multline}
The first inequality follows from \eqref{eq:reconstruct-ell-infty-control};
the third inequality uses \eqref{eq:srbd} and $C_{\rm dec}=C_{\rm recon}/8+8$; and the last inequality uses that $w_j\ge 1/\log R$.
On the other hand, the energy-derivative bound in \eqref{eq:bound-reconstruction-coefficients} and the Cauchy--Schwarz inequality give
\begin{multline}\label{eq:diff-solution-2}
  \|\tilde u_1-\tilde u_2 \|_{\ell^\infty(Q_{r_j+\ell_j+1})} \leq |E-E_0|\sqrt p\,\exp(C_{\rm recon}s_j)\|\vartheta\|_2 \\
     \le \exp(-CR) (2R+1)^{d/2} \exp(C_{\rm recon}Rw_j/16) R_{\rm in} \le
\exp\bigl(-(C-C_{\rm recon}/16-K-2)R\bigr)
\end{multline}
where the second inequality again uses \eqref{eq:srbd}.
By choosing $C$ in the theorem statement such that $C\ge C_{\rm recon}/16+K+2+2(dC_{\rm dec}+1)$,
the last expression is at most
$\exp\bigl(-2(dC_{\rm dec}+1)R\bigr)\le\tau_j/8$.
Thus we get
\begin{equation}\label{eq:total-diff-solution}
  \| u-\tilde u_2\|_{\ell^\infty(Q_{r_j+\ell_j+1})}\leq \tau_j/4.
\end{equation}

We now transfer the bound \eqref{eq:propagation-reconstructed-values} to $u$.
Since $H_j-\tau_j/4>\delta_j$, \eqref{eq:total-diff-solution} implies that every $x\in S_j$ satisfies $|\tilde u_2(x)|\ge\delta_j$.
Together with \eqref{eq:many-lattice-lines-apply}, this gives
\begin{equation}\label{eq:recontrol-1}
\left|\left\{x+\mathbb Z\nu:x\in Q_{r_j}\setminus F,\ |\tilde u_2(x)|\ge\delta_j\right\}\right|
\ge\left|\{x+\mathbb Z\nu:x\in S_j\}\right|
\ge(m_j/2)^\gamma.
\end{equation}
On the other hand, every $x\in Q_{r_j+\ell_j+1}\setminus F$ with $|\tilde u_2(x)|>\tau_j$ satisfies
$|u(x)|>3\tau_j/4>H_{j+1}$ by \eqref{eq:total-diff-solution}.
Since $Q_{r_j+\ell_j+1}\subset Q_{r_{j+1}}$, all such points belong to $S_{j+1}$, and hence
\begin{equation}\label{eq:recontrol-2}
\left|\left\{x\in Q_{r_j+\ell_j+1}\setminus F:|\tilde u_2(x)|>\tau_j\right\}\right|\le m_{j+1}.
\end{equation}
Substituting \eqref{eq:recontrol-1} and \eqref{eq:recontrol-2} into \eqref{eq:propagation-reconstructed-values}, we obtain the key recurrence
\begin{equation}\label{eq:baby-recurrence}
   2m_{j+1}\ge \frac{\ell_j}{20}(m_j/2)^\gamma
 -C_* \ell_j \frac p{w_j}\log R
 -\eta\ell_j\left(\frac{R}{\log R}\right)^\gamma.
\end{equation}
Since $m_j\ge m_0$, the choice of $\eta$ and \eqref{eq:fsa-seed} give
\begin{equation}\label{eq:seed-big-absorb}
\eta\left(\frac{R}{\log R}\right)^\gamma
\le\frac1{40}(m_j/2)^\gamma.
\end{equation}
Substituting this bound and $p\le10m_{j+1}/\ell_j$ into \eqref{eq:baby-recurrence}, and using \eqref{eq:fsa-scale-bounds} and $w_j\le1$, gives
\begin{equation}\label{eq:fsa-growth}
m_{j+1}\gtrsim_{d,\lambda,I,K}\frac{Rw_j^2}{\log R}m_j^\gamma,
\qquad 0\le j<J,
\end{equation}
with an implicit constant depending only on $d,\lambda,I,K$.
Iterating this bound, using \eqref{eq:fsa-seed} and $w_j=\gamma^{J-1-j}$, yields
\begin{equation}\label{eq:fsa-iteration}
m_J\gtrsim_{d,\lambda,I,K}
\left(\frac R{\log R}\right)^{d-(d-1)\gamma^J}
\gamma^{2\sum_{k=0}^{J-1}k\gamma^k}
\gtrsim_d\frac{R^d}{(\log R)^d}.
\end{equation}
Here the last inequality follows from $\gamma^J\log R\le1$ and $\sum_{k=0}^{J-1}k\gamma^k\le\gamma/(1-\gamma)^2$.
Since
$Q_{r_J}\subset Q_R$ and $H_J\ge \exp\bigl(-(dC_{\rm dec}+1)R\bigr)$ by \eqref{eq:low-bound-Hj}, we get \eqref{eq:fsa-quadratic-count}.
\end{proof}

\section{Wegner estimate and multiscale analysis}\label{sec:WegnerMSA}

In this section, we use the PDUC Theorem~\ref{thm:duc-linear} proved in the previous section to derive localization results.
The framework is now relatively standard \cite{BK05,DS20,LZ22,LSZ26b}, so we choose to be brief here and omit some routine steps.

\subsection{Wegner estimate}
We first establish a Wegner estimate for the ABM in $d\ge2$.
Such estimates control the spectrum of a finite dimensional operator in a given interval, and have been key ingredients in multiscale proofs of localization since the 80s~\cite{FS83,frohlich1985constructive,CKM87}.
For potential distributions with sufficient quantitative regularity, such as a bounded density or a H\"older-continuous distribution function, Wegner estimates can be obtained through eigenvalue-variation and spectral-averaging arguments;
but the usual Wegner estimates do not apply to Bernoulli potentials.
Following the arguments in \cite{Bou04,BK05,DS20,LZ22}, we instead combine the PDUC (Theorem~\ref{thm:duc-linear}), eigenvalue variation, and Sperner-type combinatorial estimates to obtain a weaker Wegner estimate that is sufficient for the MSA.

Below, for any subset $\Lambda\subset\Z^d$, let $H_\Lambda$ denote the restriction of $H$ from \eqref{eq:ABM-schrodinger-operator} to $\Lambda$ with Dirichlet boundary conditions: i.e., let $H_\Lambda u=(H\widetilde u)|_\Lambda$, where $\widetilde u$ is the extension of $u$ by zero outside $\Lambda$.
We retain the notation $Q_r(a)=\{x\in\Z^d:\|x-a\|_\infty\le r\}$ and $Q_r=Q_r(0)$ from Section~\ref{sec:graded-set}.

As a preparation, we give a geometric lemma, which is needed due to the $\ell^\infty$ norm upper bound condition in Theorem~\ref{thm:duc-linear}.

\begin{lemma}\label{lem:duc-block-lemma}
For dimension $d\ge2$ and disorder strength $\lambda>0$, and any compact interval $I\subset\R$ and $N\in \N$,
there exist $C_{\mathrm{geo}}=C_{\mathrm{geo}}(d,N), C_{\mathrm{esc}}=C_{\mathrm{esc}}(d,\lambda,I,N)>0$ such that the following holds.
Take any $L,n\in \N$ with $n\ge4$ and $L\ge C_{\mathrm{geo}}n$, and set $Q=Q_L$.
Let $\mcD=\bigcup_{j=1}^{J}Q_n(z^{(j)})$, where $0\le J\le N$, $z^{(j)}\in\Z^d$, and $Q_n(z^{(j)})\subset Q$ for every $j$.
If $E\in I$ and $u:Q\to\C$ is a nonzero function satisfying $H_Qu=Eu$, then there exists $a\in Q$ such that
\[
 Q_{n}(a)\subset Q\setminus \mcD,\quad
 |u(a)|\ge \exp(-C_{\mathrm{esc}}n)\|u\|_{\ell^\infty(Q)}.
\]
\end{lemma}
\begin{proof}
Take $C_{\rm geo}=2^{d+4}(N+1)$.
Recall that if $(x^{(0)},x^{(1)},\ldots,x^{(m)})$ is a cone-type chain in direction $\sigma e_i$, we have
\begin{equation}\label{eq:dirft-of-cone-chain}
    1\leq \sigma (x^{(s)}_i-x^{(s-1)}_i) \leq 2, \ |x^{(s)}_j-x^{(s-1)}_j|\leq 1, \ \forall \ j\neq i,
\end{equation}
for any $1\leq s\leq m$.

Now choose $\hat{a}\in Q$ with $|u(\hat{a})|=\|u\|_{\ell^\infty(Q)}$, and assume without loss of generality that $\hat{a}_1\le0$.
The Dirichlet part of Proposition~\ref{prop:cone-property} gives a cone-type chain $P\subset Q$ starting at $\hat{a}$ in direction $e_1$, of length $\lfloor L/2\rfloor$, with the corresponding lower bounds on $|u|$.

Consider slabs $\{x\in Q:|x_1-y|\le2^d n\}$ for $y\in \Z$.
Each cube in $\mcD$ intersects at most $(2^{d+1}+2)n+1$ such slabs.
Thus, since there are at most $N$ cubes in $\mcD$, by the choice of $C_{\rm geo}$, there is an integer
\[
 y_*\in\left[\hat{a}_1+\frac{C_{\rm geo}n}{4},
            \hat{a}_1+\frac{C_{\rm geo}n}{2}\right]
 \quad\text{such that}\quad
 \{x\in Q:|x_1-y_*|\le2^d n\}\cap\mcD=\varnothing.
\]
Since $L\ge C_{\rm geo}n$, the chain $P$ reaches the level $y_*$ in its first coordinate. Let $\tilde{a}\in P$ be its first crossing point.
By \eqref{eq:dirft-of-cone-chain}, $y_*\le\tilde{a}_1\le y_*+1$, which further implies that $L-|\tilde{a}_1|\ge C_{\rm geo}n/4$.

Note that the cube $Q_n(\tilde{a})$ may not be contained in $Q$, so we need to take some further steps.
Assume without loss of generality that $\tilde{a}_i\le0$ for each $2\le i\le d$.
Starting from $\tilde{a}$, successively follow a cone-type chain in direction $e_i$ for $2^{d-i}n$ steps, for $i=2,\ldots,d$, and denote the final point by $a$.
At the start of stage $i$, the $i$-th coordinate is at most $\sum_{j=2}^{i-1}2^{d-j}n$. Hence Proposition~\ref{prop:cone-property} and the Dirichlet boundary condition and $L\ge C_{\rm geo}n$ ensure that every stage exists and stays in $Q$.
Since each stage starts in $Q$, for $2\le i\le d$, \eqref{eq:dirft-of-cone-chain} gives 
\[
 a_i\ge-L+2^{d-i}n-\sum_{j=i+1}^d2^{d-j}n=-L+n,
\]
and since each $\tilde{a}_i\le0$ for $2\le i\le d$, \eqref{eq:dirft-of-cone-chain} gives
\[
 a_i\le2\sum_{j=2}^d 2^{d-j}n<L-n.
\]
We also have
\[
|a_1-\tilde{a}_1| \le \sum_{j=2}^d 2^{d-j}n = (2^{d-1}-1)n,
\]
so $|a_1|<L-n$.
Consequently $Q_n(a)\subset \{x\in Q:|x_1-y_*|\le2^d n\}\subset Q\setminus \mcD$.
Finally, the concatenated chain from $\hat{a}$ to $a$ has length less than $C_{\rm geo}n$.
Applying the cone estimate in Proposition~\ref{prop:cone-property} along this chain and letting $C_{\rm esc}=C_{\rm cone}C_{\rm geo}$ yields
\[
 |u(a)|\ge \exp(-C_{\rm cone}C_{\rm geo}n)|u(\hat{a})|
          =\exp(-C_{\rm esc}n)\|u\|_{\ell^\infty(Q)},
\]
which proves the lemma.
\end{proof}

We now give the Wegner estimate, whose statement generalizes \cite[Lemma~3.5]{LZ22} to arbitrary dimension $d\ge2$. The proof is essentially the same as those of \cite[Lemma~5.6]{DS20} and \cite[Lemma~3.5]{LZ22}.

\begin{lemma}
\label{lem:wegner}
Take $N_1,N_F\in \N$. There exists
$\varepsilon_*=\varepsilon_*(d)>0$
such that the following holds.
Take any $0<\varepsilon<\varepsilon_*$ and $0<\delta\le\varepsilon/4$.
There is a large constant 
\[
 C_{\rm weg}=C_{\rm weg}(d,\lambda,N_1,N_F,\varepsilon,\delta)>0
\]
such that, if
\begin{enumerate}
\item[(1)] $E_0\in[0,4d+\lambda]$.
\item[(2)] $L_0>\cdots>L_5\ge C_{\rm weg}$ and
\[
L_j^{1-2\delta}\ge L_{j+1}\ge L_j^{1-\varepsilon/2},
 \qquad 0\le j\le4.
\]
Especially, we take $L_0,L_3$ to be dyadic, i.e., integer powers of $2$.
\item[(3)] $Q=Q_{L_0}$.
\item[(4)] $Q'_j=Q_{L_3}(z^{(j)})\subset Q$ for $1\le j\le J$, with some $z^{(j)}\in\mathbb Z^d$,
where $1\le J\le N_1$. Denote $\mcD=\bigcup_{j=1}^{J}Q'_j$.
\item[(5)] $G\subset \mcD$ with $1\le |G|<L_0^\delta$.
\item[(6)] 
Let $F$ be a fixed discrete $(N_F,\varepsilon,\boldsymbol\rho)$-graded set, in the sense of Definition~\ref{def:fsa-graded}, with a scale sequence $\boldsymbol\rho=(\rho_1,\ldots,\rho_{n_{\mathrm{lev}}})$ satisfying $\rho_1\ge C_{\rm weg}$.
Moreover, for any $a\in\Z^d$ with $Q_{L_3}(a)\subset Q\setminus \mcD$, $F$ is normal in $Q_{L_3}(a)$.
          We fix a $v_F\in\{0,1\}^{F\cap Q}$.
\item[(7)] For every $\mathsf V\in[0,1]^Q$ with $\mathsf V|_{F\cap Q}=v_F$, let $H_Q(\mathsf V)$ be the Dirichlet restriction of $-\Delta+\lambda\mathsf V$ to $Q$.
For every $E$ with $|E-E_0|\le \exp(-L_5)$ and every $u:Q\to \C$ satisfying
$H_Q(\mathsf V)u=Eu$, we have
\begin{equation}\label{eq:six-robust-support}
 \exp(L_4)\|u\|_{\ell^2(Q\setminus \mcD)}
 \le\|u\|_{\ell^2(Q)}
 \le(1+L_0^{-\delta})\|u\|_{\ell^2(G)}.
\end{equation}
\end{enumerate}
Then
\begin{align}
 &\mathbb P\left[\|(H_Q-E_0)^{-1}\|>\exp(L_1)
                         \,\middle|\,V|_{F\cap Q}=v_F\right]
\le L_0^{-d/2+10d\varepsilon}.
 \label{eq:six-wegner-bound}
\end{align}
\end{lemma}

\begin{proof}
In addition to fixing $V|_{F\cap Q}=v_F$, we condition on the randomness $V|_{\mcD\setminus F}$. Denote
\[
 F'=(F\cap Q)\cup \mcD,\qquad \Sigma=Q\setminus F'.
\]
Under the resulting conditional law, denoted by $\mathbb P_\Sigma$, the variables $(V_x)_{x\in\Sigma}$ remain independent Bernoulli$(1/2)$ variables.

We denote
\[
 \mathcal C=\{a\in\mathbb Z^d:Q_{L_3}(a)\subset Q\setminus \mcD\}.
\]
Set $I=[-1,4d+\lambda+1]$, and take $C_{\rm geo},C_{\rm esc}$ from Lemma~\ref{lem:duc-block-lemma} with this interval and $N=N_1$.
We use Theorem~\ref{thm:duc-linear} with the same $d,\lambda,\varepsilon$, the interval $I$, and
$R=L_3$, $N=N_F$, $K=C_{\rm esc}$,
at the fixed energy $E_0$.
Let $\varepsilon_0(d)$ be given by Theorem~\ref{thm:duc-linear}, and choose $\varepsilon_*\le\min\{\varepsilon_0(d),1/2\}$.
Denote the scale threshold and constants from Theorem~\ref{thm:duc-linear} by $L_*(N_F,\varepsilon)$ and $C_{\rm duc},c_{\rm duc},R_0$.
Take $C_{\rm weg}\ge\max\{4,L_*(N_F,\varepsilon),R_0\}$.
If we increase $C_{\rm weg}$ further, the inequalities in (2) also ensure
\[
 L_0\ge C_{\rm geo}L_3,\qquad L_2\ge C_{\rm duc}L_3.
\]

For each $a\in\mathcal C$, apply Theorem~\ref{thm:duc-linear} in $Q_{L_3}(a)$, where $F$ is normal by (6), with prescribed values $v_F$ on $F\cap Q_{L_3-1}(a)$.
Since $Q_{L_3}(a)$ is disjoint from $\mcD$, the random potentials on $Q_{L_3-1}(a)\setminus F\subset\Sigma$ have not been fixed by the additional conditioning and retain the independent Bernoulli law.
Theorem~\ref{thm:duc-linear} therefore gives an event $\mathcal E_{L_3}(F,v_F,E_0;a)$, with
\[
 \mathbb P_\Sigma\!\left[\mathcal E_{L_3}(F,v_F,E_0;a)^c\right]
 \le \exp\bigl(-c_{\rm duc}(L_3/\log L_3)^\gamma\bigr),\qquad \gamma=\frac{d-1}{d},
\]
and the following is true.
On this event, simultaneously for every $E\in I$ and every nonzero $\hat{u}:Q_{L_3}(a)\to\C$ satisfying $H\hat{u}=E\hat{u}$ on $Q_{L_3-1}(a)$, and
\[
 |E-E_0|\le \exp(-C_{\rm duc}L_3),\qquad
 \|\hat{u}\|_{\ell^\infty(Q_{L_3}(a))}\le \exp(C_{\rm esc}L_3)|\hat{u}(a)|,
\]
one has
\[
 \left|\left\{x\in Q_{L_3}(a)\setminus F:
 |\hat{u}(x)|\ge \exp(-C_{\rm duc}L_3)|\hat{u}(a)|\right\}\right|
 \ge\frac{c_{\rm duc}L_3^d}{(\log L_3)^d}.
\]
We take the intersection
\begin{equation}\label{eq:QUC-event-wegner}
  \mathcal E = \bigcap_{a\in\mathcal C}\mathcal E_{L_3}(F,v_F,E_0;a).
\end{equation}
A union bound gives
\begin{equation}\label{eq:six-quc-event}
 \mathbb P_\Sigma[\mathcal E^c]
 \le |\mathcal C|\exp\bigl(-c_{\rm duc}(L_3/\log L_3)^\gamma\bigr)
 \le  (2L_0+1)^d \exp\bigl(-c_{\rm duc}(L_3/\log L_3)^\gamma\bigr)
\end{equation}
since $|\mathcal C|\le |Q|= (2L_0+1)^d$.

Now we fix an arbitrary $V|_{\Sigma}\in \mcE$. Let $u$ be a normalized eigenvector of $H_Q$ with eigenvalue $E$, i.e., $H_Qu=Eu$, with $\|u\|_{\ell^2(Q)}=1$ and $|E-E_0|\le \exp(-L_2) \leq \exp(-C_{\rm duc}L_3)$. In particular, $E\in I$. Applying Lemma~\ref{lem:duc-block-lemma}, we find $a\in\mathcal C$ with
\begin{equation}   \label{eq:uaLlbd}
 \|u\|_{\ell^\infty(Q_{L_3}(a))}\le \|u\|_{\ell^\infty(Q)}
       \le \exp(C_{\rm esc}L_3)|u(a)|.
\end{equation}
Therefore, the event $\mathcal E_{L_3}(F,v_F,E_0;a)$ gives
\[
|\{x\in Q_{L_3}(a)\setminus F:|u(x)|\ge \exp(-C_{\rm duc}L_3)|u(a)|\}|
 \ge \frac{c_{\rm duc}L_3^d}{(\log L_3)^d}.
\]
The assumption \eqref{eq:six-robust-support} yields that 
\[\| u\|_{\ell^\infty(Q)}\geq \| u\|_{\ell^\infty(G)} \geq \| u\|_{\ell^2(G)}/\sqrt{|G|} \geq 1/(1+L_0^{\delta}).\]
Then with \eqref{eq:uaLlbd} and the fact that  $Q_{L_3}(a)\setminus F\subset \Sigma$, we have
\begin{equation}\label{eq:six-high-count}
 |\{x\in \Sigma:|u(x)|\ge \exp(-L_2/4)\}|
 \ge \frac{c_{\rm duc}L_3^d}{(\log L_3)^d},
\end{equation}
where we used that
\[
(C_{\rm esc}+C_{\rm duc})L_3
       +\log(L^{\delta} _0+1)\le\frac{L_2}{4}
\]
for sufficiently large $C_{\rm weg}$.

Now consider $H_Q(\mathsf V)$, for any $\mathsf V\in[0,1]^\Sigma$, extended to an element of $[0,1]^Q$ by taking $\mathsf V|_{F'}=V|_{F'}$, i.e., the already fixed values on $F'$.
Write its eigenvalues, counted with multiplicity, as
\[
 \mu_1(\mathsf V)\ge \cdots\ge \mu_{|Q|}(\mathsf V),
\]
and choose a corresponding orthonormal eigenbasis $u_j(\mathsf V)$, $1\le j\le |Q|$.
Define
\[
\mathcal K=\left\{j: \min_{\mathsf V\in[0,1]^\Sigma}
                     |\mu_j(\mathsf V)-E_0|\le \exp(-L_2)\right\}.
\]
The bootstrap argument of Bourgain (see, e.g., \cite[Claim~5.11]{DS20}, \cite[proof of Claim~3.9]{LZ22}, and \cite[(3.57)--(3.61)]{LSZ25} for details), using \eqref{eq:six-robust-support}, gives, for every $j\in\mathcal K$,
\[
\max_{\mathsf V\in[0,1]^\Sigma}|\mu_j(\mathsf V)-E_0|
 \le \exp(-L_2)+\lambda \exp(-2L_4)\le \exp(-L_4)<\exp(-L_5).
\]
Thus, at any fixed $\mathsf V\in[0,1]^\Sigma$, all $u_j=u_j(\mathsf V)$ with $j\in\mathcal K$ satisfy \eqref{eq:six-robust-support}. The standard Hilbert--Schmidt argument then gives
\[
 \frac{|\mcK|}{(1+L_0^{-\delta})^2}
 \le\sum_{j\in\mathcal K}\| u_j\|_{\ell^2(G)}^2
 =\sum_{x\in G}\sum_{j\in\mathcal K}|u_j(x)|^2
 \le |G|.
\]
Consequently
\begin{equation}\label{eq:six-rank-number}
 |\mcK| \le(1+L_0^{-\delta})^2|G|\le2 L_0^{\delta}.
\end{equation}

Below we again consider $H_Q$ with $V|_{F'}$ fixed, and $V|_{\Sigma}$ being independent Bernoulli$(1/2)$.
We use $\mu_j$ and $u_j$ to denote its eigenvalues and (orthonormal) eigenfunctions.

Use the conventions $\mu_0=+\infty$ and $\mu_{|Q|+1}=-\infty$.
For $1\leq j_1\leq j_2\leq |Q|$ and an integer $0\leq \ell \leq 10|\mcK| \le 20 L_0^\delta$, we denote by $\mcE_{j_1,j_2,\ell}$ the event (concerning $V|_\Sigma$) such that
\[
|\mu_{j_1}-E_0|, |\mu_{j_2}-E_0|<s_{\ell} \quad\text{and}\quad |\mu_{j_1-1}-E_0|,|\mu_{j_2+1}-E_0|\geq s_{\ell+1},
\]
where for each $\ell\in \Z$,
\[
 s_\ell=\exp(-L_1+\ell(L_2-L_4)).
\]
By arguments as in \cite[Claim~5.11]{DS20}, \cite[Claim~3.9]{LZ22}, or \cite[Claim~4.5]{LSZ26b}, we have
\begin{equation}\label{eq:six-cover}
 \{V|_\Sigma:\|(H_Q-E_0)^{-1}\|>\exp(L_1)\}
 =\{V|_\Sigma:\dist(\spec(H_Q),E_0)<\exp(-L_1)\}
 \subset\bigcup_{\substack{j_1,j_2\in\mathcal K,\ j_1\le j_2\\
                             0\le\ell\le10|\mcK|}}
             \mathcal E_{j_1,j_2,\ell}.
\end{equation}
Therefore it remains to bound the probability of each $\mcE_{j_1,j_2,\ell}$.

The scale assumptions give $s_\ell\le \exp(-L_1/2)<\exp(-L_2)$ for $0\le\ell\le10|\mcK|+1$, when $C_{\rm weg}$ is sufficiently large. Thus $\mcE_{j_1,j_2,\ell}$ implies that $|\mu_{j_1}-E_0|,|\mu_{j_2}-E_0|<\exp(-L_2)$. Then under the event $\mcE_{j_1,j_2,\ell}\cap\mcE$, we have
\[ S_{\rm duc}(j_1) := \{x\in \Sigma:|u_{j_1}(x)|\ge \exp(-L_2/4)\},\qquad |S_{\rm duc}(j_1)| \geq \frac{c_{\rm duc}L_3^d}{(\log L_3)^d} ,\]
\[ S_{\rm duc}(j_2) := \{x\in \Sigma:|u_{j_2}(x)|\ge \exp(-L_2/4)\},\qquad |S_{\rm duc}(j_2)| \geq \frac{c_{\rm duc}L_3^d}{(\log L_3)^d} .\]
Decompose
\begin{equation}\label{eq:event-inclusion}
  \mcE_{j_1,j_2,\ell}\cap \mcE \subset \mcE_{j_1,j_2,\ell,0}\cup  \mcE_{j_1,j_2,\ell,1},
\end{equation}
with  $\mcE_{j_1,j_2,\ell,i}$ being the event that 
\begin{equation}\label{eq:pigeonhole-event}
    \mcE_{j_1,j_2,\ell} \ {\rm and} \ \left| S_{\rm duc}(j_1)  \cap \{x\in \Sigma:V_x = i\}  \right|\geq \frac{c_{\rm duc}L_3^d}{2(\log L_3)^d} ,
\end{equation}
for each $i=0,1$.

Now under $\mcE_{j_1,j_2,\ell,i}$, take $x\in\Sigma$ such that
\[ V_x=i \quad {\rm and}\quad |u_{j_1}(x)|\geq \exp(-L_2/4).\]
Let $V^{(x)} \in \{0,1\}^\Sigma$ be the configuration obtained by flipping the potential at $x$:
\begin{equation*}
  V_y^{(x)}=\begin{cases}
     V_y, \ {\rm if} \ y\in \Sigma \ {\rm and}\ y \neq x,\\
    1- V_x, \ {\rm if} \ y\in \Sigma \ {\rm and}\ y=x.
  \end{cases}
\end{equation*}
We apply the eigenvalue-variation lemma \cite[Lemma~5.1]{DS20} to
$\widetilde H_Q=((-1)^i(H_Q-E_0)+s_\ell)/\lambda$,
for which the flip at $x$ adds $\delta_x\otimes\delta_x$, with parameters
\[
 r_1=\frac{2s_\ell}{\lambda},\quad
 r_2=\frac{s_{\ell+1}}{\lambda},\quad
 r_3=\exp(-L_2/2),\quad r_4=\exp(-L_4),\quad
 r_5=\frac{\exp(-L_5)}{2\lambda},
\]
and take $i,j$ in \cite[Lemma~5.1]{DS20} to be
\[
 \begin{cases}
  (j_1,j_1),&i=0,\\
  (|Q|+1-j_2,|Q|+1-j_1),&i=1,
 \end{cases}
\]
so that the $j$-th eigenvector of $\widetilde H_Q=((-1)^i(H_Q-E_0)+s_\ell)/\lambda$ is $u_{j_1}$ in both cases.
Using the scale assumptions and \eqref{eq:six-robust-support}, and the choice of $x$, we see that all conditions required in \cite[Lemma~5.1]{DS20} are satisfied.
Thus we conclude, 
returning to the original operator,
\begin{equation}\label{eq:six-flip-endpoints}
 \begin{cases}
  \mu_{j_1}(V^{(x)})\ge E_0+s_\ell,&i=0,\\
  \mu_{j_2}(V^{(x)})\le E_0-s_\ell,&i=1.
 \end{cases}
\end{equation}
Then we have $V^{(x)}\notin\mcE_{j_1,j_2,\ell}$.

For each $i\in\{0,1\}$, under $\mcE_{j_1,j_2,\ell,i}$, let $A=\{x\in\Sigma:V_x=1-i\}$, and set
\[
 B(A)=S_{\rm duc}(j_1)\cap\{x\in\Sigma:V_x=i\}.
\]
By \eqref{eq:six-flip-endpoints} and eigenvalue monotonicity, for any $V':\Sigma\to\{0,1\}$, where $A'=\{x\in\Sigma: V'_x=1-i\}$ satisfies $A'\supset A$ and $A'\cap B(A)\ne\varnothing$, $V'$ lies outside $\mcE_{j_1,j_2,\ell}$.
Also, with 
\[
 \rho=\frac{c_{\rm duc}L_3^d}{2(\log L_3)^d|\Sigma|}\in(0,1],
\]
\eqref{eq:pigeonhole-event} gives $|B(A)|\ge\rho|\Sigma|\ge\rho(|\Sigma|-|A|)$.
Thus $\mcE_{j_1,j_2,\ell,i}$ corresponds to a $\rho$-Sperner family, in the sense of \cite[Definition~4.1]{DS20}.
Therefore \cite[Theorem~4.2]{DS20} implies that
\[
\mathbb P_\Sigma[\mcE_{j_1,j_2,\ell,i} ] \leq |\Sigma|^{-\frac{1}{2}} \rho^{-1}\lesssim_{d,c_{\rm duc}} \frac{L_0^{d/2}(\log L_3)^d}{L_3^d}.
\]
Then \eqref{eq:event-inclusion} yields 
\begin{align}\label{eq:single-event-prob}
  \mathbb P_\Sigma[\mcE_{j_1,j_2,\ell} \cap \mcE] \leq \mathbb P_\Sigma[\mcE_{j_1,j_2,\ell,0} ]+\mathbb P_\Sigma[\mcE_{j_1,j_2,\ell,1} ] \lesssim_{d,c_{\rm duc}} \frac{L_0^{d/2}(\log L_3)^d}{L_3^d}.
\end{align}

Combining \eqref{eq:six-quc-event}, \eqref{eq:six-cover} and \eqref{eq:single-event-prob}, we conclude that 
\begin{align*}
    \mathbb P_\Sigma [\{V|_\Sigma: \|(H_Q-E_0)^{-1}\| > \exp(L_1) \}] & \leq \mathbb P_\Sigma \left[ \bigcup_{\substack{j_1,j_2\in \mcK,\ j_1\le j_2 \\ 0\leq \ell \leq 10|\mcK| }}\mcE_{j_1,j_2,\ell}\right] \\
                      & \leq\sum_{\substack{j_1,j_2\in \mcK,\ j_1\le j_2 \\ 0\leq \ell \leq 10|\mcK| }} \mathbb P_\Sigma \left[ \mcE_{j_1,j_2,\ell}\cap \mcE\right] +\mathbb P_\Sigma[\mcE^c]\\
            & \lesssim_{d,c_{\rm duc}} L_0^{3\delta}\frac{L_0^{d/2}(\log L_3)^d}{L_3^d}+  L_0^d \exp\bigl(-c_{\rm duc}(L_3/\log L_3)^\gamma\bigr).
\end{align*}
Here the last inequality uses \eqref{eq:six-rank-number}. 
Note that the last line can be made much smaller than $L_0^{-d/2+10d \varepsilon}$, by taking $L_0\geq C_{\rm weg}$ large enough. The estimate is uniform in the revealed $V|_{\mcD\setminus F}$, so averaging over these values proves \eqref{eq:six-wegner-bound}.
\end{proof}

\subsection{Multi-scale analysis and localization}

We next give the final MSA statements from the Wegner estimate. 
This procedure follows the standard scheme in, e.g., \cite[Section 3, Appendix B]{LZ22}, with adjustments for the dimension and disorder strength; we state only the initial scale estimate and the final Green's function estimate.

The Green's function estimate at the initial scale is 
used as the first step in the inductive proof of the MSA. 
The initial scale analysis for the ABM was first given in \cite{Bou04} using large deviation events (see also \cite[Section 3.2]{LSZ26b}). In \cite[Section 7]{DS20}, the initial scale estimate is instead proved using a ``net structure'' of the potential; see also \cite[Appendix B]{LZ22}. We state a version of \cite[Proposition B.4]{LZ22} for every $d\ge2$ and $\lambda>0$. 

\begin{theorem}
\label{thm:msa-initial}
Let $\{V_x\}_{x\in\Z^d}$ be independent Bernoulli$(1/2)$. 
For any $d\ge2$, $\lambda>0$, $0<\delta<\frac1{10}$ and $\varepsilon>0$,
there exists $C_{d,\lambda,\delta,\varepsilon}$ such that, for every
$n>C_{d,\lambda,\delta,\varepsilon}$, there is an event depending only on
$V|_{Q_n\cap\lceil\varepsilon^{-1}\rceil\Z^d}$ with probability $\ge 1-\exp(-n^{\delta/20})$,
and on this event the following holds. For every
$0\le E\le n^{-\delta/5}$
and every $V':\mathbb Z^d\to[0,1]$ such that
\[
 V'|_{Q_n\cap\lceil\varepsilon^{-1}\rceil\mathbb Z^d}
 =
 V|_{Q_n\cap\lceil\varepsilon^{-1}\rceil\mathbb Z^d},
\]
let $H'_{Q_n}$ be the restriction of $-\Delta+\lambda V'$ to $Q_n$
with Dirichlet boundary conditions. Then we have
\[
\bigl\|(H'_{Q_n}-E)^{-1}\bigr\|
 \le n^{\delta/5},
\]
and
\[
\bigl|(H'_{Q_n}-E)^{-1}(a,b)\bigr|
 \le n^{\delta/5}\exp\bigl(-n^{-\delta/5}|a-b|\bigr)
 \qquad\text{for any }a,b\in Q_n.
\]
\end{theorem}
The proof follows that of \cite[Proposition B.4]{LZ22}, taking the net scale $R=n^{\delta/(10d)}$ instead of $R=n^{\delta/10}$. For $d\ge3$, we use \cite[Corollary B.3]{LZ22}.
For $d=2$, we use the proof of \cite[Lemma 7.2]{DS20}. Retaining the polynomial prefactor in its barrier estimate and allowing constants to depend on $\lambda$, we obtain, with $L=R^2\log R$,
\[\bigl\|(H'_{Q_n}-E)^{-1}\bigr\|\le C_\lambda L,\qquad \bigl|(H'_{Q_n}-E)^{-1}(a,b)\bigr|\le C_\lambda L\exp\bigl(-c_\lambda |a-b|/L\bigr),\]
uniformly for $0\le E\le c_\lambda/L$, with $c_\lambda>0$ sufficiently small. Indeed, the positive barrier satisfies $1\le\psi\le C_\lambda\log R$ and $H'_{Q_n}\psi\ge cR^{-2}$, so the barrier argument remains valid after subtracting $E$ in this range. Since $R=n^{\delta/20}$ when $d=2$, we have $L\lesssim n^{\delta/10}\log n$, and these estimates imply the stated bounds for sufficiently large $n$.
In every dimension $d\ge2$, the probability that the sites with potential one in the prescribed sublattice fail to form an $R$-net in $Q_n$ is at most $C_dn^d\exp(-c_{d,\varepsilon}R^d)\le\exp(-n^{\delta/20})$ for sufficiently large $n$, as in \cite[Proposition B.4]{LZ22}. We omit the remaining details.

Starting from the Green's function estimate at the initial scale, we can prove the corresponding estimate for large scales using the MSA.

\begin{theorem}
\label{thm:msa-short-large}
For every dimension $d\ge2$ and  disorder strength $\lambda>0$, there are $\kappa_*=\kappa_*(d)>0$,
$E_*=E_*(d,\lambda)>0$, $\gamma_*=\gamma_*(d,\lambda)>0$, and $R_*=R_*(d,\lambda)>0$ such that, for every fixed
$E_0\in[0,E_*]$, every dyadic $R\ge R_*$, and every fixed
$a\in\mathbb Z^d$,
\[
\mathbb P\!\big[
 |(H_{Q_R(a)}-E_0)^{-1}(x,y)|
 \le\exp\bigl(R^{1-\gamma_*}-\gamma_*|x-y|\bigr),
 \quad\forall \ x,y\in Q_R(a)
\big]
 \ge1-R^{-\kappa_*},
\]
where $H_{Q_R(a)}$ is the restriction of $H$ from \eqref{eq:ABM-schrodinger-operator} to $Q_R(a)$ with Dirichlet boundary conditions.
\end{theorem}

Theorem~\ref{thm:msa-short-large} is proved essentially the same way as \cite[Theorem 3.1]{LZ22}.
More precisely, we use induction on scales as in \cite[Theorem~3.10]{LZ22}, starting from Theorem~\ref{thm:msa-initial}.
At the step corresponding to \cite[Claim~3.14]{LZ22}, we apply the Wegner estimate Lemma~\ref{lem:wegner} at the scale of the enclosing defect cube.
The covering counts, grading multiplicities, auxiliary scales, and number of generations of bad subcubes are adjusted for dimension $d$ to meet the conditions for Lemma~\ref{lem:wegner}.
The readiness error and union bounds are absorbed by choosing $\kappa_*>0$ sufficiently small.
The passage to all sufficiently large dyadic scales follows the covering argument in the proof of \cite[Theorem~3.1]{LZ22}.
We omit all these details.

Now Theorem~\ref{thm:localization-near-the-edge-band} follows from Theorem~\ref{thm:msa-short-large}. Indeed, for SL, we apply the Peierls argument of Bourgain-Kenig~\cite[Section~7]{BK05}; see also \cite[Sections~6 and~7]{germinet2012comprehensive} or \cite[Appendix C]{LZ22}.
The SDL bound follows from~\cite[Theorem~1]{RZ}, with the probability exponent being $p_0=\kappa_*/d$.

\appendix

\section{Uniform discrete unique continuation}
\label{appendix:uniform-duc}

In this appendix, we strengthen Theorem~\ref{thm:duc-linear} with $d\ge 3$ and $F\cap Q_R=\varnothing$ to an estimate uniform in the disorder strength and energy over the parameter set
\[
\Theta(\lambda_-,\lambda_+,E_+):=\left\{(\lambda,E):0<\lambda_-\le\lambda\le\lambda_+,\ |E|\le E_+\right\}.
\]
The arguments for this follow an earlier strategy we obtained when exploring using LLMs, as discussed in Section~\ref{ssec:ai}. More precisely, we still use the bootstrap route in the proof of Theorem~\ref{thm:duc-linear}, but with a larger initial set, supplied by the deterministic DUC from \cite{LZ22,LSZ26quc}.
We take $F\cap Q_R=\varnothing$ because deterministic DUCs in \cite{LSZ26quc} do not exclude a graded set.

\begin{theorem}[Uniform PDUC with no graded set]\label{thm:uniform-duc-linear}
Take $d\ge 3$, $\lambda_+>\lambda_->0$, $E_+>0$, $K\ge0$. There are constants
$C,c,R_0>0$, depending on
$d,\lambda_+,\lambda_-,E_+,K$, such that the following holds.
Take any $R\ge R_0$, and let $\{V_x\}_{x\in Q_{R-1}}$ be independent Bernoulli$(1/2)$.
Then there is
an event $\mathcal E_R \subset \{0,1\}^{Q_{R-1}}$ depending only on these random variables, with
\begin{equation}\label{eq:uniform-fsa-probability}
\mathbb P\bigl[\mathcal E_R \bigr]
\ge
1-\exp\big(
-c(R/\log R)^{(d-1)/d}
\big),
\end{equation}
such that on this event, for every $(\lambda,E)\in \Theta(\lambda_-,\lambda_+,E_+)$ and 
$u:Q_R\to\mathbb C$ satisfying $(-\Delta+\lambda V)u=Eu$ on $Q_{R-1}$ and
\[
u(0)\ne0,\qquad
\|u\|_{\ell^\infty(Q_R)}\le \exp(KR)|u(0)|,
\]
we have
\begin{equation}\label{eq:uniform-fsa--count}
|\{
x\in Q_R:
|u(x)|\ge \exp(-CR)|u(0)|
\}|
\ge \frac{cR^d}{(\log R)^d}.
\end{equation}
\end{theorem}

We follow the proof of Theorem~\ref{thm:duc-linear}, with two main differences. First, we use a finite net argument in the parameter space $\Theta$, leading to an extra term in the recurrence \eqref{eq:baby-recurrence}. Second, we adapt the initial bound \eqref{eq:fsa-seed} to absorb this term, using deterministic DUC.
We next explain these changes in detail.

\subsection{Adapting the proof of Theorem~\ref{thm:duc-linear}}
Take $I=[-E_+,E_+]$ and a discrete graded set $F$ satisfying $F\cap Q_R=\varnothing$. The proof of Theorem~\ref{thm:duc-linear} goes through, with the following changes. 

\smallskip
\noindent\textbf{(1)}
Since $F\cap Q_R=\varnothing$, the sparsity bounds from Lemma~\ref{lem:fsa-packing} are automatic for all chains in $Q_R$, so the graded-set parameters $N,\varepsilon,\boldsymbol\rho$ and the associated thresholds $L_*,R_{\rm spa}$ are no longer needed.

\smallskip
\noindent\textbf{(2)}
The estimates in Theorem~\ref{thm:reconstruction}, including the bounds for $\mathscr W$, $\partial_E\mathscr W$, and \eqref{eq:reconstruct-ell-infty-control}, hold uniformly for $(\lambda,E)\in\Theta$ with $C_{\rm recon}=C_{\rm recon}(d,E_+,\lambda_+)$. The same path-product estimate also gives, uniformly for $(\lambda,E)\in\Theta$ and $k\ge1$,
\begin{equation}\label{eq:bound-on-partial-lambda}
\|\partial_\lambda\mathscr W\|_{\ell^\infty(\mathfrak D_k\times\mathfrak B_k)}
\le \exp(C_{\rm recon}s),
\end{equation}
since differentiating a one-step coefficient with respect to $\lambda$ introduces a factor $V_x\in[0,1]$.

\smallskip
\noindent\textbf{(3)}
Set $\epsilon_R=\exp(-CR)$ and choose a finite set $\Theta_{\rm net}\subset\Theta$, such that for every $(\lambda,E)\in\Theta$, one can find $(\lambda_0,E_0)\in \Theta_{\rm net}$ satisfying
\[
|\lambda-\lambda_0|\le\epsilon_R,\qquad |E-E_0|\le\epsilon_R.
\]
The net can be chosen with
\begin{equation}\label{eq:cardinality-of-net}
|\Theta_{\rm net}|\le C_\Theta\epsilon_R^{-2}
=C_\Theta\exp(2CR)\le \exp\bigl(2(C+1)R\bigr),
\end{equation}
where $C_\Theta$ depends only on $\lambda_-,\lambda_+,E_+$, and the last inequality holds for sufficiently large $R$.

\smallskip
\noindent\textbf{(4)}
Enlarge the index set \eqref{eq:index-set-of-all-system} by appending the net point:
\[
\mathfrak i=((\lambda_0,E_0),j,k,\nu,\mcA),\qquad (\lambda_0,E_0)\in\Theta_{\rm net},
\]
where $(j,k,\nu,\mcA)$ ranges over the original indices. For each index, Construction~\ref{construction:fc-system} at $(\lambda_0,E_0)$ gives a linear-martingale system, with the common bounds
\[
c_0=\lambda_-,\qquad C_{0,j}=\exp(C_{\rm recon}Rw_j/16).
\]

\smallskip
\noindent\textbf{(5)}
When applying Lemma~\ref{lem:propagation}, replace \eqref{eq:Tp} by
\[
L(\mathfrak i)=(d+1)|\mcA|\log R+(d+3)\log R+2(C+1)R
+\eta\left(\frac R{\log R}\right)^\gamma.
\]
Let $\mathcal E(\mathfrak i)$ be the resulting events and set $\mcE_R=\bigcap_{\mathfrak i\in{\rm SysIndex}}\mathcal E(\mathfrak i)$ over the enlarged index set. The weighted sum \eqref{eq:number-sys-p} gains a factor of $|\Theta_{\rm net}|$, which is absorbed by the added term $2(C+1)R$ in $L(\mathfrak i)$:
\begin{align*}
\P[\mcE_R^c]
&\le |\Theta_{\rm net}|R^{-(d+1)}
\exp\left(-2(C+1)R-\eta\left(\frac R{\log R}\right)^\gamma\right)\\
&\le R^{-(d+1)}\exp\left(-\eta\left(\frac R{\log R}\right)^\gamma\right).
\end{align*}
This gives \eqref{eq:uniform-fsa-probability} by taking $c\le\eta$.

\smallskip
\noindent\textbf{(6)}
In Step 2, for any $(\lambda,E)\in\Theta$, choose a nearby net point $(\lambda_0,E_0)$.
Define $\tilde u_1$ as before at $(\lambda,E)$ with the truncated boundary data \eqref{eq:u_1-boundary-data}. For $\tilde u_2$, replace both $\lambda$ and $E$ by $\lambda_0$ and $E_0$, keeping the same boundary data and potential.
The bound \eqref{eq:diff-solution-2} is replaced by
\begin{multline*}
\|\tilde u_1-\tilde u_2\|_{\ell^\infty(Q_{r_j+\ell_j+1})}
\le (|\lambda-\lambda_0|+|E-E_0|)\sqrt p\,\exp(C_{\rm recon}s_j)\|\vartheta\|_2\\
\le 2\exp(-CR)(2R+1)^{d/2}\exp(C_{\rm recon}Rw_j/16)R_{\rm in}
\le \exp\bigl(-(C-C_{\rm recon}/16-K-2)R\bigr),
\end{multline*}
for sufficiently large $R$. 

\smallskip
\noindent\textbf{(7)}
The additional term $2(C+1)R$ in $L(\mathfrak i)$ contributes an extra error $2(C+1)R\lfloor\ell_j/4\rfloor$ in \eqref{eq:event-discribe-out-of-Ui}. With this bounded by $C_*\ell_jR$, the key recurrence \eqref{eq:baby-recurrence} becomes
\[
2m_{j+1}\ge\frac{\ell_j}{20}(m_j/2)^\gamma
-C_*\ell_j\left(\frac p{w_j}\log R+R\right)
-\eta\ell_j\left(\frac R{\log R}\right)^\gamma,
\]
where $C_*$ depends only on $d,\lambda_-,\lambda_+,E_+,K$.
Accordingly, the absorption condition \eqref{eq:seed-big-absorb} is replaced by
\begin{equation}\label{eq:seed-big-absorb-uniform}
C_*R+\eta\left(\frac R{\log R}\right)^\gamma
\le\frac1{40}(m_j/2)^\gamma.
\end{equation}

\smallskip
\noindent\textbf{(8)}
Since $m_j\ge m_0$, it suffices to find $\iota=\iota(d)>0$ such that, after increasing $R_0$ if necessary,
\begin{equation}\label{eq:fsa-seed-uniform}
m_0>R^{1/\gamma+\iota},\qquad R\ge R_0,
\end{equation}
for every $(\lambda,E)\in\Theta$ and every solution in the theorem normalized by $u(0)=1$. Indeed, $(m_j/2)^\gamma>2^{-\gamma}R^{1+\gamma\iota}$, which implies \eqref{eq:seed-big-absorb-uniform} for sufficiently large $R$.

It remains to establish \eqref{eq:fsa-seed-uniform}. We do this in the next subsection.

\subsection{Larger initial set for the bootstrap} 
We derive the improved initial set estimate, i.e., \eqref{eq:fsa-seed-uniform}, for $d\geq 3$, from deterministic DUCs.
When $d=3$, so that $d/(d-1)=3/2$, we recall the deterministic DUC from \cite{LZ22}.
\begin{theorem}[\protect{\cite[Theorem 1.3]{LZ22}}]
\label{thm:3-d-deter-duc}
There exists a constant $p>\frac{3}{2}$ such that the following
holds. For each $M>0$, there is $C_{M}>0$ such that, for any
sufficiently large $n\in\mathbb{Z}_{+}$ and functions
$u:\mathbb{Z}^{3}\to\mathbb{C}$ and $V:\mathbb{Z}^{3}\to\mathbb{R}$ satisfying
$\Delta u=Vu$ in $Q_n$ and $\|V\|_{\infty}\le M$, we have
\[
        |\{
            a\in Q_n:
            |u(a)|\ge \exp(-C_M n)\,|u(0)|
        \}|
    \ge n^{p}.
\]
\end{theorem}
Back to the setup of Theorem~\ref{thm:uniform-duc-linear}, writing $Hu=Eu$ as $\Delta u=(\lambda V-E)u$ and taking $M=\lambda_++E_+$, we have $\|\lambda V-E\|_\infty\le M$. Recall that we take $u(0)=1$. For $R$ sufficiently large, $r_0=\ell_0$ and \eqref{eq:fsa-scale-bounds} give $C_Mr_0\le R$, so Theorem~\ref{thm:3-d-deter-duc} yields
\[m_0= |\{x\in Q_{r_0}:|u(x)|\ge \exp(-R)\}| \ge |\{x\in Q_{r_0}:|u(x)|\ge \exp(-C_M r_0)\}| \geq r_0^p.
\]
Choose $\iota=(p-3/2)/2>0$. Since $r_0=\ell_0$ and $d=3$, \eqref{eq:fsa-scale-bounds} gives $r_0\ge Rw_0/(128d)>R/(384\log R)$. Hence, for sufficiently large $R$,
\[
m_0\ge r_0^p\ge\left(\frac{R}{384\log R}\right)^p>R^{3/2+\iota}.
\]
This proves \eqref{eq:fsa-seed-uniform} for $d=3$.

When $d\geq 4$, we use the following deterministic DUC from \cite{LSZ26quc}.

\begin{theorem}[\protect{\cite[Theorem 2.4]{LSZ26quc}}]\label{thm:deter-duc-d-geq-4}
Let $d \ge4$ and $M\ge0$.  There exist $c_d, C_{d,M}>0$ such that, for any sufficiently large $n\in \mathbb{N}$
and functions
$u,V:\mathbb{Z}^{d}\to\mathbb{C}$ satisfying
$\Delta u=Vu$
in $Q_n$ and $\|V\|_{\infty}\le M$, we have
\begin{equation}\label{eq:deter-cardinality}
 |\{x\in Q_n:
 |u(x)|\ge
 \exp\!\left(-C_{d,M}n^2 \right)|u(0)|
 \}|
 \ge c_d\frac{n^{\lceil d/2\rceil}}{\log(2+n)}.
\end{equation}
\end{theorem}

Theorem \ref{thm:deter-duc-d-geq-4} cannot be directly applied to deduce \eqref{eq:fsa-seed-uniform}, since the amplitude threshold in \eqref{eq:deter-cardinality} is taken to be $\exp (-C_{d,M}n^2 )|u(0)|$, instead of $\exp(-C_{d,M}n)|u(0)|$.
Therefore we need to use a trade-off argument, using the cone property (Proposition \ref{prop:cone-property}) to construct a sufficiently long cone-type chain and apply Theorem \ref{thm:deter-duc-d-geq-4} in many cubes centered on this cone-type chain. Such a trade-off argument has appeared in 
\cite[Section 6]{LZ22}. 

We also remark that, although \cite[Theorem~2.4]{LSZ26quc} also applies when $d=3$, the trade-off argument gives only the exponent $(\lceil d/2\rceil+1)/2$, insufficient for \eqref{eq:fsa-seed-uniform}. This is why we use the stronger estimate from \cite{LZ22} in $d=3$.

\begin{theorem}[Deterministic DUC for $d\geq 4$, trade-off]\label{thm:deter-duc-d-geq-4-trade-off}
Let $d \ge4$ and $M\ge0$.  There exist $c_d, C_{d,M}>0$ such that, for any sufficiently large $n\in \mathbb{N}$
and functions
$u,V:\mathbb{Z}^{d}\to\mathbb{C}$ satisfying
$\Delta u=Vu$
in $Q_n$ and $\|V\|_{\infty}\le M$, we have
\begin{equation}\label{eq:deter-cardinality-trade-off}
 |\{x\in Q_n:
 |u(x)|\ge
 \exp\!\left(-C_{d,M}n \right)|u(0)|
 \}|
 \ge c_d\frac{n^{\lceil d/2\rceil/2+1/2}}{\log(2+n)}.
\end{equation}
\end{theorem}

\begin{proof}
Assume without loss of generality that $u(0)=1$. 
By Proposition~\ref{prop:cone-property} (which extends to $\Delta u=Vu$ with bounded complex-valued $V$ by the triangle inequality), with $C_{\rm cone}=C_{\rm cone}(d,M)$, there is a cone-type chain $P=(x_0,\ldots,x_m)$ rooted at $x_0=0$ in direction $e_1$, with length
$m=\lfloor n/32\rfloor$, contained in $Q_{n/8}$, and
\begin{equation}\label{eq:trade-off-path}
 |u(x_j)|\ge \exp(-C_{\rm cone}n),\qquad \forall\ 0\le j\le m.
\end{equation}
Let $r=\lfloor\sqrt n \rfloor$.
Since the cone-type chain $P$ has direction $e_1$, we can choose a series of at least
$m/(2r+2)>r/200$ vertices in 
$P$,
whose first
coordinates are separated by $\ge 2r+1$. We denote the set of these vertices by $P_{\rm separated}$.
For sufficiently large $n$, the cubes $Q_r(y)$, $y\in P_{\rm separated}$, are pairwise disjoint and contained in $Q_n$.
Applying Theorem~\ref{thm:deter-duc-d-geq-4} in each cube gives
\[
\left|\{x\in Q_r(y):|u(x)|\ge \exp(-C'r^2)|u(y)|\}\right|
\ge\frac{c'r^{\lceil d/2\rceil}}{\log(r+2)},\qquad y\in P_{\rm separated},
\]
where $c', C'>0$ are the constants $c_d, C_{d,M}$ in that theorem.
By \eqref{eq:trade-off-path} and $r^2\le n$,
\[
\exp(-C'r^2)|u(y)|\ge \exp(-C'r^2-C_{\rm cone}n)
\ge \exp\bigl(-(C'+C_{\rm cone})n\bigr).
\]
Set $C_{d,M}=C'+C_{\rm cone}$.
Summing over the disjoint cubes and using $|P_{\rm separated}|>r/200$, we obtain
\begin{align*}
\left|\{x\in Q_n:|u(x)|\ge \exp(-C_{d,M}n)\}\right|
&\ge\sum_{y\in P_{\rm separated}}
\left|\{x\in Q_r(y):|u(x)|\ge \exp(-C_{d,M}n)\}\right|\\
&\ge\frac r{200}\cdot\frac{c'r^{\lceil d/2\rceil}}{\log(r+2)}
\gtrsim_d\frac{n^{\lceil d/2\rceil/2+1/2}}{\log(2+n)}.
\end{align*}
This proves \eqref{eq:deter-cardinality-trade-off}.
\end{proof}
Finally, back to the setup of Theorem~\ref{thm:uniform-duc-linear}, set \[
\iota=\frac{\lceil d/2\rceil+1}{4}-\frac{d}{2(d-1)}>0.
\]
Recall that we take $u(0)=1$.
By writing $\Delta u=(\lambda V-E)u$ and applying Theorem~\ref{thm:deter-duc-d-geq-4-trade-off} with $M=\lambda_++E_+$,
we get that for sufficiently large $R$,
\[
m_0\ge \left|\{x\in Q_{r_0}:|u(x)|\ge\exp(-C_{d,M}r_0)\}\right|\ge c_d\frac{r_0^{(\lceil d/2\rceil+1)/2}}{\log(2+r_0)}
>R^{d/(d-1)+\iota}.
\]
Here the first inequality uses $C_{d,M}r_0\le R$, which holds for sufficiently large $R$ by $r_0=\ell_0$ and \eqref{eq:fsa-scale-bounds}. Also by $r_0=\ell_0$ and \eqref{eq:fsa-scale-bounds}, we have $r_0\ge Rw_0/(128d)>R/(128d\log R)$ and $r_0<R$ for sufficiently large $R$, so the last inequality holds by the choice of $\iota$. 
This proves \eqref{eq:fsa-seed-uniform} for $d\ge 4$ and completes the proof of Theorem~\ref{thm:uniform-duc-linear}.

\newcommand{\etalchar}[1]{$^{#1}$}


\begin{thebibliography}{DYYY25b}

\bibitem[AGKW09]{AGKW09}
Michael Aizenman, Fran{\c{c}}ois Germinet, Abel Klein, and Simone Warzel.
\newblock On {Bernoulli} decompositions for random variables, concentration bounds, and spectral localization.
\newblock {\em Probab. Theory Related Fields}, 143(1-2):219--238, 2009.

\bibitem[AL25]{AL25}
Amol Aggarwal and Patrick Lopatto.
\newblock Mobility edge for the {Anderson} model on the {Bethe} lattice.
\newblock {\em arXiv preprint arXiv:2503.08949}, 2025.

\bibitem[AM93]{AM93}
Michael Aizenman and Stanislav Molchanov.
\newblock Localization at large disorder and at extreme energies: an elementary derivation.
\newblock {\em Comm. Math. Phys.}, 157(2):245--278, 1993.

\bibitem[And58]{And58}
Philip Warren Anderson.
\newblock Absence of diffusion in certain random lattices.
\newblock {\em Phys. Rev. (2)}, 109(5):1492--1505, 1958.

\bibitem[BDF{\etalchar{+}}19]{bucaj2019localization}
Valmir Bucaj, David Damanik, Jake Fillman, Vitaly Gerbuz, Tom VandenBoom, Fengpeng Wang, and Zhenghe Zhang.
\newblock Localization for the one-dimensional {Anderson} model via positivity and large deviations for the {Lyapunov} exponent.
\newblock {\em Trans. Amer. Math. Soc.}, 372(5):3619--3667, 2019.

\bibitem[BGV15]{BGV15}
Ilia Binder, Michael Goldstein, and Mircea Voda.
\newblock On fluctuations and localization length for the {Anderson} model on a strip.
\newblock {\em J. Spectr. Theory}, 5(1):193--225, 2015.

\bibitem[BK05]{BK05}
Jean Bourgain and Carlos E.~Kenig.
\newblock On localization in the continuous {A}nderson-{B}ernoulli model in higher dimension.
\newblock {\em Invent. Math.}, 161(2):389--426, 2005.

\bibitem[BK13]{BK13}
Jean Bourgain and Abel Klein.
\newblock Bounds on the density of states for {S}chr\"{o}dinger operators.
\newblock {\em Invent. Math.}, 194(1):41--72, 2013.

\bibitem[BLMS22]{BLMS}
Lev Buhovsky, Alexander Logunov, Eugenia Malinnikova, and Mikhail Sodin.
\newblock A discrete harmonic function bounded on a large portion of {$\Bbb Z^2$} is constant.
\newblock {\em Duke Math. J.}, 171(6):1349--1378, 2022.

\bibitem[Bou04]{Bou04}
Jean Bourgain.
\newblock On localization for lattice {S}chr\"{o}dinger operators involving {B}ernoulli variables.
\newblock In {\em Geometric aspects of functional analysis}, volume 1850 of {\em Lecture Notes in Math.}, pages 77--99. Springer, Berlin, 2004.

\bibitem[Bou05]{bourgain2005anderson}
Jean Bourgain.
\newblock {Anderson--Bernoulli} models.
\newblock {\em Mosc. Math. J.}, 5(3):523--536, 2005.

\bibitem[Bou13]{Bou13Strip}
Jean Bourgain.
\newblock A lower bound for the {Lyapounov} exponents of the random {Schr\"odinger} operator on a strip.
\newblock {\em J. Stat. Phys.}, 153(1):1--9, 2013.

\bibitem[BRCG25]{BAW25}
Ahmed Bou-Rabee, William Cooperman, and Shirshendu Ganguly.
\newblock Unique continuation on planar graphs.
\newblock {\em Discrete Anal.}, pages Paper No. 16, 12, 2025.

\bibitem[BYY20]{bourgade2020random}
Paul Bourgade, Horng-Tzer Yau, and Jun Yin.
\newblock {Random Band Matrices in the Delocalized Phase I: Quantum Unique Ergodicity and Universality}.
\newblock {\em Comm. Pure Appl. Math.}, 73(7):1526--1596, 2020.

\bibitem[CKM87]{CKM87}
Ren\'{e} Carmona, Abel Klein, and Fabio Martinelli.
\newblock Anderson localization for {B}ernoulli and other singular potentials.
\newblock {\em Comm. Math. Phys.}, 108(1):41--66, 1987.

\bibitem[CPSS24]{CPSS24}
Giorgio Cipolloni, Ron Peled, Jeffrey Schenker, and Jacob Shapiro.
\newblock Dynamical localization for random band matrices up to {$W\ll N^{1/4}$}.
\newblock {\em Comm. Math. Phys.}, 405(3):Paper No. 82, 2024.

\bibitem[CS22]{CS22}
Nixia Chen and Charles K.~Smart.
\newblock Random band matrix localization by scalar fluctuations.
\newblock {\em arXiv preprint arXiv:2206.06439}, 2022.

\bibitem[CFKS87]{CS87}
Hans Ludwig Cycon, Richard G. Froese, Werner Kirsch, and Barry Simon.
\newblock {\em {Schr{\"o}dinger} Operators: With Applications to Quantum Mechanics and Global Geometry}.
\newblock Texts and Monographs in Physics. Springer-Verlag, Berlin, 1987.

\bibitem[Dro25]{Dro25}
Reuben Drogin.
\newblock Localization of one-dimensional random band matrices.
\newblock {\em arXiv preprint arXiv:2508.05802}, 2025.

\bibitem[DS20]{DS20}
Jian Ding and Charles K.~Smart.
\newblock Localization near the edge for the {A}nderson {B}ernoulli model on the two dimensional lattice.
\newblock {\em Invent. Math.}, 219(2):467--506, 2020.

\bibitem[DSS02]{damanik2002localization}
David Damanik, Robert Sims, and G\"{u}nter Stolz.
\newblock Localization for one-dimensional, continuum, {Bernoulli-Anderson} models.
\newblock {\em Duke Math. J.}, 114(1):59--100, 2002.

\bibitem[DYYY25a]{DFYYY25}
Sofiia Dubova, Fan Yang, Horng-Tzer Yau, and Jun Yin.
\newblock Delocalization of non-mean-field random matrices in dimensions {$d\ge 3$}.
\newblock {\em arXiv preprint arXiv:2507.20274}, 2025.

\bibitem[DYYY25b]{DKYYY25}
Sofiia Dubova, Kevin Yang, Horng-Tzer Yau, and Jun Yin.
\newblock Delocalization of two-dimensional random band matrices.
\newblock {\em arXiv preprint arXiv:2503.07606}, 2025.

\bibitem[EGH25]{EGH25}
Dor Elboim, Antoine Gloria, and Felipe Hern\'{a}ndez.
\newblock Diffusivity of the {L}orentz mirror walk in high dimensions.
\newblock {\em arXiv preprint arXiv:2505.01341}, 2025.

\bibitem[ER25]{ER25}
L\'{a}szl\'{o} Erd\H{o}s and Volodymyr Riabov.
\newblock The zigzag strategy for random band matrices.
\newblock {\em arXiv preprint arXiv:2506.06441}, 2025.

\bibitem[FMSS85]{frohlich1985constructive}
J\"{u}rg Fr{\"o}hlich, Fabio Martinelli, Elisabetta Scoppola, and Thomas Spencer.
\newblock Constructive proof of localization in the {Anderson} tight binding model.
\newblock {\em Comm. Math. Phys.}, 101(1):21--46, 1985.

\bibitem[FS83]{FS83}
J\"{u}rg Fr\"{o}hlich and Thomas Spencer.
\newblock Absence of diffusion in the {A}nderson tight binding model for large disorder or low energy.
\newblock {\em Comm. Math. Phys.}, 88(2):151--184, 1983.

\bibitem[GK13]{germinet2012comprehensive}
Fran{\c{c}}ois Germinet and Abel Klein.
\newblock A comprehensive proof of localization for continuous {Anderson} models with singular random potentials.
\newblock {\em J. Eur. Math. Soc. (JEMS)}, 15(1):53--143, 2013.

\bibitem[GMP77]{goldsheid1977pure}
Ilya Ya.~Goldsheid, Stanislav A.~Molchanov, and Leonid A.~Pastur.
\newblock A pure point spectrum of the stochastic one-dimensional {Schr\"{o}dinger} operator.
\newblock {\em Funct. Anal. Appl.}, 11(1):1--8, 1977.

\bibitem[Gol22]{Gol22}
Michael Goldstein.
\newblock Fluctuations and localization length for random band {GOE} matrix.
\newblock {\em arXiv preprint arXiv:2210.04346}, 2022.

\bibitem[HMY24]{HMY24}
Jiaoyang Huang, Theo McKenzie, and Horng-Tzer Yau.
\newblock {Ramanujan} property and edge universality of random regular graphs.
\newblock {\em arXiv preprint arXiv:2412.20263}, 2024.

\bibitem[Hur26a]{Hur26a}
Omar Hurtado.
\newblock Localization and unique continuation for non-stationary {Schr\"odinger} operators on the 2d lattice.
\newblock {\em Comm. Math. Phys.}, 407(4):64, 2026.

\bibitem[Hur26b]{Hur26b}
Omar Hurtado.
\newblock Localization for non-stationary {A}nderson models in three dimensions.
\newblock {\em arXiv preprint arXiv:2603.17810}, 2026.

\bibitem[Imb21]{Imb21}
John Z.~Imbrie.
\newblock Localization and eigenvalue statistics for the lattice {A}nderson model with discrete disorder.
\newblock {\em Rev. Math. Phys.}, 33(8):Paper No. 2150024, 50, 2021.

\bibitem[Jit07]{Jit07}
Svetlana Jitomirskaya.
\newblock Ergodic {S}chr\"{o}dinger operators (on one foot).
\newblock In {\em Spectral theory and mathematical physics: a {F}estschrift in honor of {B}arry {S}imon's 60th birthday}, volume~76 of {\em Proc. Sympos. Pure Math.}, pages 613--647. Amer. Math. Soc., Providence, RI, 2007.

\bibitem[JZ19]{JZ19}
Svetlana Jitomirskaya and Xiaowen Zhu.
\newblock Large deviations of the {L}yapunov exponent and localization for the 1{D} {A}nderson model.
\newblock {\em Comm. Math. Phys.}, 370(1):311--324, 2019.

\bibitem[Kir08]{Kir08}
Werner Kirsch.
\newblock An invitation to random {S}chr\"{o}dinger operators.
\newblock In {\em Random {S}chr\"{o}dinger operators}, volume~25 of {\em Panor. Synth\`eses}, pages 1--119. Soc. Math. France, Paris, 2008.
\newblock With an appendix by Fr\'{e}d\'{e}ric Klopp.

\bibitem[Kle98]{Kle98}
Abel Klein.
\newblock Extended states in the {A}nderson model on the {B}ethe lattice.
\newblock {\em Adv. Math.}, 133(1):163--184, 1998.

\bibitem[Kry24]{Kry24}
Stanislav Krymskii.
\newblock On the lowest possible dimension of supports of solutions to the discrete {Schr\"odinger} equation.
\newblock {\em arXiv preprint arXiv:2401.02800}, 2024.

\bibitem[KS80]{kunz1980spectre}
Herv\'{e} Kunz and Bernard Souillard.
\newblock Sur le spectre des op{\'e}rateurs aux diff{\'e}rences finies al{\'e}atoires.
\newblock {\em Comm. Math. Phys.}, 78(2):201--246, 1980.

\bibitem[Li21]{Li21}
Linjun Li.
\newblock On the {M}anhattan pinball problem.
\newblock {\em Electron. Commun. Probab.}, 26:Paper No. 25, 11, 2021.

\bibitem[Li22]{Li22}
Linjun Li.
\newblock Anderson-{B}ernoulli localization at large disorder on the 2{D} lattice.
\newblock {\em Comm. Math. Phys.}, 393(1):151--214, 2022.

\bibitem[Li26a]{Li26}
Linjun Li.
\newblock Discrete unique continuation on simplex.
\newblock {\em arXiv preprint arXiv:2608.02707}, 2026.

\bibitem[Li26b]{Li26support}
Linjun Li.
\newblock On support cardinality for the discrete {Schr\"odinger} equation.
\newblock {\em Lett. Math. Phys.}, 116(4):90, 2026.

\bibitem[Li26c]{Li26Mirror}
Linjun Li.
\newblock Polynomial bound for the localization length of the {L}orentz mirror model on the {1D} cylinder.
\newblock {\em Electron. Commun. Probab.}, 31:1--14, 2026.

\bibitem[LL26]{LL26}
Suhan Liu and Patrick Lopatto.
\newblock Mobility edge for the {Anderson} model on random regular graphs.
\newblock {\em arXiv preprint arXiv:2603.14230}, 2026.

\bibitem[LSZ25]{LSZ25}
Shihe Liu, Yunfeng Shi, and Zhifei Zhang.
\newblock On localization for the alloy-type {A}nderson-{B}ernoulli model with long-range hopping.
\newblock {\em arXiv preprint arXiv:2508.12714}, 2025.

\bibitem[LSZ26a]{LSZ26}
Shihe Liu, Yunfeng Shi, and Zhifei Zhang.
\newblock Anderson localization for the hierarchical {A}nderson-{B}ernoulli model on $\mathbb{Z}^d$.
\newblock {\em arXiv preprint arXiv:2604.18989}, 2026.

\bibitem[LSZ26b]{LSZ26quc}
Shihe Liu, Yunfeng Shi, and Zhifei Zhang.
\newblock Quantitative unique continuation on simplex and $\mathbb{Z}^n$.
\newblock {\em arXiv preprint arXiv:2608.11602}, 2026.

\bibitem[LSZ26c]{LSZ26b}
Shihe Liu, Yunfeng Shi, and Zhifei Zhang.
\newblock Localization and unique continuation for the {Anderson-Bernoulli} model with long-range hopping on $\mathbb{Z}$.
\newblock {\em arXiv preprint arXiv:2607.03472}, 2026.

\bibitem[LZ22]{LZ22}
Linjun Li and Lingfu Zhang.
\newblock Anderson-{B}ernoulli localization on the three-dimensional lattice and discrete unique continuation principle.
\newblock {\em Duke Math. J.}, 171(2):327--415, 2022.

\bibitem[RZ25]{RZ}
Nishant Rangamani and Xiaowen Zhu.
\newblock Dynamical localization for the singular {Anderson} model in {$\mathbb{Z}^d$}.
\newblock {\em J. Math. Phys.}, 66(2):023505, 2025.

\bibitem[Sch09]{Sch09}
Jeffrey Schenker.
\newblock Eigenvector localization for random band matrices with power law band width.
\newblock {\em Comm. Math. Phys.}, 290(3):1065--1097, 2009.

\bibitem[Shc20]{Shc20}
Tatyana Shcherbina.
\newblock Characteristic polynomials for random band matrices near the threshold.
\newblock {\em J. Stat. Phys.}, 179:920--944, 2020.

\bibitem[SS17]{shcherbina2017characteristic}
Mariya Shcherbina and Tatyana Shcherbina.
\newblock Characteristic polynomials for {1D} random band matrices from the localization side.
\newblock {\em Comm. Math. Phys.}, 351(3):1009--1044, 2017.

\bibitem[SS21]{shcherbina2019universality}
Mariya Shcherbina and Tatyana Shcherbina.
\newblock Universality for {1d} random band matrices.
\newblock {\em Comm. Math. Phys.}, 385(2):667--716, 2021.

\bibitem[XYYY24]{XYYY24}
Changji Xu, Fan Yang, Horng-Tzer Yau, and Jun Yin.
\newblock Bulk universality and quantum unique ergodicity for random band matrices in high dimensions.
\newblock {\em Ann. Probab.}, 52(3):765--837, 2024.

\bibitem[YY25]{YY25}
Horng-Tzer Yau and Jun Yin.
\newblock Delocalization of one-dimensional random band matrices.
\newblock {\em arXiv preprint arXiv:2501.01718}, 2025.


\end{thebibliography}
\end{document}